\documentclass[11pt,letterpaper]{article}

\input{preamble}

\title{Robust and Leakage-Resilient\texorpdfstring{\\}{ } Device-Independent Oblivious Transfer\texorpdfstring{\\}{ } in MiniQCrypt}
\newif\ifanonymous
\anonymousfalse %
\ifanonymous
  \author{}
\else
  \author{%
    Rishabh Batra$^1$\thanks{\nolinkurl{rishabh.batra@epfl.ch}} \quad
    Zhili Chen$^2$\thanks{\nolinkurl{chen.zhili@u.nus.edu}} \quad
    Rahul Jain$^{2,3,4}$\thanks{\nolinkurl{dcsrahul@nus.edu.sg}} \quad
    YaoNan Zhang$^2$\thanks{\nolinkurl{zhang.yaonan@u.nus.edu}}\\[0.5em]
    \normalsize $^1$EPFL\\
    \normalsize $^2$Centre for Quantum Technologies, Singapore\\
    \normalsize $^3$Department of Computer Science, National University of Singapore\\
    \normalsize $^4$MajuLab, UMI 3654, Singapore}
\fi
\date{}

\begin{document}
\raggedbottom
\setlength{\emergencystretch}{2em}
\hypersetup{pdftitle={Robust and Leakage-Resilient Device-Independent Oblivious Transfer in MiniQCrypt}}
\ifanonymous
  \hypersetup{pdfauthor={}}
\else
  \hypersetup{pdfauthor={Rishabh Batra, Zhili Chen, Rahul Jain, YaoNan Zhang}}
\fi
\setlength{\textfloatsep}{10pt plus 2pt minus 2pt}
\setlength{\floatsep}{9pt plus 2pt minus 2pt}
\setlength{\intextsep}{10pt plus 2pt minus 2pt}
\renewcommand{\topfraction}{0.9}
\renewcommand{\bottomfraction}{0.9}
\renewcommand{\textfraction}{0.08}
\renewcommand{\floatpagefraction}{0.85}
\setcounter{topnumber}{3}
\setcounter{bottomnumber}{3}
\setcounter{totalnumber}{5}
\makeatletter
\setlength{\@fptop}{0pt}
\setlength{\@fpsep}{12pt}
\setlength{\@fpbot}{0pt plus 1fil}
\makeatother
\pagestyle{plain}

\hypersetup{pageanchor=false}
\begin{titlepage}
\maketitle
\thispagestyle{empty}

\begin{abstract}
Assuming post-quantum one-way functions, we construct device-independent oblivious transfer and bit commitment that tolerate independent honest faults at a constant rate and adaptive, bidirectional quantum leakage between the laboratories at a constant rate. Honest parties use trusted classical computation to operate untrusted quantum devices, which may share arbitrary entanglement and perform joint measurements without assumptions of independence or identical distribution. Security is simulation-based against quantum polynomial-time adversaries and composes sequentially with efficient simulators. For security parameter $\lambda$ and a prescribed total leakage budget $\ell$, one complete one-bit oblivious transfer uses $\Theta(\ell+\log^2\lambda)$ device coordinates in at most three single-use device windows, including preprocessing, and polynomial classical computation and communication. The leakage budget is shared across these windows. The commitment has efficient simulators against both parties and yields device-independent coin tossing with abort; each uses one window. Completeness of oblivious transfer gives device-independent secure computation with abort for every efficiently computable classical functionality on a fixed number of parties, against static corruption of any proper subset, provided each reached call to oblivious transfer satisfies the device and budget conditions.

The construction combines short selectors with code-based choice masking, a global audit, and error correction across field symbols. The commitment extractor changes a classical code coset while leaving the device execution and its leakage unchanged. Sender security uses a global Fourier analysis and a selector-aware Magic Square repetition bound derived from the two-round threshold theorem of Kundu and Tan (2023).

\end{abstract}
\end{titlepage}
\hypersetup{pageanchor=true}

\pagenumbering{roman}
\setcounter{tocdepth}{1}
\tableofcontents
\clearpage
\pagenumbering{arabic}

\section{Introduction}
\label{sec:intro}

Bit commitment (BC) and oblivious transfer (OT) are foundational primitives of two-party cryptography. A commitment lets Alice fix a bit while hiding it until a later opening. Kilian showed that OT is complete for secure computation~\cite{Killian_OT_implies_everything}: in one-out-of-two OT, Alice holds $(S_0,S_1)$, Bob holds a choice bit $D$, and Bob learns $S_D$ and nothing about $S_{1-D}$, while Alice learns nothing about $D$. Classically, OT is not known to follow from one-way functions, and Impagliazzo and Rudich rule out fully black-box constructions~\cite{IR89}. With quantum communication, it does follow: Grilo, Lin, Song, and Vaikuntanathan~\cite{OT1} and Bartusek, Coladangelo, Khurana, and Ma~\cite{OT2} independently construct simulation-secure OT from post-quantum one-way functions (qOWFs). This places OT in MiniQCrypt, the world in which qOWFs exist and quantum computation and communication are available. Their security proofs require trusted quantum hardware: the sender prepares BB84 states exactly, and the receiver performs the prescribed measurements.

\paragraph{Device independence.}
Devices obtained from a potentially colluding vendor need not implement the advertised quantum operations. A proof that assumes exact state preparation or measurement need not apply to such devices. In the device-independent (DI) model, security instead uses only observed correlations, here those of the Magic Square game of Mermin and Peres~\cite{Mermin_Magic_Sqaure,Peres_Magic_Square}. Honest users have trusted classical computation, communication, and private storage; they send prescribed classical questions to untrusted devices and receive classical answers. All quantum operations are performed by the untrusted devices, which may share arbitrary prior entanglement and need not act independently or identically across coordinates (non-IID). Unlike trusted-quantum protocols, the honest implementation is constrained only through this classical interface. In our protocols the two devices play $L$ Magic Square games in parallel, one per \emph{coordinate}, in a single use that we call the \emph{device window}.

We study two further requirements. Robustness means tolerating faults in the honest device implementation; leakage resilience means preserving security under bounded communication between the laboratories. Both are needed when the devices are faulty and imperfectly isolated.

\paragraph{The question.}
Device independence has so far been reached for OT only under further restrictions: earlier DI protocols assume bounded quantum storage (Broadbent and Yuen~\cite{Broadbent}; Batra, Chakraborty, Jain, and Kapshikar~\cite{BCJK25}), restrict the attack to be sequential (Kaniewski and Wehner~\cite{KW16}), or achieve constant cheating bounds (Kundu, Sikora, and Tan~\cite{KST22}); to our knowledge, none obtains simulation-secure OT from one-way functions alone. We therefore ask:
\begin{quote}
\itshape
Do post-quantum one-way functions suffice for oblivious transfer, and hence for secure computation, when the quantum devices are untrusted, faulty, and imperfectly isolated?
\end{quote}

\paragraph{The answer.}
We construct DI OT and DI bit commitment with simulation security against quantum polynomial-time (QPT) adversaries, efficient simulators, and sequential composition. Both tolerate constant rates of independent honest faults and physical leakage. The complete OT uses at most three device windows, including preprocessing, and polylogarithmically many device coordinates when the leakage budget is polylogarithmic. Classical communication remains polynomial-round. Under the per-call conditions stated below, OT completeness gives secure computation for every efficiently computable classical functionality on a fixed number of parties. The commitment has simulators against both parties and gives DI coin tossing; each uses one device window.

\paragraph{Robustness and leakage resilience.}
The commitment, the OT protocol, and their composition must all preserve these two properties.

\emph{Robustness.} Honest manufacturing, preparation, or measurement faults violate some Magic Square constraints. A protocol that certifies perfect correlations aborts on real devices, and a noiseless-only commitment would defeat the robustness of the OT built on it. Our protocols tolerate a fixed fraction of such faults, in the commitment and in the OT alike. Independent honest faults govern correctness only; security holds against non-IID devices.

\emph{Leakage resilience.} DI security rests on the devices in the two laboratories not communicating, and a colluding vendor's devices may communicate through hidden side channels. Charles Bennett contrasts DI-QKD with measurement-device-independent QKD, where the untrusted device stays outside the users' enclosures: in DI-QKD ``they must trust Eve not to have deliberately created side channels'' from the untrusted devices to herself~\cite[Box~VI.B.2]{XMZLP20}. We allow such a channel to be deliberately built into the devices, subject to a bound of $\ell$ qubits of adaptive, bidirectional communication. The model below specifies the window and charges every use of the channel. Let $L_{\rm total}$ be the total number of device coordinates in one complete OT. For fixed constants $c_0,\nu_0>0$, a complete invocation tolerates total communication at most $c_0L_{\rm total}$ qubits and independent honest coordinate-fault probability at most $\nu_0$. The communication budget is shared across the three single-use batches, rather than renewed at each window (\Cref{sec:model-assumptions}).

\emph{Simultaneous guarantees.} Both constants hold in the same protocol, with arbitrary joint receiver measurements (\Cref{tab:intro-result-summary}). They are fixed by the proof and need not be large enough for a practical implementation.

\subsection{The model, and what each assumption is for}
\label{sec:model-assumptions}

We assume qOWFs and a possibly nonuniform QPT adversary with polynomial-size quantum advice. Corruption is static. In a two-party protocol, at most one party is corrupt; in the multiparty corollary, the corrupted set is an arbitrary proper subset, which acts as one adversary. The vendor of all devices may collude with it. The devices may share arbitrary entanglement, retain quantum memory, and perform joint measurements, subject to the QPT model and the conditions below. At each call, honest trusted storage is isolated, fresh honest coins are independent of the device state, and only the current batch is accessible. These conditions also hold at every reached boundary in the simulation hybrids (\Cref{def:admissible-extension,def:ot-call-boundary}).

\paragraph{Trusted classical control.}
The trusted classical controller of the DI model has no leakage interface. Without it, an honest party has nothing with which to test a device. Preprocessing masks and their opening witnesses remain private behind the prescribed commitment and proof operations; the caller cannot read unused masks.

\paragraph{Bounded communication between the laboratories.}
A batch is used in parallel, in a single window, and during batch $j$'s window the two laboratories exchange at most $\ell_j$ qubits (\Cref{def:ot-admissible-leakage}). A global budget $\ell$ requires $\sum_j\ell_j\le\ell$. The parallel use is what makes the bound meaningful: every additional window in which the devices are active is an additional opportunity to leak, so each batch is used once. Some bound is unavoidable: Magic Square correlations certify nothing if the devices can exchange their questions, and a receiver's device that reports its questions reveals the receiver's choice bit. The standard DI assumption is $\ell=0$; we make the bound explicit and charge it in the entropy analysis. Isolation can be replaced by a computational assumption: the compiler of Kalai, Lombardi, Vaikuntanathan, and Yang turns a nonlocal game into a single-prover protocol whose soundness rests on quantum homomorphic encryption, and its quantum value is now bounded for every two-player game~\cite{KLVY23,NZ23,KMPSW24}. That route needs encryption schemes based on structured assumptions such as learning with errors, and a prover that runs the encrypted game. We keep isolation, with an explicit leakage budget, and assume only one-way functions.

\paragraph{Fresh, isolated batches.}
Every invocation of \Cref{prot:di-bc} or \OTprotocolparts{} uses its own device batch, once, and a batch exchanges no information with a retired or pending one, also within one laboratory (\Cref{def:compiled-device-calls}). This excludes the memory attacks of Barrett, Colbeck, and Kent~\cite{DIQKD_3}, in which a reused device releases earlier private data through later public outputs. Devices are distributed before question randomization; the window count concerns their subsequent use. Each batch shares one leakage budget across its blocks, and classical subprotocols run outside its window.

\paragraph{Published device outputs also disclose information.}
Let $t_A$ be the number of openly tested coordinates. Bob's published test answers and the resulting accept/abort decision carry at most $3t_A+1$ bits from his device to Alice; Alice's audit openings similarly disclose information to Bob. These prescribed classical disclosures occur outside the device window and are charged separately from its $\ell$-qubit communication budget: against the selector entropy on Alice's side and the pad entropy on Bob's side (\Cref{sec:Bob-to-Alice-leakage,tab:game-correspondence}). The hidden test set is revealed only after the answer records are fixed, so the device must produce the tested and untested answers before it learns which coordinates will be tested. Bob's questions are uniform coordinatewise; he chooses the question class and opens its selector afterwards.

\paragraph{Joint measurements.}
The honest receiver's untrusted device may measure jointly across its entire batch. This permits correlations among all its answers; the security proof must account for the information in the published test answers even under such measurements.

\paragraph{The ideal experiment keeps the honest party's device share.}
The honest party's device is supplied by the adversary and may be entangled with the adversary's workspace. No simulator can fabricate it. The ideal experiment therefore gives the simulator the same share, through the interface that the honest party has: the prescribed classical questions, once, with no reset and no copy (\Cref{def:comp-aug-sec}). The simulator receives no honest input. Security then says that whatever the adversary obtains from the real execution, it could obtain from the ideal functionality together with a party that operates the honest device as prescribed and holds no private input. Because every batch is used once and then retired, this notion composes sequentially (\Cref{thm:sequential-composition}).

\paragraph{Honest faults.}
The honest coordinate implementations are independent, but need not be identical. Their loss probability is bounded for every question pair. Independence is used for honest correctness only; security permits arbitrary entangled attacks.

\subsection{Our results}
\label{sec:results}

In this model, each DI reduction uses the same parameter family. Let $\lambda$ be the security parameter, let $\ell$ be a prescribed polynomially bounded global leakage budget, and put $s_\lambda=\lceil\log_2^2(\lambda+2)\rceil$. The outer batch contains $m=\Theta(\ell+s_\lambda)$ field symbols, each using a fixed number of device coordinates. The two preprocessing batches have sizes within fixed factors of the outer size. All security and correctness errors are $\negl(\lambda)$.

\begin{table}[!htbp]
\centering
\caption{Guarantees for one complete one-bit DI OT, including both preprocessings. The constants are independent of $\lambda$ and $\ell$; the leakage budget is global across all three batches.}
\label{tab:intro-result-summary}
\normalsize
\setlength{\tabcolsep}{5pt}
\renewcommand{\arraystretch}{1.1}
\begin{tabular}{>{\raggedright\arraybackslash}p{0.30\textwidth}
                >{\raggedright\arraybackslash}p{0.61\textwidth}}
\hline
\textbf{Property} & \textbf{Guarantee}\\
\hline
Physical leakage & Adaptive bidirectional communication;
$\ell\le c_0L_{\rm total}$ qubits for a fixed $c_0>0$\\
Honest-device robustness & Independent coordinate-fault probability
at most a fixed $\nu_0>0$\\
Receiver measurements & Arbitrary joint measurements\\
Device coordinates & $L_{\rm total}=\Theta(\ell+\log^2\lambda)$\\
Device windows & At most three, including preprocessing\\
Classical resources & Polynomial computation and communication\\
\hline
\end{tabular}
\end{table}

\begin{theorem}[Device-independent OT from qOWFs, informal; see \Cref{thm:qowf-di-ot}]
\label{thm:intro-qowf}
Under qOWFs and the conditions in \Cref{def:admissible-extension,def:compiled-device-calls}, let $\Pi_{\rm OT}$ be \Cref{prot:three-window-ot}, with \Cref{tab:intro-result-summary}. Then $\Pi_{\rm OT}$ is sequentially composable for either corruption, and satisfies
\[
 \begin{gathered}
 \forall\mathcal A\in\QPT\quad\exists\Sim\in\QPT
 \quad\forall\eta^{\Rref}_{\rm adm}:\\
 \mathsf{Real}_{\Pi_{\rm OT},\mathcal A}(\eta^{\Rref}_{\rm adm})
 \approx_c
 \mathsf{Ideal}_{\mathsf{OT},\Sim}(\eta^{\Rref}_{\rm adm}).
 \end{gathered}
\]
Here $\eta^{\Rref}_{\rm adm}$ ranges over admissible inputs and extensions; $\varepsilon_{\rm corr}=\negl(\lambda)$, and $R_{\rm DI}\le3$, including preprocessing. Total device use is $\Theta(\ell+\log^2\lambda)$; total classical computation and communication are $\poly(\lambda)$.
\end{theorem}

Here $\approx_c$ includes the adversary's quantum side information. For OT and the bit commitment of \Cref{thm:intro-bc}, one QPT simulator per adversary works for every admissible input. A device window is one parallel input--output interval; it does not count classical rounds or bound the adaptive leakage messages inside it.

OT completeness extends this construction to multiparty computation (MPC).

\begin{corollary}[Device-independent multiparty computation, informal; see \Cref{cor:di-mpc}]
\label{cor:intro-mpc}
Assume qOWFs and private authenticated classical channels, and impose the invocation conditions of \Cref{def:ot-call-boundary,def:compiled-device-calls}, with the corresponding correctness assumptions for honest--honest calls.
Every classical PPT functionality on a fixed number of parties, with public input/output lengths and corrupt-output-first, adversarially suppressible honest delivery, admits a DI protocol with simulation security with abort against static corruption of any proper subset of the parties.
The protocol uses polynomially many sequential pairwise DI OT calls and polynomial total resources, with each admissible call inheriting \Cref{tab:intro-result-summary}.
\end{corollary}

The three-window bound and constant leakage rate apply per OT call. The proof in \Cref{app:mpc-proof} uses the classical OT-to-MPC reduction of Ishai, Prabhakaran, and Sahai~\cite{IPS08} and Unruh's quantum lifting theorem~\cite{Unruh_2}. Only trusted classical operations occur outside DI OT.

We also obtain the following commitment and coin-tossing guarantees. \Cref{sec:di-bc-coin} gives the bit-commitment and coin-tossing constructions and their security proofs.

\begin{theorem}[Device-independent commitment and coin tossing, informal; see \Cref{thm:di-vector-commitment,thm:di-bc-two-sided,cor:di-coin}]
\label{thm:intro-bc}
Under qOWFs and the conditions in \Cref{def:admissible-extension,def:compiled-device-calls}, with \Cref{tab:intro-result-summary}: for $q=\polylog(\lambda)$, \Cref{prot:di-bc} is a one-window DI commitment to a $q$-bit string, with an efficient simulator that extracts the string at commitment completion and with joint computational hiding against a corrupt receiver. \Cref{prot:di-bc-certified} is a DI bit commitment with efficient simulators against both parties, and \Cref{prot:di-coin} realizes coin tossing with abort. Each of the latter two protocols uses one device window. The string commitment uses $O(q+\ell+\log^2\lambda)$ coordinates; all three have polylogarithmic device use when their input lengths and leakage budgets are polylogarithmic.
\end{theorem}

\paragraph{Composition and efficiency.}
\Cref{def:comp-aug-sec,thm:sequential-composition} permit polynomially many adaptive, non-overlapping calls, each replaced with its caller, continuation, corrupt state, honest input--output interface, reference, and abort branches retained. The simulations also preserve pending commitment records across delayed openings. An outer OT batch of $L$ coordinates stores $O(L)$ bits and performs two prescribed polynomial-time predicate checks in addition to projections. The results of \Cref{thm:di-vector-commitment,thm:masked-cp} give polynomial total computation and communication, with polylogarithmic device use per OT for polylogarithmic leakage budgets.

\paragraph{Significance.}
Under this model, parties can securely compute every efficient classical functionality using trusted classical computers and untrusted quantum devices, with efficient simulators, constant honest-fault tolerance, and constant physical-leakage rate. The computational assumption is qOWFs. The alternative routes differ in their hardware and computational assumptions. The trusted-quantum constructions also give simulation-based security from qOWFs but need exact state preparation and measurement. Bounded-storage DI OT needs no computational assumption but assumes a limit on the adversary's quantum memory, and its security proof uses that restriction. Compiled nonlocal games remove the isolation between the laboratories, but need structured assumptions such as learning with errors and a prover that runs the encrypted game. The price paid here is the model of \Cref{sec:model-assumptions}: isolated laboratories with an explicit leakage budget, a fresh single-use batch per elementary call, and the fixed constants for honest faults and physical leakage.

\subsection{Protocol and proof overview}
\label{sec:results-overview}

The construction has three ingredients: OT from commit-and-prove, commit-and-prove from a string commitment, and a device-independent implementation of that commitment. We describe them in this order, which is also the order of the technical sections.

\paragraph{Commit-and-prove.}
A commit-and-prove (CP) functionality $\ideal(\mathsf{CP})$ stores immutable strings; a receipt reveals only the record identifier, roles, and length. An owner may certify a predicate of an ordered tuple of its existing records. A partial opening publishes a substring and certifies its equality to the stored projection. Every later check uses the same records (\Cref{def:cp}).

\begin{figure}[!htbp]
\centering
\begin{tikzpicture}[
  box/.style={draw,rounded corners,minimum height=6.5mm,align=center,font=\footnotesize,inner xsep=3pt,inner ysep=2pt},
  ours/.style={box,thick},
  lbl/.style={align=center,font=\scriptsize},
  row/.style={font=\footnotesize,anchor=east},
  arr/.style={-{Latex[length=2.2mm]},thick},
  x=1cm,y=1cm]
 \node[row] at (-0.2,3.0) {\cite{OT1}};
 \node[box] (a1) at (0.9,3.0) {qOWFs};
 \node[box] (a2) at (3.35,3.0) {weak selective-\\opening BC};
 \node[box] (a3) at (5.9,3.0) {weak\\parallel OT};
 \node[box] (a4) at (8.05,3.0) {weak\\vCDS};
 \node[box] (a5) at (10.75,3.0) {simulated selective-\\opening BC};
 \node[box] (a6) at (13.3,3.0) {OT};
 \foreach \p/\q in {a1/a2,a2/a3,a3/a4,a4/a5,a5/a6} \draw[arr] (\p) -- (\q);
 \node[row] at (-0.2,1.5) {\cite{OT2}};
  \node[box] (b1) at (0.9,1.5) {qOWFs\\or EFI};
 \node[box] (b2) at (3.35,1.5) {statistically\\binding BC};
 \node[box] (b3) at (5.9,1.5) {equivocal\\BC};
 \node[box] (b4) at (8.05,1.5) {extractable\\BC};
 \node[box] (b5) at (10.75,1.5) {extractable and\\equivocal BC};
 \node[box] (b6) at (13.3,1.5) {OT};
 \foreach \p/\q in {b1/b2,b2/b3,b3/b4,b4/b5,b5/b6} \draw[arr] (\p) -- (\q);
 \node[row] at (-0.2,0) {this work};
 \node[ours] (c1) at (0.9,0) {qOWFs};
  \node[ours] (c2) at (3.10,0) {post-quantum\\equivocal BC\\[-2pt]{\scriptsize Section~\ref{sec:bc-step-one}}};
 \node[ours] (c3) at (5.8,0) {DI string\\commitment\\[-1pt]{\scriptsize Protocol~\ref{prot:di-bc}}\\[-2pt]{\scriptsize Section~\ref{sec:di-bc}}};
 \node[ours] (c4) at (8.70,0) {CP\\[-1pt]{\scriptsize Protocol~\ref{prot:masked-cp}}\\[-2pt]{\scriptsize Section~\ref{sec:bc-cp-compiler}}};
 \node[ours] (c5) at (11.0,0) {DI OT\\[-1pt]{\scriptsize Protocol~\ref{prot:di-ot}}\\[-2pt]{\scriptsize Section~\ref{sec:leakage-ot}}};
 \node[ours] (c6) at (13.45,0) {MPC\\[-1pt]{\scriptsize Section~\ref{sec:complete-construction}}};
 \draw[arr] (c1) -- (c2);
 \draw[arr] (c2) -- node[lbl,above=1pt]{DI} (c3);
 \draw[arr] (c3) -- node[lbl,above=1pt]{qOWFs} (c4);
 \draw[arr] (c4) -- node[lbl,above=1pt]{DI} (c5);
 \draw[arr] (c5) -- (c6);
\end{tikzpicture}
\caption{Reduction chains. The two earlier chains use trusted quantum communication and place no condition on devices; ``weak'' permits inefficient simulation, vCDS is verifiable conditional disclosure of secrets. EFI pairs are efficiently generated quantum states that are statistically far but computationally indistinguishable. They also suffice for the chain of~\cite{OT2}, using quantum communication \cite{brakerski2022computational}. Ours removes the trusted quantum operations; in exchange, its two DI arrows require fresh isolated batches with bounded leakage. Its boxes name the protocol and the section that proves each step. The middle arrow uses one string commitment per owner and only classical operations afterwards. CP also gives BC and coin tossing (\Cref{sec:bc-both-sides,sec:coin-toss}); OT completeness gives MPC.}
\label{fig:chain}
\end{figure}

The two DI arrows in \Cref{fig:chain} use the same device experiment, with code lengths chosen for their respective payloads. Their classical mode and opening phases differ as described below.

\paragraph{The OT protocol in words.}
Alice has bits $(S_0,S_1)$ and Bob has a choice $D$. Their devices play $L$ parallel Magic Square games: Alice's device returns a row, Bob's a column, and a winning pair agrees at the intersection. Outside a small hidden test set, the coordinates form $m$ field symbols. Each symbol consists of six \emph{microblocks}, each of a fixed size $N_0$. A private bit $\beta_{i,j}$ chooses one of two question classes in a microblock; its unaudited positions form a \emph{selector} $J_{i,j}$. The six bits represent a field element $\beta_i\in\mathbb F_{64}$.

The OT proof uses two consistency predicates. $\Label$ checks that every selected position has its claimed question label and that the public code mask is the sum of Bob's label vector and a codeword encoding one receiver-choice bit. $\Pay$ checks that Alice's payloads use the hashes of her committed table and are consistent with codewords encoding a single pair of sender inputs. Both predicates bind these claims to the same immutable records (\Cref{def:selector-consistency-proof,def:payload-consistency-proof}). The following numbers match \Cref{prot:di-ot}.
\begin{enumerate}[label=\textnormal{\arabic*.},leftmargin=1.8em,itemsep=1pt,topsep=3pt]
\item \emph{Alice} commits to a small hidden test set.
\item \emph{Bob} commits to uniformly random question offsets.
\item \emph{Alice} sends fresh offsets and queries her device.
\item \emph{Bob} queries his device with the sum of the two offsets. The physical window closes, and he commits to his full answer record.
\item \emph{Alice} opens the test set and sends a fresh random ordering that partitions the remaining positions into microblocks.
\item \emph{Bob} samples one audit position and one private label per microblock. If either candidate question class is too small or too large, he erases the whole symbol. He commits to all audits, labels, erasure flags, and selectors without revealing a size decision.
\item \emph{Alice} samples late independent anchor flags, replaces anchored records by a fixed row and question, and commits to her table.
\item \emph{Bob} opens the audits and his question-offset and answer records on the test set.
\item \emph{Alice} tests the sampled Magic Square correlations and, on success, opens her table on the audits.
\item \emph{Bob} checks one aggregate audit-loss threshold and a cap on erased symbols, then opens the erasure flags and selectors.
\item \emph{Alice} checks the selectors' positions and sizes.
\item \emph{Bob} samples a fresh diagonal scaling of the masking code and a codeword $R$ with secret $D$. He sends $M=\beta+R$ and certifies $\Label$ on his early and late records jointly.
\item \emph{Alice} opens her questions and anchor flags on the selectors. She sends two hashed, masked field payloads per active microblock and certifies their $\Pay$ relation with her table.
\item \emph{Bob} unmasks his selected payloads and decodes one codeword across the surviving symbols. He outputs its secret, with a fixed private fallback on decoding failure and no public decoding flag.
\item \emph{Alice} returns the empty output or the prescribed abort.
\end{enumerate}
The two offsets make Bob's physical question uniform while keeping it hidden from Alice. The audit uses one sample per microblock, with a single threshold for the whole batch. Error correction acts on the field symbols, so a small constant fraction of physical faults is correctable even though individual microblocks have fixed size.

\paragraph{Why the code payloads give OT.}
Write $\CodeB^\sigma$, $\CodeD$, and $\CodeA^\sigma$ for the masking, difference, and payload codes, where $\sigma$ is the public diagonal scaling. The linear map $\secmap_\sigma$ returns the secret of a payload codeword; $e_1,\ldots,e_6$ form a fixed binary basis of $\mathbb F_{64}$. The codes support coordinatewise multiplication: $\CodeB^\sigma*\CodeD\subseteq\CodeA^\sigma$. Alice samples $f\in\CodeA^\sigma$ with secret $S_0$ and $g\in\CodeD$ with secret $S_0+S_1$. Within symbol $i$ she shares $f_i+M_ig_i$ among the branch-zero payloads; switching branch $j$ adds $e_jg_i$. Bob's actual selected sum is therefore
\[
 H=f+(M+\beta)*g=f+R*g\in\CodeA^\sigma,
 \qquad \secmap_\sigma(H)=S_D.
\]
$\Pay$ enforces this identity for a cheating sender's uniquely determined effective input pair. Each payload is masked by a fresh matrix hash of the corresponding column bits of Alice's recorded rows.

Receiver security uses the entropy of the whole label vector after the selector disclosure. Exchanging question classes zero and one in an unaudited microblock and flipping its label preserves the selector. Physical leakage and the few public test answers reduce that entropy by a linear debit. A random diagonal scaling turns the zero-secret code into a small-bias family, so $M$ hides which secret $R$ encodes. To obtain full simulation, first use the audit to replace Bob's decoded output by the corresponding extracted sender input; only then apply choice hiding to the common continuation.

For sender security, the selected codeword $H$ depends only on $S_D$. The repetition and Fourier argument \hyperref[par:intro-core]{below} hides the remaining dependence on $S_{1-D}$.

\paragraph{The commitment protocol in words.}
The source commits to a string $u\in\mathbb F_2^q$, using $q$ secret evaluation positions in the same code family. Its formal protocol denotes the difference codeword $g$ by $D$. It uses classical inner commitments rather than CP. Its numbers are those of \Cref{prot:di-bc}. Here $\Mode$ is the commitment analogue of $\Label$: it checks the selected question classes against Bob's committed offsets and labels, and requires the public mask minus the witnessed label vector to be a codeword with secret $0^q$ (\Cref{blk:def:mode}). Steps~1--2 commit the test and offset records; Step~3 performs the one physical window and commits Bob's answers; Steps~4--8 fix the partition, selectors, and table and perform the test, audit, and selector checks.
\begin{enumerate}[label=\textnormal{\arabic*.},leftmargin=1.8em,itemsep=1pt,topsep=3pt,start=9]
\item \emph{Bob} sends a code mask with zero secret and proves the $\Mode$ relation in ordinary post-quantum zero knowledge.
\item \emph{Alice} sends the payloads above with $f$ of secret zero and $g$ of secret $u$, and opens the selected questions and anchors.
\item \emph{Bob} issues a receipt after the public checks and discards his private questions, answers, labels, mode witnesses, and device state.
\item \emph{Alice} opens $u$ and her table. Bob verifies the public payload equations and their unique extensions to the two codes.
\end{enumerate}
In ordinary mode the selected codeword has secret zero, so the joint payload argument hides $u$. Extraction must recover $u$ while preserving the continuing state and agreement with every native opening. The obstacles below explain how the same device experiment provides this guarantee.

\begin{figure}[!htbp]
\centering
\setlength{\tabcolsep}{5pt}
\renewcommand{\arraystretch}{1.12}
\begin{tabular}{>{\raggedright\arraybackslash}p{0.14\textwidth}
                >{\raggedright\arraybackslash}p{0.77\textwidth}}
\hline
\textbf{OT steps} & \textbf{Order of records and disclosures}\\
\hline
1--4 & The test mask and offsets are committed; the device window
closes and Bob's answers are committed before the test set opens.\\
5--6 & The partition, audits, full labels, erasure flags, and selectors
are fixed before Alice samples anchors.\\
7 & Alice's canonical table is fixed before any audit mask opens.\\
8--11 & The test, aggregate audit, and selector decisions precede
the fresh diagonal and code mask.\\
12--15 & $\Label$ binds the mask to the early records; $\Pay$ binds the
payloads to the table; decoding has no public success flag.\\
\hline
\end{tabular}
\caption{The timing used by the entropy and simulation proofs. The source commitment uses the same order through its $\Mode$ message.}
\label{fig:cuts}
\end{figure}

\paragraph{Implementing CP from a committed mask string.}
A \emph{mask bank} is a commitment to a uniform string $\mu\in\Ftwo^q$, with its private opening witness and an allocation counter. Write $\Ver_{\rm str}(\tau,u,w)=1$ for a valid native opening of its transcript $\tau$. The phases of \Cref{prot:masked-cp} are:
\begin{enumerate}[label=\textnormal{\arabic*.},leftmargin=1.8em,itemsep=1pt,topsep=3pt]
\item \emph{Alice} completes the string commitment to $\mu$ with Bob.
\item \emph{Alice} allocates an unused interval $I_a$ for a later message $M_a$ and sends $c_a=M_a\oplus\mu_{I_a}$ with its identifier.
\item \emph{Bob} checks and stores this record, then sends a receipt.
\item \emph{Alice} publishes $(\mathbf a,P,y)$ for an ordered tuple of existing records and proves that one native bank opening $(u,w)$ unmasks all its components to values satisfying $P$; she aborts if $P(M_{a_1}\|\cdots\|M_{a_r},y)=0$.
\item \emph{Bob} verifies the ordinary post-quantum zero-knowledge proof and sends its result. Selective openings use equality predicates.
\end{enumerate}
Only preprocessing uses devices. The bank is never natively opened; proof-internal commitments are ordinary classical commitments from qOWFs.

Against a corrupt owner, extraction fixes $u^*$ and hence $M_a^*=c_a\oplus u^*_{I_a}$. Agreement with all native openings makes inconsistent predicates false NP statements, so soundness suffices. Against a corrupt verifier, first simulate proofs of the true statements, then replace the bank commitment by a commitment to zero using joint hiding. Unused masks remain independent of the preceding quantum view: the simulator sends uniform ciphertexts and needs only ideal CP's public requests and outcomes for the simulated proofs.

The complete OT creates Alice's bank, then Bob's, then uses the outer batch. Longer source codes supply the required mask capacity; audit thresholds use each batch's own decoding radius. A one-bit record with a full certified projection gives BC; commit-and-reveal gives coin tossing with abort. Both use one bank (\Cref{sec:bc-both-sides,sec:coin-toss}).

\paragraph{Three obstacles.}
Three constraints govern the simulation and entropy arguments.

\emph{Extraction without changing the device execution.} The simulator may only ask the prescribed device questions once; changing the physical execution could also change its leakage. Our extractor changes the mask codeword's secret from zero to the all-one vector and privately decodes $u$ from its genuine device answers. Post-quantum equivocal commitments provide an alternative $\Mode$ witness, and choice masking hides the change. The physical questions, answers, and leakage are unchanged. A single global audit controls every native opening in this same extraction experiment (\Cref{sec:bc-step-two}).

\emph{Accounting for leakage.} The transcript-guessing bound (\Cref{lem:interactive-transcript-guessing}) charges each leaked qubit two bits. The selector vector has linear entropy in the number of fixed-size microblocks, and the global payload proof certifies entropy linear in the full batch. Code masking and error correction across field symbols leave a constant margin for both leakage and honest noise (\Cref{sec:unified-parameters}).

\emph{Entropy against non-IID devices.} Sender security needs min-entropy in Alice's device answers, on a set of coordinates that a cheating receiver chooses, when the devices are arbitrary and entangled, and the receiver may measure jointly and late. We incorporate this choice into the repetition game and sample independent erasures after the selector is fixed. This permits the anchored threshold theorem to yield the bound described below.

\paragraph{Repetition with a receiver-chosen selector.}
\label{par:intro-core}

Sender security of OT and hiding of the commitment both reduce to one statement about $n$ parallel Magic Square games. Charlie receives row questions and returns rows. Dave receives column questions and fixes his answers $B$ together with a selector $J\subseteq[n]$ of his choice. Only then is each coordinate erased, independently with a fixed probability $\delta_*$, and Dave receives Charlie's questions on the coordinates that were not erased. Dave must output Charlie's rows on $J$; an erased coordinate counts as correct. Let $L$ be the number of unerased coordinates on which the Magic Square condition fails, and $Q$ the number of selected coordinates that are erased or on which Dave both guesses the row correctly and wins the Magic Square game.

\begin{theorem}[Selector-aware repetition, informal version of \Cref{thm:restricted-order-repetition}]
For every selector density $0<\rho\le1$ and every $\tfrac{1+\delta_*}{2}\rho<\gamma<\rho$ there is $\beta>0$ such that $\Pr[L\leq\beta n\ \wedge\ |J|\leq\rho n\ \wedge\ Q\geq\gamma n] \leq2^{-\Omega(n)}$, for every strategy with arbitrary entanglement and quantum memory, and with no lower bound on $|J|$.
\end{theorem}

\noindent The bound has to be parallel. An honest party may use its device one coordinate at a time, but a corrupt receiver holds its device and may measure all coordinates jointly and as late as it likes, so sequential analyses such as entropy accumulation do not describe it. For one copy, Fu and Miller~\cite{FM-BRE} show that a Dave who wins the Magic Square game with probability $1-\varepsilon$ guesses a second bit of Charlie's row with probability at most $\tfrac12+O(\sqrt\varepsilon)$; this is why the game is Magic Square and not CHSH, in which Dave's answer and both questions determine Charlie's winning answer and leave nothing to hide (\Cref{sec:selector-repetition-input}). To handle a selector, we mix the Magic Square test with a row-guess test in which an unselected coordinate is completed by a fair coin. The mixture is a two-round product-anchored game of value below one, and the threshold theorem of Kundu and Tan~\cite{kundu_tan} applies to it. The erasures play the role of the anchor; because they are sampled after $(B,J)$ is frozen, they are independent of the selector, and concentration needs no conditioning on Dave's wins.

For sender security, the code's dual distance forces every input-dependent Fourier character of the payload to involve a constant fraction of the symbols. Their selectors form an early-frozen set of constant density in the whole device batch. Guessing the complementary raw bits on that set reconstructs Alice's rows using Bob's committed answers and row parity. The repetition theorem forbids this, with an exponent linear in the whole batch. \Cref{tab:game-correspondence} accounts for subsequent disclosures before the Fourier bound is summed. There is one common subnormalized good branch for all characters, so the repetition error is charged once.

\subsection{Previous work and technical contributions}
\label{sec:previous-work}
\label{sec:comparison-novelty-ot}

Our construction draws on two lines of work: quantum constructions of OT from qOWFs and device-independent entropy certification. Combining these approaches in our setting requires entropy bounds that allow arbitrary joint measurements and receiver-chosen selectors, together with extraction that preserves the prescribed device execution. The comparisons below explain how these requirements shape the reduction and its commit-and-prove interface, and how simulation security, robustness, and leakage resilience are maintained throughout.

\paragraph{Assumptions and reduction chains.}
Grilo, Lin, Song, and Vaikuntanathan~\cite{OT1} and Bartusek, Coladangelo, Khurana, and Ma~\cite{OT2} obtain OT from qOWFs with efficient simulation and sequential composition. Their constructions use trusted preparation and measurement of BB84 states. Earlier DI protocols achieve security under different restrictions, such as bounded quantum storage or sequential attacks (\Cref{tab:di-comparison}); Kundu, Sikora, and Tan use Magic Square but obtain one-shot constant cheating bounds~\cite[Sections~1.3--1.5]{KST22}. Our construction gives DI OT and hence secure computation from qOWFs, with at most three device windows for one complete OT, including preprocessing. These guarantees rely on the computational and device assumptions of \Cref{sec:model-assumptions}, so the earlier protocols and ours apply in different regimes.

\begin{table}[!htbp]
\centering
\caption{Device-independent two-party protocols: assumptions beyond isolated laboratories and the corresponding security guarantees.}
\label{tab:di-comparison}
\normalsize
\setlength{\tabcolsep}{4pt}
\renewcommand{\arraystretch}{1.1}
\begin{tabular}{
 >{\raggedright\arraybackslash}p{0.13\textwidth}
 >{\raggedright\arraybackslash}p{0.19\textwidth}
 >{\raggedright\arraybackslash}p{0.29\textwidth}
 >{\raggedright\arraybackslash}p{0.29\textwidth}}
\hline
\textbf{Work} & \textbf{Task} & \textbf{Security rests on} &
\textbf{Guarantee}\\
\hline
\cite{SAC11,ACK14} & commitment, coin flipping & nothing beyond
isolation & constant cheating probability and bias\\
\cite{KST22} & XOR OT, one copy & nothing beyond isolation & constant
cheating bounds\\
\cite{KW16} & two-party tasks & attacks restricted to be sequential &
security against sequential attacks\\
\cite{Broadbent,BCJK25} & OT & bounded quantum storage, in
\cite{Broadbent} together with computational assumptions and IID devices & OT in that
model; \cite{BCJK25} is robust and composable\\
This work & OT, commitment, coin tossing, MPC & qOWFs; isolated
single-use batches; leakage at most $\ell$ & simulation-based,
sequentially composable, negligible error, robust\\
\hline
\end{tabular}
\end{table}

The reduction chains in \Cref{fig:chain} locate the main difficulty in the commitment used by OT: it needs extraction and simulation guarantees beyond ordinary binding and hiding (\Cref{def:vector-commitment}).\footnote{The Impagliazzo--Rudich barrier concerns fully black-box constructions of key agreement from one-way functions~\cite{IR89}; it does not separate ideal BC from OT.} In the trusted quantum-communication model, Bartusek, Coladangelo, Khurana, and Ma can build the needed guarantees from black-box access to a statistically binding, quantum-computationally hiding commitment~\cite[Corollary~6.5]{OT2}. EFI pairs therefore suffice~\cite{EFI,brakerski2022computational}. Brakerski, Canetti, and Qian establish equivalences between EFI pairs and cryptographic primitives, while Ananth, Qian, and Yuen obtain cryptographic tasks from pseudorandom quantum states~\cite{brakerski2022computational,OT3}. These results permit trusted quantum operations or communication. An EFI-based commitment sends quantum states between the parties, whereas our inner commitment must have classical messages; for this primitive, qOWFs are the minimal known assumption.

The two trusted-quantum constructions obtain these stronger commitments by different routes. Grilo, Lin, Song, and Vaikuntanathan~\cite{OT1} first combine Naor commitments with post-quantum ZK; the resulting BB84 OT has an inefficient receiver simulator. Parallel composition, garbling, commitments, and ZK yield weak verifiable conditional disclosure of secrets (vCDS), which releases a secret on a valid NP witness and certifies input--output consistency. A simulated trapdoor then lets an efficient CDS receiver recover the committed secret; consistency proofs bind it to later openings. The resulting selective-opening commitment, which permits opening chosen positions of a committed string, is substituted back into OT for efficient simulation on both sides \cite[Sections~3--4, Figure~6]{OT1}.

Bartusek, Coladangelo, Khurana, and Ma~\cite{OT2} first make Naor's commitment equivocal. Trusted BB84 states and delayed-measurement simulation then supply extraction; applying the equivocation compiler again preserves extraction and gives the commitments used in the final BB84 OT \cite[Sections~2.2--2.5, Theorems~4.1, 5.1, and~6.2]{OT2}. Our chain retains classical commitments and proofs from qOWFs, and our NP relations, like the proof machinery of~\cite{OT1}, use classical algorithms inside NP statements. Replacing the trusted BB84 step requires both entropy against arbitrary joint device measurements and an extractor that preserves the prescribed physical execution.

\paragraph{Entropy certification and selector-aware repetition.}
For the entropy requirement, earlier DI OT protocols provide a common starting point: test the devices to certify entropy in a raw answer string, then use privacy amplification to obtain a mask. Kaniewski and Wehner use CHSH, extending a single-round uncertainty tradeoff by min-entropy additivity for memoryless devices and a martingale argument for sequential attacks \cite[Sections~II.2.1 and~II.3.1--II.3.2]{KW16}. Broadbent and Yuen use Learning-with-Errors-based single-device self-testing, an IID device model, an entropic uncertainty relation, and two-universal hashing~\cite{Broadbent}.

Batra, Chakraborty, Jain, and Kapshikar give the closest technical antecedent to our entropy argument: they combine Magic Square, parallel repetition, and a seeded extractor \cite[Section~1 (Proof overview), Sections~4.3 and~5]{BCJK25}. The parallel DI key distribution analyses of Jain, Miller, and Shi and of Jain and Kundu provide related repetition and entropy techniques \cite{JMS_DIQKD,Jain_Kundu_DIQKD}. For our sender-security proof, the additional issue is that Bob chooses his answer record and selector jointly. Applying a bound for a fixed set of coordinates therefore requires an argument that accounts for this choice.

Our selector-aware bound includes Bob's choice in the game itself (see the \hyperref[par:intro-core]{proof overview}). The threshold theorem of Kundu and Tan~\cite{kundu_tan} bounds wins over all copies; we use it to bound success on the chosen subset together with losses in the full batch. The reduction scores the selector without conditioning on the adversary's wins and samples the anchors after the selector is frozen, keeping the lifted game product-anchored. This supplies the repetition bound needed for sender security; \Cref{tab:game-correspondence} connects it to the protocol and accounts for the intervening disclosures.

\paragraph{Extraction, simulation, and commit-and-prove.}
Simulation security also requires the commitment to admit extraction without changing the device execution. The delayed-measurement extractor of~\cite{OT2} relies on control of the receiver's quantum measurements. Our DI string-commitment extractor instead changes only classical records, preserving the physical execution and its leakage (\Cref{sec:bc-step-two}).

To use the extracted values throughout the outer protocol, later checks must refer to the same immutable records. This is the role of commit-and-prove: an owner can successively certify predicates of its records, including projections. The public-output commitments of~\cite{CDR16} provide an interface-level precedent. \Cref{thm:masked-cp} realizes our interface from one extractable string commitment per owner, preserving the complete continuing state, including the adversary's quantum side information, device leakage, and reference (\Cref{sec:masked-cp-proof}). The outer OT can therefore use the certificates through stateful substitution.

Preserving the complete state is also essential when simulating the outer OT against a cheating Alice. Marginal indistinguishability of her views for $D=0,1$ does not by itself justify adjoining Bob's output: his raw answers may remain correlated with her residual quantum state. The trusted-quantum constructions account for this correlation through joint-state hybrids \cite[Appendix~B.1 and Figure~7]{OT1}\cite[Section~6.3]{OT2}. Our common-continuation argument ties the joint abort/output distribution to one extracted sender input pair while preserving the physical execution (\Cref{sec:results-overview}).

Beyond OT, the same interface yields DI bit commitment with simulators for both parties and DI coin tossing with negligible error, each in one device window (\Cref{sec:di-bc-coin}). For comparison, the information-theoretic DI protocols of Silman, Chailloux, Aharon, Kerenidis, Pironio, and Massar~\cite{SAC11} and of Aharon, Chailloux, Kerenidis, Massar, Pironio, and Silman~\cite{ACK14} have constant cheating probability and bias.

\paragraph{Robustness and leakage.}
Because the full reduction uses devices in both commitment preprocessing and the outer OT, honest noise must be tolerated throughout the construction. Robustness is already achieved by the bounded-storage DI OT protocol of Batra, Chakraborty, Jain, and Kapshikar~\cite{BCJK25}. Noise-tolerant extraction also appears in the trusted-quantum protocol of Diamanti, Grilo, Innocenzi, Lefebvre, Yacoub, and Y\'ang\"uez, which uses syndrome-based error correction \cite[Sections~1.2 and~2.5]{practical-ot}. Here the guarantee concerns the classical input--output interface of untrusted devices and must survive the commitment and the commit-and-prove compiler as well as the outer protocol.

Leakage imposes a corresponding requirement on simulation: the physical communication must be preserved through the reduction. Nielsen, Nordholt, Orlandi, and Burra use random bucketing to remove leakage of individual authenticated-OT choices in classical secure computation~\cite[Section~5]{NNOB12}; in our setting, quantum leakage may correlate the full selector vector. In DI key distribution, Zapatero and Curty~\cite{ZC26} model partial leakage of the inputs as a noisy channel in a CHSH-based protocol. We allow adaptive, bidirectional quantum communication between the laboratories, with a party acting as the adversary. Reproducing this leaked transcript unchanged is part of the same requirement that governs our extraction and OT simulation.

\subsection{Open problems}
\label{sec:open-problems}

We highlight two directions for simplifying the construction: reducing classical interaction and using a simpler nonlocal game. In both cases, the goal is to preserve simulation security, robustness, and leakage resilience under the same qOWF assumption.

\paragraph{Constant-round classical communication.}
The protocol uses at most three device windows, but its classical communication remains polynomial-round because the equivocal commitments and post-quantum zero-knowledge proofs use sequential repetition (\Cref{app:concrete-primitives}; see also~\cite[Section~1.3]{OT2}). Can all classical interaction, including preprocessing, be reduced to a constant number of rounds under qOWFs while retaining our simulation, robustness, and leakage guarantees? This would make the complete OT protocol constant-round.

\paragraph{CHSH-based constructions.}
Can the same guarantees be obtained using CHSH correlations in place of Magic Square? Our proof uses residual uncertainty in Alice's row after Bob has fixed his answer and learned her question. In CHSH, the winning relation $a\oplus b=xy$ determines Alice's answer from Bob's answer and both questions, so the residual-guessing argument does not carry over directly (\Cref{sec:selector-repetition-input}). A CHSH-based construction would need a different way to retain and certify the entropy required for OT against arbitrary joint receiver measurements, while accommodating noise and leakage.


\section{Preliminaries}
\label{sec:proof-tools-security}
This section fixes notation, recalls the Magic Square game, records standard probabilistic, quantum-information, and algebraic tools, and defines the simulation and commitment interfaces. The protocol constructions and their technical analyses begin in \Cref{sec:leakage-ot}. The conventions that govern every proof are in \Cref{sec:proof-conventions}.

For trits, $+$ and $-$ denote addition and subtraction in $\Zthree$, applied coordinatewise to trit strings. The symbol $\oplus$ denotes bitwise XOR. All logarithms are base two. Write $\rho\approx_\varepsilon\sigma$ if $\tfrac12\|\rho-\sigma\|_1\leq\varepsilon$; $\approx_c$ denotes quantum-computational indistinguishability.

\subsection{The Magic Square game}
\label{sec:magic-square-game}
In the Magic Square game~\cite[Section~3]{Coudron_Natarajan}, a referee samples independent uniform indices $x,y\in\{0,1,2\}$ and sends $x$ to Alice and $y$ to Bob. The players may share entanglement but cannot communicate after receiving their questions. Alice returns row $x$ as $a\in\{0,1\}^3$ and Bob returns column $y$ as $b\in\{0,1\}^3$. They win if the row has even parity, the column has odd parity, and the two answers agree at their intersection. The following predicate records these conditions.

\begin{definition}[Magic Square predicate]
\label{def:ms-predicate}
For $x,y\in\{0,1,2\}$ and $a,b\in\{0,1\}^3$, define
\[
 \MS(x,y,a,b):=
 \Id\!\left\{a_y=b_x,\quad
       \bigoplus_{j=0}^2a_j=0,\quad\bigoplus_{j=0}^2b_j=1\right\}.
\]
Set $\MS=0$ on malformed inputs. Every $\MS$ check includes all three conditions.
\end{definition}

\subsection{Concentration, entropy, and communication}
\label{sec:standard-tools}

We use, without further mention, the contractivity of trace distance under quantum channels~\cite{Tom16} and closure of computational indistinguishability under efficient postprocessing, which follows by composing the distinguisher.

\subsubsection{Concentration}
\label{sec:concentration-tools}

\begin{fact}[Chernoff--Hoeffding bounds \cite{Hoe63,SSS93}]
\label{fact:Chernoff}
Let $X_1,\ldots,X_n\in\{0,1\}$ be independent, $X=\sum_iX_i$, and $\mu=\Exp X$. For $0<\epsilon\leq1$,
\[
 \Pr[|X-\mu|\geq\epsilon\mu]\leq2e^{-\epsilon^2\mu/3}.
\]
For independent $X_1,\ldots,X_n\in[0,1]$ and $t>0$,
\[
 \Pr[X-\mu\geq tn]\leq e^{-2t^2n},\qquad
 \Pr[X-\mu\leq-tn]\leq e^{-2t^2n}.
\]
The same inequalities hold for a uniform sample without replacement from a finite population, with the population contained in $\{0,1\}$ for the multiplicative bound and in $[0,1]$ for the additive bounds \cite[Section~6]{Hoe63}.
\end{fact}

\subsubsection{Min-entropy and seeded extraction}
\label{sec:extractor-tools}

\begin{definition}[Conditional min-entropy~\cite{KRS09}]
  \label{def:conditional_min_entropy}
For a subnormalized cq-state $\rho^{XE}=\sum_x\ketbra{x}{x}\otimes\rho_x^E$, define
\begin{equation}
 \begin{aligned}
 \pg(X\mid E)_\rho
   &:=\max_{\substack{M_x\succeq0\\\sum_xM_x=I_E}}
                      \sum_x\Tr(M_x\rho_x^E),\\
 \hmin(X\mid E)_\rho&=-\log \pg(X\mid E)_\rho,\\
 \hmin^{\varepsilon}(X\mid E)_\rho
   &:=\sup_{\substack{\widetilde\rho\text{ cq, subnormalized}\\
                      P(\widetilde\rho,\rho)\leq\varepsilon}}
                 \hmin(X\mid E)_{\widetilde\rho},
 \end{aligned}
 \label{eq:guessing_prof_def}
\end{equation}
where $P$ is purified distance and $\log0=-\infty$.
\end{definition}

\begin{fact}[Classical leakage chain rule~{\cite[Lemma~6.18]{Tom16}}]
  \label{fact:minentropy-calculus}
For a subnormalized cq-state $\rho^{XEL}$ with classical $L$,
\[
 \hmin^{\varepsilon}(X\mid E,L)_\rho
 \geq \hmin^{\varepsilon}(X\mid E)_\rho-\log|\mathcal L|.
\]
\end{fact}

\begin{lemma}[Restriction to a cq-substate]
\label{lem:cq-substate-calculus}
For cq-substates $0\preceq\sigma^{XE}\preceq\rho^{XE}$,
\[
 \begin{aligned}
 \pg(X\mid E)_\sigma&\leq \pg(X\mid E)_\rho,\\
 \hmin(X\mid E)_{\sigma/p}
 &\geq \hmin(X\mid E)_\rho-\log(1/p),
 \qquad p:=\Tr\sigma>0.
 \end{aligned}
\]
Moreover, $\Tr(\rho-\sigma)\leq\delta$ implies
\[
 \begin{gathered}
 P(\rho,\sigma)\leq\sqrt{2\delta},\\
 \pg(X\mid E)_\sigma\leq2^{-h}
 \quad\Longrightarrow\quad
 \hmin^{\sqrt{2\delta}}(X\mid E)_\rho\geq h.
 \end{gathered}
\]
\end{lemma}
\begin{proof}
Every guessing measurement has no greater success on $\sigma$ than on $\rho$, and normalization divides its success by $p$. Since $\rho-\sigma\succeq0$, their generalized trace distance is $\Tr(\rho-\sigma)$. The relation between generalized trace distance and purified distance~\cite[Section~3.3]{Tom16} gives the displayed bound. The last implication follows from the definition of smoothing.
\end{proof}

\begin{fact}[Quantum leftover hashing~{\cite[Section~7.3.2]{Tom16}}]
\label{fact:quantum-linear-hash}
Let $\rho^{XE}$ be a normalized cq-state and let $\{h_s:\mathcal X\to\{0,1\}^q\}_{s\in\mathcal S}$ be a two-universal family: for all $x\ne x'$, $\Pr_S[h_S(x)=h_S(x')]\leq2^{-q}$, where $S\gets_U\mathcal S$ is independent of $XE$. Then
\[
 \frac12\left\|
 \rho_{h_S(X)SE}-U_q\otimes\rho_{SE}
 \right\|_1
 \leq\frac12\sqrt{2^{q-\hmin(X\mid E)_\rho}}.
\]
\end{fact}

\begin{lemma}[Linear extraction with a public input length]
\label{lem:linear-extraction}
Let $q\geq1$ be fixed and let $\rho^{XEN}$ be a subnormalized cq-state with classical $N\geq1$ and $X\in\Ftwo^N$. Conditional on $N=n$, sample $S\gets_U\Ftwo^{q\times n}$ independently of $XE$ and set $\ext_{n,q}(x;S)=Sx$. If $\hmin^{\varepsilon}(X\mid E,N)_\rho\geq h$, then
\[
 \frac12\left\|\rho_{\ext_{N,q}(X;S)SEN}
          -U_q\otimes\rho_{SEN}\right\|_1
 \leq2\varepsilon+\tfrac12 2^{(q-h)/2}.
\]
The family has seed length $qn$, uniform Boolean depth $O(\log n)$, and size $O(qn)$.
\end{lemma}
\begin{proof}
For every nonzero $x-x'$, the vector $S(x-x')$ is uniform in $\Ftwo^q$, so the matrix family is two-universal. First take $\varepsilon=0$. Write $\rho_n^{XE}=\langle n|\rho^{XEN}|n\rangle$, $p_n=\Tr\rho_n$, and $r_n=\pg(X\mid E)_{\rho_n}$. For $p_n>0$, apply \Cref{fact:quantum-linear-hash} to $\rho_n/p_n$ and multiply by $p_n$. Summing the resulting trace distances over the classical length blocks gives at most
\[
 \tfrac12\sqrt{2^q}\sum_n\sqrt{p_nr_n}
 \leq\tfrac12\sqrt{2^q\Bigl(\sum_np_n\Bigr)
                                  \Bigl(\sum_nr_n\Bigr)}
 \leq\tfrac12 2^{(q-h)/2},
\]
since $\sum_np_n\leq1$ and $\sum_nr_n=\pg(X\mid E,N)_\rho$. For $\varepsilon>0$, choose a smoothing state classical on $X,N$, apply the unsmoothed bound, and use contractivity and the triangle inequality on the two marginals, adding at most $2\varepsilon$. Each output bit is the parity of $n$ products; a balanced XOR tree gives the stated circuit bounds.
\end{proof}

For $q=1$, this is the inner-product extractor $\ext_n(x;S)=\langle x,S\rangle_{\Ftwo}$. The independent public seed requires no additional entropy deduction. For constant $q$, $h=\omega(\log\lambda)$ and $\varepsilon=\negl(\lambda)$ give negligible extraction error.

\begin{fact}[Helstrom bound; see, e.g., \cite{Wil17}]
\label{fact:helstrom}
Let $\sigma_0,\sigma_1\geq0$ act on one space, with $\operatorname{Tr}(\sigma_0+\sigma_1)=1$. The maximum, over all measurements, of the probability of guessing $s$ from $\sigma_s$, where $s$ occurs with probability $\operatorname{Tr}\sigma_s$, equals $\tfrac12\bigl(1+\|\sigma_0-\sigma_1\|_1\bigr)$.
\end{fact}

\subsubsection{Bounded communication}
\label{sec:communication-tools}

The following form of communication elimination follows the transcript-guessing argument of Jain and Kundu \cite[proofs of Lemmas~8 and~19]{Jain_Kundu_DIQKD}.

\begin{lemma}[Eliminating bounded interactive communication]
\label{lem:interactive-transcript-guessing}
Let $\mathcal S$ exchange at most $c$ classical bits and $q$ qubits within a specified causal window, with arbitrary prior entanglement. The classical budget includes adaptive routing, direction, lengths, continuation, and stopping; fixed public schedules have zero cost. For every event $\mathcal E$, there exists a strategy $\mathcal S_0$ with no communication in that window satisfying
\[
 \Pr_{\mathcal S_0}(\mathcal E)
 \geq2^{-(c+2q)}\Pr_{\mathcal S}(\mathcal E)
 \geq2^{-2b}\Pr_{\mathcal S}(\mathcal E),\qquad b:=c+q.
\]
The strategy $\mathcal S_0$ respects the prescribed causal order. On the branch where its transcript guess is correct, its state on the original registers is the state of $\mathcal S$, including any reference registers.
\end{lemma}
\begin{proof}
Teleport each communicated qubit using fresh shared entanglement, replacing it with two classical correction bits. Encode and pad the resulting classical transcript, including the charged schedule information, to length $c+2q$. The parties share a uniform guess of this string. Each recipient uses the guessed incoming messages at their prescribed times, and each sender checks its own outgoing messages against the guess. A party detecting a mismatch uses a legal fallback continuation.

On a jointly consistent branch, these operations reproduce the original strategy. For every complete transcript, guessing that transcript has probability $2^{-(c+2q)}$; hence the consistent substate on the original registers equals $2^{-(c+2q)}$ times the original state. Its contribution to $\mathcal E$ proves the first inequality, and $c+2q\leq2(c+q)$ proves the second. Guessed messages are used only when the corresponding messages would have arrived, so the construction preserves the causal order.
\end{proof}

We apply \Cref{lem:interactive-transcript-guessing} to the device input--output interval. The simulation leaves immutable records unchanged; the prescribed operations outside that interval are unchanged.

\subsection{Algebraic codes and quantum masking}
\label{sec:coding-masking-tools}

The following facts give the algebraic and masking tools used in \Cref{sec:unified-parameters}.

\paragraph{Function fields and evaluation codes.}
For a function field over $\Ffield$, a divisor $G$ specifies the finite-dimensional space $\mathcal L(G)$ of functions whose poles are bounded by $G$. Evaluation at distinct rational places outside its support is an $\Ffield$-linear map.

\begin{fact}[Riemann--Roch and effective tower codes]
\label{fact:effective-ag-tower}
For a function field of genus $g$, $\dim\mathcal L(G)\geq\deg G+1-g$, with equality when $\deg G>2g-2$. A nonzero function with poles bounded by an effective divisor of degree $a$ has at most $a$ distinct zeros outside that divisor \cite[Chapters 1--2]{Stichtenoth09}. The tower
\[
 F_0=\mathbb F_{64}(x_0),\qquad
 F_{j+1}=F_j(x_{j+1}),\qquad
 x_{j+1}^{8}+x_{j+1}=\frac{x_j^{8}}{x_j^{7}+1}
\]
has a unique pole $P_\infty$ of $x_0$. It has $N_j=56\cdot8^j$ rational evaluation places above $x_0=\alpha$ with $\alpha^8+\alpha\ne0$, all distinct from $P_\infty$, and genus $g_j<8^{j+1}$ \cite{GarciaStichtenoth96}. For fixed field and pole degree $O(N_j)$, bases of $\mathcal L(aP_\infty)$ and their evaluations at these places are computable deterministically in time polynomial in $N_j$ \cite{ShumEtAl01}.
\end{fact}

\begin{fact}[Quantum masking by a small-bias family]
\label{fact:quantum-small-bias}
Let $(Z_s)_s$ be distributions on $\mathbb F_2^n$ such that for every nonzero binary character $\chi$, $\mathbb E_s|\mathbb E\chi(Z_s)|^2\leq\delta^2$. Let $\rho_{XE}$ be a normalized cq state with $\hmin(X\mid E)\geq h$. Choose $s$ uniformly from the finite index set, independently of $XE$, and sample $Z_s$ independently of $XE$ given $s$. Then
\[
 \bigl\|\rho_{s,X+Z_s,E}
       -\rho_s\otimes\tau_n\otimes\rho_E\bigr\|_1
 \leq\delta\,2^{(n-h)/2},
\]
where $\tau_n$ is uniform on $\mathbb F_2^n$ \cite[Theorem 3.2 and its min-entropy formulation, arXiv:0706.2606v2]{FehrSchaffner08}.
\end{fact}

\subsection{Simulation security}
\label{sec:simulation-security}

Simulator-based security compares real and ideal executions. We first specify the physical leakage interface and the admissible initial states, then define the comparison and the boundaries between calls.

\subsubsection{Physical leakage interface}
\label{sec:physical-interface}

One device window contains all parallel Magic Square coordinates of one batch. The following interface bounds the communication available to the laboratories during that window.

\begin{definition}[Admissible OT leakage]
  \label{def:ot-admissible-leakage}
Fix a corruption experiment.
\begin{enumerate}[label=\textnormal{(\roman*)},leftmargin=1.9em,itemsep=1pt]
\item \emph{Window and budget.} From the instant the first Magic Square input is supplied until all outputs in the complete block batch have been collected, the laboratories may exchange arbitrary adaptive, bidirectional quantum messages of total cost at most $\ell$ qubits, correlated across blocks and with arbitrary private memory. The budget is global, not renewed per block. The interface then closes permanently.
\item \emph{Accounting.} Classical bits count as computational-basis qubits. Every adaptive control symbol, including direction, routing, message length, continuation, and stopping, is encoded in a self-delimiting classical transcript and charged to the same budget; only a fixed public communication schedule is free. All systems delivered to the corrupted laboratory, its retained memory, and its classical transcript copies enter the corrupted party's side information.
\item \emph{Exclusion.} The interface cannot read honest trusted inputs, coins, or storage. An honest device receives only its prescribed question and charged traffic: Alice's offset, test masks, anchors, selectors, and later transcripts are not separate device inputs.
\end{enumerate}
We write $\ell_{\sf A}$ for the budget in the experiment with corrupted Alice and $\ell_{\sf B}$ for the budget with corrupted Bob.
\end{definition}

\noindent The exclusion is necessary: an interface that could read the honest choice could transmit it, and receiver privacy would be impossible. Later selectors, fresh seeds, CP operations, certified projections, audits, the transfer, and the predicate-proof messages run outside the window; source-dependent disclosures are charged explicitly in the entropy proof.

\subsubsection{Real and ideal executions}

The admissible initial state separates trusted classical storage from the untrusted devices and the adversary's reference.

\begin{definition}[Admissible initial extension]
  \label{def:admissible-extension}
Fix classical party inputs and one corruption. Write $I_H,W_H,K_H$ for honest trusted input, blank workspace, and fresh coins; $D_A,D_B$ for the outer device laboratories; and $P,\Rref$ for corrupt workspace and inaccessible reference.
\begin{enumerate}[label=\textnormal{(\roman*)},leftmargin=1.9em,itemsep=1pt]
\item \emph{State.} An admissible initial extension has
\[
 \begin{aligned}
 \eta^{I_HW_HD_AD_BP\Rref}
 &=\ketbra{i_H}{i_H}\otimes\ketbra{0}{0}^{W_H}
       \otimes\eta^{D_AD_BP\Rref},\\
 \eta^{K_HD_AD_BP\Rref}&=\tau_{K_H}\otimes\eta^{D_AD_BP\Rref},
 \end{aligned}
\]
where $\eta^{D_AD_BP\Rref}$ is arbitrary, subject to the fixed corrupt input.
\item \emph{Computation.} Active instruments are nonuniform QPT on polynomial-size registers, and no operation accesses $\Rref$.
\item \emph{Trusted storage.} Honest trusted storage has no physical leakage interface. Its only outputs are prescribed classical messages and device inputs.
\item \emph{Communication and subprotocols.} Inter-laboratory communication is exactly \Cref{def:ot-admissible-leakage}. CP and ZK operations never access $D_A,D_B,\Rref$ and introduce no outer device communication. They occur outside the outer input--output window; only saved records persist within it. Compiled calls use separate batches satisfying \Cref{def:compiled-device-calls}.
\end{enumerate}
\end{definition}

We use the standard auxiliary-input-uniform quantifier order
\[
  \forall\,\mathsf P^*\ \exists\,\Sim\
  \forall\,\eta^{\Rref},
\]
where one simulator chosen from the adversary works for every admissible extension, uniformly also over the fixed honest input values. All simulators are QPT, with a polynomial bound on every branch.

\begin{definition}[Computational augmented security]
  \label{def:comp-aug-sec}
\begin{enumerate}[label=\textnormal{(\roman*)},leftmargin=1.9em,itemsep=1pt]
\item \emph{Compared states.} $\mathsf{Real}_{\mathcal P_f,\mathsf P^*}(\eta^{\Rref})$ retains
\[
 (\text{corrupt final cq-state},\ \text{honest input/output},\
   \text{public transcript},\ \Rref),
\]
tracing out only honest call-local workspace irrevocably discarded at return; $\mathsf{Ideal}_{f,\Sim}(\eta^{\Rref})$ retains the same systems.
\item \emph{Requirement.} Computational augmented security means
\[
 \begin{gathered}
 \forall\,\mathsf P^*\in\QPT\ \exists\,
 \Sim_{\mathsf P^*}\in\QPT\ \forall\,\eta^{\Rref}:\\
 \mathsf{Real}_{\mathcal P_f,\mathsf P^*}(\eta^{\Rref})
 \approx_c\mathsf{Ideal}_{f,\Sim_{\mathsf P^*}}(\eta^{\Rref}).
 \end{gathered}
\]
For each nonuniform QPT test, its negligible advantage bound is uniform over admissible $\eta^{\Rref}$ and fixed honest inputs.
\item \emph{The device resource.} Both experiments use the same initial call-local honest device share $D_H$. In the ideal experiment, a declared interface $\mathsf{Dev}_H$ allows SIM only the prescribed honest classical device operations and leakage window on that share, preserving its correlations with the corrupt workspace and $\Rref$. No reset, copying, or re-preparation of $D_H$ is allowed.
\item \emph{What SIM receives.} SIM receives the corrupt initial workspace, including its device share, and ideal releases, but no honest trusted input or protected pending state; it acts as $\Sim\otimes\mathrm{id}_{\Rref}$.
\end{enumerate}
All realization claims, including the notation $\mathsf{Ideal}_{f,\Sim}$, are relative to these declared device resources.
\end{definition}

Throughout, a statement that one functionality is realized or reduced to another includes this efficient-simulation requirement: every efficient corrupted strategy has an auxiliary-input-uniform QPT simulator. The BC-to-CP compiler uses the continuing-state contract in \Cref{def:vector-commitment}, realized by our DI source commitment.  The ideal functionalities themselves are interfaces; security is a property of the protocols realizing them.

\subsubsection{Invocation boundaries}

Sequential calls preserve the joint continuing state and use fresh honest coins at each reached invocation.

\begin{definition}[Invocation boundary]
\label{def:ot-call-boundary}
Calls have distinct labels and disjoint execution intervals.
\begin{enumerate}[label=\textnormal{(\roman*)},leftmargin=1.9em,itemsep=1pt]
\item \emph{Admissible state and fresh coins.} At every reached call in every replacement hybrid, and for every classical input $i$ with $\Pr[I=i]>0$,
\[
 \rho_{\mid I=i}\in\mathsf{Adm},\qquad
 \rho_{K,\mathrm{rest}\mid I=i}=\tau_K\otimes\rho_{\mathrm{rest}\mid I=i},
\]
where $\mathsf{Adm}$ comprises \Cref{def:admissible-extension} and the subroutine's device conditions, and $K$ denotes fresh honest call coins.
\item \emph{Reference.} The reference contains all untouched honest outer, scheduler, channel, and future-call registers.
\item \emph{Return interface.} It retains honest input/output and the entire corrupt state, including transcript, leakage, and device residuals. Only irrevocably discarded honest scratch and continuation-inaccessible devices are traced out.
\item \emph{Devices.} Only the current batch is accessible, through its declared device interface; future batches remain protected. At return, the current interface closes, and its retired honest share is discarded in both experiments.
\item \emph{Aborts.} Abort branches are included in the real--ideal comparison. On such branches, the honest parties make no further calls.
\end{enumerate}
\end{definition}

\begin{lemma}[Efficient continuation and polynomial hybrids]
\label{lem:efficient-composition-hybrids}
For joint state ensembles $\rho\approx_c\sigma$ and every QPT channel $\mathcal E$,
\[
 \mathcal E(\rho)\approx_c\mathcal E(\sigma).
\]
For $q=\poly(\lambda)$, hybrids $H_0,\ldots,H_q$, and a QPT test $\mathcal D$, put $p_j=\Pr[\mathcal D(H_j)=1]$. Then
\[
 |p_0-p_q|\leq\sum_{j=1}^q|p_{j-1}-p_j|
 \leq q\max_{1\leq j\leq q}|p_{j-1}-p_j|.
\]
In particular, $\max_j|p_{j-1}-p_j|=\negl(\lambda)$ implies $|p_0-p_q|=\negl(\lambda)$, including arbitrary auxiliary registers.
\end{lemma}

\begin{proof}
Compose the distinguisher with the channel for the first assertion. The triangle inequality on the telescoped acceptance probabilities gives the second. Both arguments retain arbitrary auxiliary registers.
\end{proof}

The selector argument uses the following elementary marginal identity.

\begin{fact}[Local no-signalling~\cite{Wil17}]
  \label{fact:local-instrument-marginal}
For every $\rho^{ABR}$ and CPTP map $\mathcal N:B\to B'$,
\[
 \Tr_{B'}[(\mathrm{id}_{AR}\otimes\mathcal N)(\rho)]
 =\rho^{AR}.
\]

\end{fact}

\subsection{Oblivious transfer}
\label{sec:ot-interface}

The ideal functionality delivers the selected bit to the receiver and no value to the sender. Its delivery rule permits abort after the corrupted party receives its prescribed output.

\begin{definition}[Ideal functionality for oblivious transfer]
\label{def:idealOT}
For an ordered sender/receiver pair, $\ideal(\mathsf{OT})$ has
\[
 \mathcal I_S=\{0,1\}^2,\quad \mathcal I_R=\{0,1\},\qquad
 ((S_0,S_1),d)\longmapsto(\emptystring,S_d).
\]
$\perp$ is not an input. The simulator may issue a global abort at any time before honest delivery, including before corrupt input; then every honest party receives $\perp$ and the call terminates. After valid inputs, release the corrupt prescribed output to the simulator first. For the honest-party set $H$, subsequently deliver
\[
 \mathsf{out}_H=
 \begin{cases}
 \perp^H,&\text{abort},\\
 (\emptystring,S_d)|_H,&\text{otherwise}.
 \end{cases}
\]
Thus a corrupt receiver may learn $S_d$ before suppressing the sender's empty output. Alice is the OT sender below.
\end{definition}

\subsection{Commit-and-prove functionality}
\label{sec:cp-layer}

The outer OT uses a commit-and-prove functionality to store immutable strings and certify owner-authorized predicates on them. Public-output commitments provide an interface-level precedent for linking protocol components through commitments~\cite{CDR16}. \Cref{thm:masked-cp} realizes this functionality from qOWFs and the source commitment of \Cref{def:vector-commitment}.

\begin{definition}[Commit-and-prove functionality]
\label{def:cp}
$\ideal(\mathsf{CP})$ is a classical functionality that maintains polynomially many immutable bit strings of prescribed public polynomial lengths. Each record has a unique session and record identifier, a designated owner, and a designated verifier. Identifiers and roles are part of every command and transcript and are left implicit in the notation.
\begin{enumerate}[label=\textnormal{(\roman*)},leftmargin=1.9em,itemsep=1pt]
\item \emph{Commitment.} On the owner's first $(\mathsf{commit},M)$ command, check that $M$ is a bit string of the prescribed public length; otherwise abort the session without creating a record. On success, store $M$ and issue a receipt consisting of the identifier, the roles, and the length.
\item \emph{Predicate certification.} Only the common owner may issue $(\mathsf{check},\mathbf a,P,y)$, where $\mathbf a=(a_1,\ldots,a_r)$ is an ordered nonempty tuple of distinct previously committed record identifiers from this session, all with that owner and the same designated verifier. Every record must already have its prescribed receipt. Check the tuple, public lengths, and authorized predicate schedule. The Boolean circuit $P$ and the tuple length obey their public polynomial bounds. Publish the request identifier, $\mathbf a$, $P$, and $y$. If $P(M_{a_1}\|\cdots\|M_{a_r},y)=1$, issue a success certificate subject to the delivery rule below; otherwise abort the session.

\item \emph{Selective opening.} $\mathsf{Open}(M_S)$ is the request $(\mathsf{check},P_S,(S,z))$ with $P_S(M;S,z)=[\pi_S(M)=z]$, where $\pi_S(M)$ is the substring of $M$ at the positions in $S$. Any two successfully certified openings must agree on all positions common to their requested sets in the same immutable record.
\item \emph{Rejections.} Wrong-role, cross-session, duplicate-commitment, unknown or unreceipted handles, repeated handles, wrong-owner or wrong-verifier tuples, malformed, and out-of-schedule commands cause session abort. Apart from receipts, only the contents of an owner-authorized public request and its success or $\bot$ are released. No record is ever updated.
\item \emph{Delivery with abort.} The adversary schedules the delivery of receipts, certificates, and $\bot$, and receives its own releases first. It may send $\bot$ at any time. On abort, every honest party receives $\bot$, and the session terminates; no further operation on any record in that session is permitted. No fairness is required~\cite[Chapter~7]{Goldreich04}.
\end{enumerate}
\end{definition}

\noindent Four consequences are used repeatedly. A receipt is no function of the contents of $M$. A failed check does not certify that the predicate is false, and success is impossible when $P(M,y)=0$. An opening request makes $(S,z)$ public even if the session aborts, and only a successful certificate authenticates $z$; no opening randomness and no unselected entry is revealed. The verifier learns nothing about an honest owner's record without an owner-initiated request, and predicate checks remain available after successful projections. Honest OT controllers authorize only the requests that the protocol prescribes, and the functionality never acts on a device register or on the reference.

Write $\mathsf{Commit}(M_a)$ for creation of record $a$ and $\mathsf{Check}((a_1,\ldots,a_r);P,y)$ for an owner-authorized check on that ordered tuple. When $r=1$, $\mathsf{Check}(M_a;P,y)$ is shorthand for the same operation. Identifiers, rather than hidden contents, locate records. Every handle and stored value is immutable after its receipt; a tuple check does not recommit or update any component. A simulator may log the corrupted owner's classical commitment inputs; there is no public read or extraction command for an honest owner's record.

The commitment and predicate simulators must preserve the complete corrupted continuation and untouched reference, uniformly over all admissible inputs. All records of one OT call, including opposite-owner records and delayed projections, belong to one CP session and share the stateful simulation contract proved in \Cref{thm:masked-cp}. The externally compared continuation excludes an honest record's internal commitment seeds and proof coins. They remain encapsulated behind the prescribed operations; protected outer state and other pending honest sessions remain part of the continuing reference.

\begin{fact}[Post-quantum commitments from one-way functions]
\label{fact:pqowf-np-anchors}
Post-quantum one-way functions imply classical-message statistically binding, quantum-computationally hiding bit commitments, by the pseudorandom-generator construction of H{\aa}stad, Impagliazzo, Levin, and Luby~\cite{HILL} and Naor's commitment~\cite{Naor91}; see \cite[Section~3.1.1 and Footnote~6]{OT2} for the post-quantum instantiation.
\end{fact}

\Cref{app:concrete-primitives} gives the concrete commitments and zero-knowledge proofs used by the compiler and the source protocol.

\subsection{Commitment interfaces}
\label{sec:commitment-interfaces}

The reductions use two commitment interfaces: an extractable string commitment for preprocessing, and a bit commitment with simulation against either party.

\subsubsection{Extractable string commitments}

We identify $q$-bit strings with elements of $\Ftwo^q$ for the code algebra. The source interface opens the entire string; the CP compiler in \Cref{sec:bc-cp-compiler} supplies certified substring openings. No succinctness requirement is imposed on per-position opening proofs.

Fix a public length $1\le q\le\poly(\lambda)$. A string commitment to $u\in\Ftwo^q$ has a completed-commitment public transcript $\tau$ and a polynomial-time native-opening predicate $\Ver_{\rm str}(\tau,u,w)$. The witness $w$ is a polynomial-length classical string. The predicate includes the session, syntax, and consistency with every field already opened. The honest receiver's retained verification record is determined by $\tau$.

\begin{definition}[Extractable string commitment]
\label{def:vector-commitment}
A correct string commitment satisfies the following properties in either role assignment, uniformly over the admissible auxiliary and reference states of \Cref{def:comp-aug-sec}.
\begin{enumerate}[label=\textnormal{(\roman*)},leftmargin=2em,itemsep=3pt]
\item \emph{State-preserving extraction.} For every QPT committer, one QPT simulator returns $u^*\in\Ftwo^q\cup\{\perp\}$ at commitment completion. Omitting this private output, its joint corrupt state, public transcript, honest retained verification state, and untouched reference are computationally indistinguishable from the real ones.
\item \emph{Agreement with all native openings.} For that same simulator,
\begin{equation}
 \Pr_{\rm sim}\bigl[\exists u,w:\
 \Ver_{\rm str}(\tau,u,w)=1\ \land\ u\ne u^*\bigr]
 \le\varepsilon_{\rm ao}(\lambda)=\negl(\lambda).
 \label{eq:vector-all-openings}
\end{equation}
Every string differs from $\perp$.
\item \emph{Joint hiding.} For every QPT receiver, commitments to any prescribed $u_0,u_1\in\Ftwo^q$ give computationally indistinguishable receiver--reference views through commitment completion. The bound is uniform over these messages and admissible auxiliary states. The comparison may retain a common classical label specifying both candidate messages, provided that label belongs to the admissible pre-input state and is independent of fresh game questions, or is a function of receiver information already available or explicitly charged in the hiding proof. In particular, one may compare commitments to a fresh random $\mu$ and to $0^q$ with the same $\mu$ label on both sides. The honest native witness is private.
\item \emph{Honest native witnesses.} Against every receiver, on each continuing branch an honest committer can compute $w$ from its saved classical state with $\Ver_{\rm str}(\tau,u,w)=1$, without a further device operation.
\end{enumerate}
The guarantees permit sequential completed phases and subsequent QPT continuations. Other pending honest states remain untouched during a source simulation; retired honest device shares remain inaccessible.
\end{definition}

\subsubsection{Bit commitments with simulation for both parties}

For a bit commitment, the receiver's simulator must also reproduce the opening after learning the committed bit.

\begin{definition}[Bit commitment with simulation for both parties]
\label{def:bc-two-sided}
A correct bit commitment has a commitment receipt and one opening attempt, after which the session closes. It is simulation-secure if, uniformly over the admissible auxiliary and reference states:
\begin{enumerate}[label=\textnormal{(\roman*)},leftmargin=2em,itemsep=3pt]
\item For every QPT committer, one QPT simulator extracts a bit $b^*$ by the receipt. Omitting this private output, the joint corrupt state, retained honest verification state, and reference are computationally indistinguishable from the real state. Any subsequent accepted opening has value $b^*$ except with negligible probability.
\item For every QPT receiver, one QPT simulator runs the commitment phase without the honest bit. Given that bit only at opening, it simulates the opening and the receiver's continuing state, jointly with the reference and the prescribed honest input--output interface.
\end{enumerate}
The guarantees include rejection and withholding, permit sequential completed phases and delayed opening, and use only the declared honest device interface of \Cref{def:comp-aug-sec}.
\end{definition}

\subsection{Device invocations}
\label{sec:device-invocations}

In a compiled protocol, a source commitment supplies a random string for masking later classical records. Call the committed random string, its native witness, and its allocation counter a \emph{mask bank}. The following conditions specify the device access and retirement rules for such invocations.

\begin{definition}[Device invocations in the compiled protocol]
\label{def:compiled-device-calls}
An admissible compiled schedule satisfies the following conditions.
\begin{enumerate}[label=\textnormal{(\roman*)},leftmargin=1.9em,itemsep=3pt]
\item \emph{Batches and budgets.} Each invocation has a distinct single-use batch $\mathcal B_j$, fresh controller questions, and a budget $\ell_j$ covering all its coordinates and adaptive routing. A global budget $\ell$ requires
\[
 \operatorname{cost}(\mathcal L_j)\le\ell_j,
 \qquad \sum_j\ell_j\le\ell.
\]
\item \emph{Reached states.} Every reached invocation, including in replacement hybrids, satisfies its input and device hypotheses. Arbitrary initial entanglement is permitted subject to these hypotheses; untouched outer and future-device registers are retained.
\item \emph{Timing.} Classical phases lie outside device windows. Every equivocation or proof-simulation segment completes before any other outer operation or device invocation.
\item \emph{Isolation and private records.} Distinct batches never communicate, also within one laboratory. Honest masks, native witnesses, and proof coins are private behind their prescribed interfaces; callers cannot read or alter unused masks.
\item \emph{Retirement.} At receipt, each honest source receiver discards its questions, answers, labels, and device share. Banks keep their private classical witnesses only until the CP session's terminal delivery or abort. Then all honest call-local masks, raw tables, witnesses, and coins are erased, retaining only prescribed input--output and public transcript. Retired honest devices are never reused.
\end{enumerate}
\end{definition}

\subsection{Proof conventions}
\label{sec:proof-conventions}
Three conventions hold in every proof. \emph{(i) Subnormalized branches.} For an event $\mathcal E$, $\rho_{\mathcal E}$ denotes its subnormalized branch, of trace $\Pr(\mathcal E)$. Error, entropy, and hashing bounds are stated for such branches; they are therefore absolute and hold at every acceptance probability. Abort branches are retained, and renormalization is stated explicitly where it occurs. \emph{(ii) Fixed statistical hybrids.} Every information-theoretic step takes place inside one fixed experiment: the ideal CP execution for OT, and the equivocal-commitment executions $\Eqa,\Eqb$ of \Cref{sec:di-bc} for BC. Inside such an experiment, the unique value of a statistically binding commitment may be assigned by an unbounded function of the classical transcript. This value is an analysis register: no algorithm reads it, and it never enters a computational comparison. Only the final efficient real/ideal comparison is transferred to the concrete implementation; computational errors are added to the statistical bound and are never multiplied by a guessing factor. \emph{(iii) Totalization.} A variable undefined on a stopped branch receives a fixed default in a private analysis register; the compared views contain only the actual stopped transcript.

\subsection{Glossary of recurring terms}
\label{sec:glossary}

\Cref{tab:glossary} collects the terms used by the three reductions.
\begin{table}[!htbp]
\centering
\caption{Recurring terms for the common device experiment.}
\label{tab:glossary}
\normalsize
\setlength{\tabcolsep}{5pt}
\renewcommand{\arraystretch}{1.1}
\begin{tabular}{>{\raggedright\arraybackslash}p{0.26\textwidth}
                >{\raggedright\arraybackslash}p{0.66\textwidth}}
\hline
\textbf{Term} & \textbf{Meaning}\\
\hline
coordinate, batch & One Magic Square game; all coordinates of one
single-use invocation.\\
device window & The interval in which a batch receives its questions
and returns its answers (\Cref{def:ot-admissible-leakage}).\\
symbol, microblock & One of $m$ field coordinates; one of its $k=6$
groups of $N_0$ device coordinates.\\
test set, audit & Alice's hidden set $T_A$ of openly tested coordinates;
Bob's one hidden audit coordinate in each microblock.\\
label, selector & Bob's private bit $\beta_{ij}$ and the unaudited
coordinates $J_{ij}$ of the corresponding question class. The six bits
form one $\Ffield$-valued label $\beta_i$.\\
anchor & A late independent flag $C_x=0$, on which Alice replaces her
record by $(X_x,A_x)=(0,000)$.\\
erased symbol & An index $i\in\mathcal E$ for which all six selectors are empty;
distinct from a coordinate anchor.\\
source & Alice's raw device answers on the frozen selectors.\\
receipt, certified projection & CP's acknowledgement containing only
identifiers, roles, and length; a public substring with a success
certificate (\Cref{def:cp}).\\
code mask, mode & The message $M=\beta+R$; the prescribed secret of
$R$ in the string commitment.\\
good audit & An accepted audit whose underlying whole-batch discrepancy
count is within the outer code's decoding bound.\\
independent-mask experiment & An analysis experiment with a uniform
$M\in\Ffield^m$ independent of the preceding view.\\
call-local & Belonging to one invocation and discarded when it returns.\\
totalized & Assigned a private default on a stopped branch, while the
compared transcript retains its actual abort.\\
\hline
\end{tabular}
\end{table}

\section{Device-independent oblivious transfer from commit-and-prove}
\label{sec:leakage-ot}

We first construct OT with immutable classical records and certified predicates. The code family of \Cref{lem:multiplicative-codes} links linearly many private selector bits to one receiver choice. A global audit and one global decoder tolerate errors across the batch. \Cref{sec:bc-cp-compiler,sec:di-bc} realize the required classical functionality; the present section uses its ideal interface. We use the OT functionality and physical interface of \Cref{def:idealOT,def:ot-admissible-leakage}. After specifying honest faults and execution conventions, we construct the codes and mask, choose the parameters, and give the transfer protocol and its proof.

\subsection{Honest-device faults and execution conventions}
\label{sec:ot-notions}

We specify the honest-device noise assumption and how the outer OT uses the interfaces of \Cref{sec:proof-tools-security}.

\paragraph{Honest-device faults.}
Honest controllers encode an active row as $(a_0,a_1,a_0\oplus a_1)$ and a column as $(b_0,b_1,1\oplus b_0\oplus b_1)$ from the two recorded bits. The honest coordinate implementations are independent, though they need not be identical. Normalization preserves every winning Magic-Square record. The fault rate concerns the resulting parity-correct records: for every question pair,
\begin{equation}
 \Pr[\MS(x,y,A_i,B_i)=0\mid X_i=x,Y_i=y]\leq q_{\rm dev}.
 \label{eq:honest-device-fault}
\end{equation}
Classical computation, communication, storage, commitments, and postprocessing have their ordinary correctness guarantees. Anchors contribute the known bit zero. Independence is used only for honest correctness; security permits the stated arbitrary entangled attacks.

Under the schedule, take $q_{\rm dev}\le\nu/2$. Chernoff's bound then gives the consequence used by the proofs:
\begin{equation}
 \Pr\!\left[\bigl|\{a\in[L]:
  \MS(\Xraw_a,Y_a,\Araw_a,B_a)=0\}\bigr|>\nu L\right]
 \le\negl(\lambda).
 \label{eq:honest-noise-promise}
\end{equation}
Trusted test and audit coins are fresh and are not used during the physical execution.

We use CP commitment, owner-authorized certified projections, and predicate checks from \Cref{prot:cp-layer}. This classical layer does not act on an outer device register. Device actions assigned to the same window are concurrent local instructions. A global abort returns $\perp$ to every honest party. At delivery or abort, an honest party retires all call-local devices and discards private call-local scratch, retaining its input--output interface and public transcript; untouched outer state remains protected. A corrupted receiver may request abort after learning its selected bit; the simulator relays it before honest-sender delivery under \Cref{def:idealOT}. Otherwise Alice outputs $\emptystring$.

\subsection{Codes, choice masking, and parameters}
\label{sec:unified-parameters}

The transfer and source commitment use the same codes, selectors, and audit thresholds. A batch has $m$ code coordinates, each supplied by $k=6$ physical groups of a fixed size $N_0$. The field is $\Ffield$; $q$ denotes the number of secret slots, with $q=1$ for bit OT and $q$ binary slots for a source commitment. We construct the codes and choice mask from the tools of \Cref{sec:coding-masking-tools}, then choose the parameters.

\paragraph{Codes with secret evaluation points.}

For the tower in \Cref{fact:effective-ag-tower}, use the admissible lengths
\[
 m_j=\left\lfloor\frac{100N_j}{101}\right\rfloor,\qquad j\geq1.
\]
Then $m_1=443$, $m_j+\lfloor m_j/100\rfloor\leq N_j$, $g_j<0.146m_j$, and $m_{j+1}<9m_j$. For the genus bound, write $x=8^j\geq8$ and use $m_j>5600x/101-1$ and $8x<0.146(5600x/101-1)$. For the length bound, $m_{j+1}\leq8m_j+7<9m_j$. Let $\mathsf{Next}(a)$ be the first admissible length at least $\max\{443,a\}$. Thus
\[
 \max\{443,a\}\leq\mathsf{Next}(a)<9\max\{443,a\}.
\]

Fix an admissible $m$, $1\leq q\leq\lfloor m/100\rfloor$, $q$ secret places $Q_1,\ldots,Q_q$, and $m$ tail places $P_1,\ldots,P_m$. Put $G=\lfloor0.35m\rfloor P_\infty$ and
\[
 \CodeB=\CodeD=\operatorname{ev}_{P}\mathcal L(G),\qquad
 \CodeA=\operatorname{ev}_{P}\mathcal L(2G).
\]
For each of these codes, $\secmap$ evaluates the corresponding function at $(Q_1,\ldots,Q_q)$. A subscript $u$ denotes the fiber of this map over $u\in\Ffield^q$; subscript zero denotes its kernel. The product $v*w$ is coordinatewise multiplication.

\begin{lemma}[Multiplication-compatible codes]
\label{lem:multiplicative-codes}
Tail evaluation is injective, all three secret maps are onto, and
\[
 \CodeB*\CodeD\subseteq\CodeA,\qquad
 \secmap(b*d)=\secmap(b)*\secmap(d).
\]
The distances satisfy
\[
 \begin{split}
 d(\CodeB_0^\perp),d(\CodeD_0^\perp)
   &\geq\deg G-q-2g+2>0.047m,\\
 d(\CodeD)&\geq m-\deg G\geq0.65m,\qquad
 d(\CodeA)\geq m-2\deg G\geq0.30m.
 \end{split}
\]
Code membership, secret recovery, uniform sampling in any secret fiber, and recovery of a codeword from all but at most $\lfloor m/100\rfloor$ erased coordinates are polynomial-time operations.
\end{lemma}
\begin{proof}
The pole degrees $\deg G$ and $2\deg G$ are smaller than $m$, so the zero bound in \Cref{fact:effective-ag-tower} proves injectivity and the two code-distance bounds. Write $S=\sum_{a=1}^q Q_a$. Since $\deg G-q>2g-2$, Riemann--Roch gives $\dim\mathcal L(G)-\dim\mathcal L(G-S)=q$. This proves surjectivity; the same argument applies to $2G$. Multiplying two functions in $\mathcal L(G)$ gives a function in $\mathcal L(2G)$, and evaluation respects multiplication.

The zero fiber is the evaluation of $\mathcal L(G-S)$. On any $r\leq\deg G-q-2g+1$ tail places, its evaluation map is onto, again by the Riemann--Roch dimension formula. Consequently no nonzero dual word has support of size at most $r$. This gives the stated dual-distance bound, since $\deg G-q-2g+2>0.048m$.

Compute the one-point bases using \Cref{fact:effective-ag-tower}. All remaining operations are linear algebra over $\Ffield$. A nonempty fiber is an affine space, so independent uniform field coefficients sample it uniformly. The distance bounds make extension across the permitted erasures unique.
\end{proof}

For a nonzero diagonal $\sigma\in(\Ffield^*)^m$, set
\[
 \CodeB^\sigma=\sigma\CodeB,\qquad
 \CodeA^\sigma=\sigma\CodeA,\qquad
 \secmap_\sigma(v)=\secmap(\sigma^{-1}v).
\]
The code $\CodeD$ is unchanged. Thus $\CodeB^\sigma*\CodeD\subseteq\CodeA^\sigma$ and $\secmap_\sigma(b*d)=\secmap_\sigma(b)*\secmap(d)$. Scaling preserves distances and reduces decoding to the unscaled code.

\begin{lemma}[Polynomial decoding with errors and erasures]
\label{lem:ag-decoding}
Let $\mathcal E\subseteq[m]$ be the erasure set and let $T$ be a public integer with $0\le T\le m$. There is a deterministic polynomial-time decoder for $\CodeA^\sigma$ that recovers any codeword with at most $T$ errors outside $\mathcal E$ whenever
\begin{equation}
 2T+|\mathcal E|<m-2\deg G-g.
 \label{eq:unified-code-radius}
\end{equation}
On other inputs it returns a codeword within that error budget if one is found, and otherwise a fixed failure value.
\end{lemma}
\begin{proof}
Undo the diagonal. Put $F=(T+g)P_\infty$ and solve
\[
 h(P_i)=z(P_i)y_i\quad(i\notin\mathcal E),\qquad
 z\in\mathcal L(F),\quad h\in\mathcal L(2G+F).
\]
If the errors are supported on $U$, $|U|\leq T$, Riemann--Roch gives a nonzero $z\in\mathcal L(F-\sum_{i\in U}P_i)$. The pair $(z,zf)$ is a solution for the transmitted function $f$. For any nonzero solution, $h-zf$ vanishes on at least $m-|\mathcal E|-T$ correct positions and has pole degree at most $2\deg G+T+g$. Equation~\eqref{eq:unified-code-radius} therefore forces $h=zf$. A nonzero solution cannot have $z=0$, because then $h$ would vanish on more than its pole degree.

At least $m-|\mathcal E|-(T+g)>2\deg G$ nonerased positions have $z(P_i)\ne0$. Interpolate $f\in\mathcal L(2G)$ at these positions, check its error count, and restore the diagonal. The solution is unique. The bases and all linear systems have polynomial size by \Cref{fact:effective-ag-tower}.
\end{proof}

\paragraph{Quantum masking by a random diagonal.}
Identify $\Ffield^m$ with $\mathbb F_2^{km}$ using a fixed basis. For $\alpha,x\in\Ffield^m$, the associated binary character is $(-1)^{\operatorname{Tr}_{\Ffield/\mathbb F_2}(\sum_i\alpha_i x_i)}$.

\begin{lemma}[Diagonal code masking]
\label{lem:diagonal-mask}
Put $\xi=6(0.047)-\log_2(64/63)>0.25$. Let $\rho_{\beta E}$ be a possibly subnormalized cq state on $\Ffield^m$ with $\hmin(\beta\mid E)\geq km-b$. After this state is fixed, sample an independent uniform $\sigma\in(\Ffield^*)^m$ and, given $\sigma$, sample $R$ uniformly in $\CodeB^\sigma_u$, independently of $\beta E$, for any prescribed secret $u\in\Ffield^q$. Then
\begin{equation}
 \bigl\|\rho_{\sigma,\beta+R,E}
       -\rho_\sigma\otimes\tau_{km}\otimes\rho_E\bigr\|_1
 \leq 2^{-(\xi m-b)/2}.
 \label{eq:diagonal-choice-mask}
\end{equation}
The bound is uniform in $u$, including a retained common classical label, included in $E$, that selects $u$ and the cq state.
\end{lemma}
\begin{proof}
For a linear code $C\subseteq\Ffield^m$ and a nonzero character index $w$, the bias of a uniform word in $\sigma C$ is the indicator of $\sigma w\in C^\perp$. Write $s=|\operatorname{supp}(w)|$ and $d_\perp=d(C^\perp)$. The event is impossible if $s<d_\perp$. Otherwise the generator matrix of $C$ restricted to that support has rank at least $d_\perp-1$. At most $64^{s-d_\perp+1}$ vectors on the support satisfy its equations, while $\sigma w$ is uniform on the $63^s$ vectors with that exact support. Hence the squared bias averaged over $\sigma$ is at most
\[
 (64/63)^m64^{-(d_\perp-1)}.
\]
For $C=\CodeB_0$, the preceding proof gives $d_\perp-1\geq0.048m$, so this is at most $2^{-\xi m}$. Translating each zero fiber to a prescribed secret fiber changes only the phase of each character. \Cref{fact:quantum-small-bias} proves the normalized-state claim.

For a subnormalized state of mass $p>0$, normalization changes conditional min-entropy by $\log_2p$. Apply the normalized bound and multiply by $p$; the resulting extra factor is $\sqrt p\leq1$. The zero-mass case is immediate. The same calculation on the classical-label blocks proves the last assertion, using Cauchy--Schwarz if their entropy bounds are expressed through the joint guessing probability.
\end{proof}

All samplers have fixed polynomial clocks. Field elements use six fair bits. To sample $\sigma_i\ne0$, reject zero at most $m$ times and then use a fixed nonzero fallback. The distance from the ideal diagonal distribution is at most $m64^{-m}$. Uniform trits and finite-set samples use the same bounded rejection convention with negligible total error. These errors are added after the ideal-sampling estimates.

\paragraph{The schedule.}
Fix $a_*=0.047$, anchoring probability $\delta_*=1/2$, and selector bounds $\rho_{\min}=1/8$, $\rho_{\max}=1/2$. Apply \Cref{lem:density-bucketing} on the density interval $[a_*/[8(k+1)],1]$, and let $\gamma,\eta_{\rm MS}>0$ be its common guessing exponent and loss threshold. Choose the fixed integer
\[
 N_0\geq
 \max\left\{166,\left\lceil\frac{200}{k\gamma}\right\rceil\right\}.
\]
For security parameter $\lambda$, prescribed leakage budget $\ell$, and binary secret capacity $q$, put
\begin{equation}
 \begin{gathered}
 s_\lambda=\left\lceil\log_2^2(\lambda+2)\right\rceil,\qquad
 m=\mathsf{Next}\bigl(\max\{100q,100\ell,s_\lambda\}\bigr),\\
 n=kN_0m,\qquad t_A=\lceil m/200\rceil,\qquad L=n+t_A .
 \end{gathered}
 \label{eq:unified-block-schedule}
\end{equation}
Use one audit coordinate per physical group, at most $\lfloor m/100\rfloor$ erased code coordinates, and decoder budget $T=\lfloor m/40\rfloor$. Set
\begin{equation}
 \begin{gathered}
 \eta_A=\min\{\eta_{\rm MS}/4,\ 1/(40kN_0)\},\qquad
 \kappa_{\rm aud}=\eta_A/2,\qquad
 \kappa_{\rm test}=\eta_{\rm MS}/8,\\
 \nu=\frac18\min\{\kappa_{\rm aud},\kappa_{\rm test}\}.
 \end{gathered}
 \label{eq:batchwise-slack}
\end{equation}
Independent honest coordinate faults of probability at most $\nu/2$ give at most $\nu L$ actual Magic Square losses except with negligible probability. Fresh hidden audit and test masks are independent of the resulting fault pattern.

For the complete OT, apply \eqref{eq:unified-block-schedule} with $q=1$ to obtain $m_0,L_0$. Each owner's mask bank has capacity $q_{\rm bank}=10L_0$. Take $m_{\rm pre}=\mathsf{Next}(100q_{\rm bank})$ for both source invocations, and use the same $N_0$ and thresholds. This meets their leakage and security lower bounds because $100q_{\rm bank}=1000L_0\geq m_0$. Write $K=kN_0$ and set
\begin{equation}
 \begin{gathered}
 C_{\rm size}=1+18000(K+1),\qquad
 c_{\rm local}=\frac1{100(K+1)},\\
 c_{\rm global}=\frac1{100(K+1)C_{\rm size}},\qquad
 \nu_{\rm global}=\frac{\nu}{C_{\rm size}} .
 \end{gathered}
 \label{eq:constant-global-rates}
\end{equation}

\begin{lemma}[Feasibility of the common schedule]
\label{lem:unified-parameters}
The schedule gives uniform polynomial-time code operations and fixed positive noise and leakage thresholds. Its label-masking, hidden-payload, audit, test, and size errors are negligible. For an invocation of capacity $q$ its physical size is $\Theta(q+\ell+\log^2\lambda)$. For the complete OT,
\[
 L_{\rm total}=2L_{\rm pre}+L_0\leq C_{\rm size}L_0
 =O(\ell+\log^2\lambda).
\]
A total physical-window budget $\ell_{\rm total}\leq c_{\rm global}L_{\rm total}$ meets every invocation's local budget, regardless of its distribution among the three invocations. A total empirical-loss bound $\nu_{\rm global}L_{\rm total}$ likewise meets each invocation's honest-noise bound.
\end{lemma}
\begin{proof}
\emph{Step 1: entropy margins.} The selector density used to fix the repetition constants is sufficient: a sensitive code character meets at least $a_*m$ code coordinates, contributing at least $a_*mN_0/8$ physical positions, while $L\leq(k+1)N_0m$. Thus the choice of $\gamma$ precedes that of $N_0$. The label bound of \Cref{lem:post-J-selector-entropy} and \Cref{lem:diagonal-mask} have slack
\[
 \xi m-2\ell-3t_A-1
 >(0.25-0.02-0.015)m-4.
\]
In \Cref{lem:ordered-block-hidden-pads}, the public raw-answer debit is at most $3km+O(1)$, the granted selected pads cost $k^2m$ bits, and the Fourier count is at most $2^{k^2m}$. The final exponent therefore has margin
\[
 \gamma L-2\ell-3km-3k^2m-O(1)
 \geq(200-0.02-18-108)m-O(1)>0
\]
for sufficiently large $\lambda$.

\emph{Step 2: auditing and decoding.} The global audit of \Cref{claim:blockwise-vanishing-audit} bounds the number $W$ of erroneous selected physical coordinates by $\eta_A n\leq m/40$, except with $\exp[-2(\eta_A-\kappa_{\rm aud})^2km]$ probability. Every erroneous code coordinate contains such a physical coordinate. Hence
\[
 2T+|\mathcal E|\leq 0.06m
 <0.15m<m-2\deg G-g,
\]
so \Cref{lem:ag-decoding} applies. For honest questions the probability of a size failure in a code coordinate is at most $24e^{-N_0/18}<1/400$. These events depend only on independent questions and audit coins on disjoint groups. A Chernoff bound gives $\exp[-\Omega(m)]$ probability of more than $m/100$ erasures. The gaps in \eqref{eq:batchwise-slack} give the same exponential audit-completeness bound for each fixed permitted fault pattern. The hidden-test estimate of \Cref{claim:sparse-test-to-random-blocks} is $\exp[-\Omega(t_A)]$.

\emph{Step 3: lengths and global budgets.} The bound for $\mathsf{Next}$ proves the stated size of a single invocation. Since $Km\leq L(m)\leq(K+1)m$,
\[
 m_{\rm pre}<9000L_0,\qquad
 L_{\rm pre}\leq9000(K+1)L_0.
\]
Both source lengths are at least $L_0$, so every batch length $L_i$ satisfies $L_i/L_{\rm total}\geq1/C_{\rm size}$. Consequently
\[
 \ell_i\leq\ell_{\rm total}
 \leq c_{\rm global}L_{\rm total}
 \leq c_{\rm local}L_i
 \leq m_i/100.
 \qquad
 \nu_{\rm global}L_{\rm total}\leq\nu L_i.
\]
This charges a single total budget. Finally, $m_i\geq s_\lambda$ makes all exponential tails and bounded sampling errors negligible. The code algorithms are polynomial by \Cref{lem:multiplicative-codes,lem:ag-decoding}; the classical cryptographic operations retain security parameter $\lambda$.
\end{proof}

\subsection{Records and the two certified predicates}
\label{app:cp-operations}
\label{sec:ot-pay}

We use the same-owner tuple form of \Cref{def:cp}: an owner-authorized check can name an ordered tuple of previously committed records and evaluate a predicate of their concatenation. Their individual commitment times and immutable contents remain fixed. This is needed to relate Bob's pre-offset question record to his later selector record. All records use fixed canonical bit encodings. Later projections of the same record reuse previously certified entries and must agree on overlaps, as required by \Cref{def:cp}.

\begin{algorithmica}[!htbp]
\protocolstyle
\emph{Roles:} one owner and one verifier, with fixed session and record identifiers and public record lengths.
\begin{protocolsteps}
\item \emph{Commitment.} The owner submits the prescribed bit string. The verifier receives only its identifier, roles and length.
\item \emph{Projection.} The owner publishes $(S,z)$ and requests $[\pi_S(M)=z]$ for that record. The verifier checks routing, syntax, overlap consistency and the success certificate.
\item \emph{Predicate.} The owner names the prescribed tuple $(M_1,\ldots,M_r)$ and public parameters $x$, and requests $\mathsf{Check}((M_1,\ldots,M_r);P,x)$. The verifier requires success for that exact tuple, predicate and public parameters.
\end{protocolsteps}
\smallskip\noindent\textit{Where each step is used.} All steps:
\Cref{def:cp,thm:masked-cp}; the predicates below are \Cref{def:selector-consistency-proof,def:payload-consistency-proof}.
\caption{The outer CP operations. Public projection claims are released before certification, including on rejection.}
\label{prot:cp-layer}
\end{algorithmica}

Use $\Ffield$, $k=6$, and a fixed $\Ftwo$-basis $e_1,\ldots,e_k$. For OT the codes have one secret coordinate. Let $n=kN_0m$ and $L=n+t_A$ as in \Cref{eq:unified-block-schedule}. A public partition divides the data positions into disjoint sets $I_{i,j}$ of size $N_0$, with $i\in[m]$, $j\in[k]$. An audit set contains one position $T_{i,j}\subset I_{i,j}$. The early symbol-erasure set is $\mathcal E\subseteq[m]$; write $I=[m]\setminus\mathcal E$. Each selector $J_{i,j}$ has its inherited coordinate order. The private bits $\beta_{i,j}$ give $\beta_i=\sum_j e_j\beta_{i,j}\in\Ffield$.

Bob's records are his two-bit-per-trit offset $U$, his three-bit columns $B$, and one fixed-length record
\[
 M_B=\operatorname{enc}\bigl((T_{i,j}),(J_{i,j}),
                    (\beta_{i,j}),\mathcal E\bigr).
\]
Audit and selector incidence masks occupy $n$ bits each. Every label, including those on erased symbols, has a committed bit. Alice's table is one record
\begin{equation}
 M_A=\operatorname{enc}((X_a,A_a,C_a)_{a\notin T_A}),\qquad |M_A|=6n.
 \label{eq:registered-alice-table}
\end{equation}
Canonical syntax means valid trits and bits, even active rows, and $(X_a,A_a)=(0,000)$ when $C_a=0$.

\begin{definition}[Selector and code-mask predicate]
\label{def:selector-consistency-proof}
For Bob's records $(U,M_B)$ and public $(V,T_A,\pi,(T_{i,j}),(J_{i,j}),\mathcal E,\sigma,M)$, $P_{\Label}$ checks their prescribed syntax and consistency with all preceding certified projections. It requires $|\mathcal E|\le\lfloor m/100\rfloor$, one audit position per microblock, and
\[
 \begin{cases}
 J_{i,j}=\varnothing,&i\in\mathcal E,\\
 J_{i,j}\subseteq I_{i,j}\setminus T_{i,j},\
 N_0/8\le |J_{i,j}|\le N_0/2,\
 U_a+V_a=\beta_{i,j}\quad(a\in J_{i,j}),&i\notin\mathcal E.
 \end{cases}
\]
Trit arithmetic is modulo three. Finally, it requires $\sigma\in(\Ffield^*)^m$ and
\[
 R=M+\beta\in\CodeB^\sigma,\qquad \secmap_\sigma(R)\in\{0,1\}.
\]
Bob authorizes $\Pi_{\rm label}=\mathsf{Check}((U,M_B);P_{\Label},x_{\rm label})$.
\end{definition}

The predicate checks inclusion and sizes. It does not certify that a selector equals a whole question class or that $\mathcal E$ follows the honest erasure rule. These conditions permit the alternative native witness used in \Cref{sec:di-bc}.

The pads use the linear family of \Cref{lem:linear-extraction} with $k=6$ output bits and public input length $|J_{i,j}|$.

\begin{definition}[Payload-consistency predicate]
\label{def:payload-consistency-proof}
For each nonerased $(i,j)$ and branch $b$, the public payload contains $S_{i,j}^b\in\Ftwo^{k\times |J_{i,j}|}$ and $F_{i,j}^b\in\Ffield$. On Alice's stored canonical table define
\begin{equation}
 Z_{i,j}^b=(A_a[b])_{a\in J_{i,j}},\quad
 K_{i,j}^b=S_{i,j}^b Z_{i,j}^b,\quad
 p_{i,j}^b=F_{i,j}^b+K_{i,j}^b.
 \label{eq:block-local-pads}
\end{equation}
The fixed basis identifies binary $k$-vectors with $\Ffield$. $P_{\Pay}$ checks all dimensions, canonical syntax, prescribed projections and the following equations. For every $i\in I$,
\[
 g_i=(p_{i,j}^0+p_{i,j}^1)/e_j
 \quad\hbox{is independent of }j,\qquad
 f_i=\sum_{j=1}^k p_{i,j}^0+M_i g_i.
\]
It requires extensions $g\in\CodeD$ and $f\in\CodeA^\sigma$, with
\[
 s_0=\secmap_\sigma(f)\in\{0,1\},\qquad
 s_1=\secmap_\sigma(f)+\secmap(g)\in\{0,1\}.
\]
Alice authorizes $\Pi_{\rm pay}=\mathsf{Check}(M_A;P_{\Pay},x_{\rm pay})$ for the exact public tuple.
\end{definition}

These extensions are unique: two extensions differ on at most $|\mathcal E|<\min\{d(\CodeA),d(\CodeD)\}$ positions. They are found by linear algebra, so both predicates are polynomial time. Checking all branch equations ties reconstruction to the stored table even when Alice chooses her payload after seeing the code mask.

\begin{corollary}[Certificates concern the same stored records]
\label{cor:cp-soundness}
A successful prescribed CP check implies its predicate on the actual stored records exactly. In particular, $\Label$ fixes a binary effective choice, and $\Pay$ fixes one input pair $(s_0,s_1)$.
\end{corollary}

\subsection{The transfer protocol}
\label{sec:blockwise-leakage-protocol}

\begin{protocolbox}
\protocolstyle
\emph{Inputs:} Alice holds $(S_0,S_1)\in\{0,1\}^2$ and Bob holds $d\in\{0,1\}$. Parameters and codes are from \Cref{sec:unified-parameters}, with one secret coordinate. All records and checks use \Cref{prot:cp-layer}. Invalid syntax or a failed prescribed opening/check causes global abort. All finite-set samplers use the bounded implementations of the parameter section. No device operation occurs after Step~4.
\begin{protocolsteps}
\item \emph{Alice.} Sample independent uniform $\Xraw,V\in\{0,1,2\}^L$ and a uniform size-$t_A$ set $T_A$. Commit its incidence string.
\item \emph{Bob.} Sample uniform $U\in\{0,1,2\}^L$ and commit its fixed two-bit trit encoding as one record.
\item \emph{Alice.} Send $V$, input $\Xraw$, and record the parity-correct rows $\Araw$ in the window shared with the next step.
\item \emph{Bob.} Input $Y=U+V$ and record the parity-correct columns $B$. Once both complete batches are recorded, close the leakage interface and commit the full $B$ record.
\item \emph{Alice.} Open $T_A$ and send a fresh uniform ordering $\pi$ of its complement, whose consecutive size-$N_0$ groups define the $I_{i,j}$.
\item \emph{Bob.} Verify $T_A,\pi$. Sample one uniform audit position in each microblock and independent fair bits $\beta_{i,j}$. Form $J_{i,j}^b=\{a\in I_{i,j}\setminus T_{i,j}:Y_a=b\}$ for $b=0,1$. Put $i\in\mathcal E$ if either size in any of its $k$ microblocks is outside $[N_0/8,N_0/2]$. Set $J_{i,j}=J_{i,j}^{\beta_{i,j}}$ for $i\notin\mathcal E$ and $J_{i,j}=\varnothing$ otherwise. Commit $M_B$ on its fixed schedule; keep all size flags private.
\item \emph{Alice.} Sample independent $C_a\gets\operatorname{Bern}(1-\delta_*)$ on all $L$ positions. For $C_a=1$ set $(X_a,A_a)=(\Xraw_a,\Araw_a)$; otherwise set $(X_a,A_a)=(0,000)$ and discard the raw controller record. Commit the data table $M_A$.
\item \emph{Bob.} Open the audit masks and the offsets and columns on $T_A$, using projections of $M_B,U,B$.
\item \emph{Alice.} Verify these openings and audit sizes. Set $\mathsf{tst}=1$ iff at most $\kappa_{\rm test}t_A$ test positions have $C_a=1$ and fail $\MS(X_a,U_a+V_a,A_a,B_a)=1$. Send $\mathsf{tst}$; on success open her canonical table on all audits.
\item \emph{Bob.} Verify the audit openings and their canonical syntax. Set $\mathsf{aud}=1$ iff $|\mathcal E|\le\lfloor m/100\rfloor$ and
\[
 \sum_{i,j}\sum_{a\in T_{i,j}}
 C_a\,\ind\{A_a[Y_a]\ne B_a[X_a]\}\le\kappa_{\rm aud} km.
\]
Send this one aggregate decision. On success open $\mathcal E$ and all selectors from the same $M_B$ record.
\item \emph{Alice.} Verify the erasure cap, empty erased selectors, and $J_{i,j}\subseteq I_{i,j}\setminus T_{i,j}$ with the stated size bounds otherwise. Send $\mathsf{sel}=1$ on success.
\item \emph{Bob.} Sample uniform $\sigma\in(\Ffield^*)^m$ and uniform $R\in\CodeB^\sigma_d$. Send $\sigma,M=\beta+R$ and authorize $\Pi_{\rm label}$.
\item \emph{Alice.} Require the valid $\Label$ certificate. Put $\Delta=S_0+S_1$, sample independently $f\gets_U\CodeA^\sigma_{S_0}$ and $g\gets_U\CodeD_\Delta$. For every $i$ sample $(Q_{i,j})_{j=1}^k$ uniformly with sum $f_i+M_i g_i$ and set $p_{i,j}^b=Q_{i,j}+b e_jg_i$. For $i\notin\mathcal E$ sample independent uniform matrices $S_{i,j}^b$ and send $F_{i,j}^b=p_{i,j}^b+S_{i,j}^bZ_{i,j}^b$, together with the matrices. Send no payload on erased symbols. Open $(X_a,C_a)$ on the selectors and authorize $\Pi_{\rm pay}$.
\item \emph{Bob.} Check both branches' public syntax and all prescribed projections, and require $\Pay$. On each selector use the raw estimate $\widehat Z_{i,j}=(C_a B_a[X_a])_{a\in J_{i,j}}$, and put
\[
 \widehat H_i=\sum_{j=1}^k
       \bigl(F_{i,j}^{\beta_{i,j}}
          +S_{i,j}^{\beta_{i,j}}\widehat Z_{i,j}\bigr)
       \quad(i\notin\mathcal E).
\]
Decode in $\CodeA^\sigma$ with erasures $\mathcal E$ and error budget $\lfloor m/40\rfloor$. Output $\secmap_\sigma(H)$ if decoding succeeds with a binary secret, and $0$ otherwise. Signal completion, revealing no decoder or selected-output flag. A corrupt Bob may instead request the late abort allowed by \Cref{def:idealOT}.
\item \emph{Alice.} Output $\emptystring$ on completion and $\perp$ on global abort.
\end{protocolsteps}
\smallskip\noindent\textit{Where each step is used.} Steps 1--4:
\Cref{lem:cp-causal-locality,lem:post-output-block-restriction}. Steps 5--7: \Cref{def:selector-consistency-proof}. Steps 8--10: \Cref{claim:sparse-test-to-random-blocks,claim:blockwise-vanishing-audit}. Steps 10--12: \Cref{lem:post-J-selector-entropy,cor:selector-hash-uniform}. Steps 13--15: \Cref{lem:committed-payload-pairing,lem:ordered-block-hidden-pads,lem:pay-success-continuation}.
\end{protocolbox}
\begingroup
\captionsetup{hypcap=false}
\captionof{algorithmica}{DI OT from CP, using private microblock selectors and global reconstruction.}
\label{prot:di-ot}
\endgroup

On an honest selected branch the field sums are
\begin{equation}
 H_i=\sum_jp_{i,j}^{\beta_{i,j}}=f_i+R_i g_i,\qquad
 H\in\CodeA^\sigma,\qquad \secmap_\sigma(H)=S_d.
 \label{eq:blockwise-Bob-output}
\end{equation}
The linear secret maps and the multiplication law of \Cref{lem:multiplicative-codes} give the last equality.

\subsection{The theorem and its simulators}
\label{sec:ot-theorem}

We state the OT reduction and give its two simulators. The following subsections prove correctness and the two security bounds.

\begin{theorem}[DI reduction from CP to OT]
\label{thm:unified-di-ot}
Under \Cref{thm:restricted-order-repetition} and the parameter schedule of \Cref{lem:unified-parameters}, \Cref{prot:di-ot} realizes $\ideal(\mathsf{OT})$ in the CP hybrid with computationally augmented, auxiliary-input-uniform security. For each admissible QPT corrupted party there is one strict-QPT simulator working for every admissible initial extension. The protocol uses one batch of $L=kN_0m+t_A$ coordinates, permits arbitrary joint device instruments and the adaptive leakage of \Cref{def:ot-admissible-leakage}, and has negligible completeness error under \eqref{eq:honest-device-fault} with $q_{\rm dev}\le\nu/2$. Its honest noise and leakage fractions can both be fixed positive constants. Sequential composition holds at \Cref{def:ot-call-boundary}.
\end{theorem}

The simulator boxes list actions in chronological order at the corresponding message boundaries; the corrupted party continues its given strategy throughout.

\begin{algorithmica}[!htbp]
\protocolstyle
\emph{Corrupt Alice; SIM plays Bob.} SIM logs Alice's classical commitment inputs and preserves her quantum workspace and reference.
\begin{protocolsteps}
\item \emph{Physical prefix.} Execute the honest Bob interface, masks, tests and projections of Steps 1--11. Relay public aborts.
\item \emph{Dummy choice.} Sample the honest diagonal and $R\gets_U\CodeB^\sigma_0$, send $M=\beta+R$, and provide the genuine true $\Label$ certificate.
\item \emph{Ideal pair and delivery.} From Alice's stored table and the received payload, evaluate $\Pay$ by linear algebra. On success obtain its unique $(f,g)$ and submit $(\secmap_\sigma(f),\secmap_\sigma(f)+\secmap(g))$ to ideal OT. Reproduce public delivery or abort and retire honest scratch. SIM does not use its dummy physical decoder as the ideal honest output.
\end{protocolsteps}
\smallskip\noindent\textit{Where each step is used.} Step 1:
\Cref{claim:blockwise-vanishing-audit,lem:post-J-selector-entropy}. Steps 2--3: \Cref{thm:receiver-leakage-simulation}.
\caption{The corrupt-sender OT simulator.}
\label{prot:di-ot-sim}
\label{prot:di-ot-sim-a}
\end{algorithmica}

\begin{algorithmica}[!htbp]
\protocolstyle
\emph{Corrupt Bob; SIM plays Alice.} SIM logs Bob's CP inputs and uses the genuine honest physical interface.
\begin{protocolsteps}
\item \emph{Physical prefix.} Execute Steps 1--12 and their prescribed checks. Relay every public rejection.
\item \emph{Ideal choice.} At a valid $\Label$ compute $R=M+\beta$ from the logged records and $d=\secmap_\sigma(R)\in\{0,1\}$. Submit $d$ to ideal OT and receive $S_d$.
\item \emph{Payload and delivery.} Use the ordinary honest payload algorithm with dummy pair $(S_d,S_d)$ and the same canonical table. Provide its true $\Pay$ certificate. Relay Bob's possible late abort, or complete Alice's empty output, and retire honest scratch.
\end{protocolsteps}
\smallskip\noindent\textit{Where each step is used.} Steps 1--2:
\Cref{lem:cp-causal-locality,cor:cp-soundness}. Step 3: \Cref{lem:ordered-block-hidden-pads,lem:pay-success-continuation,thm:sender-leakage-simulation}.
\caption{The corrupt-receiver OT simulator.}
\label{prot:di-ot-sim-b}
\end{algorithmica}

\subsection{The causal cut and global sampling}
\label{app:deferred-main-proofs}
\label{app:ot-correctness}

We first place the frozen records in the game of the repetition theorem, then use sampling to bound discrepancies in Alice's table.

\begin{lemma}[CP records preserve the causal cut]
\label{lem:cp-causal-locality}
Before Alice fixes $M_A$, Bob's accepted CP inputs immutably determine $U,B$, every audit and selector, $\mathcal E$, and the full label vector. Metadata-only receipts convey no Alice-side source to this early instrument. Later certified projections refer to those same values and act on neither devices nor the reference.
\end{lemma}
\begin{proof}
Retain an inaccessible classical copy of each actual commitment input. The receipt is simulated locally from its prescribed metadata. Before $M_A$ no source-dependent Alice predicate is requested. Rejected commands are totalized with defaults without conditioning on success. Later operations cannot change the saved copy.
\end{proof}

\begin{lemma}[Embedding in the whole-batch game]
\label{lem:post-output-block-restriction}
After causal communication elimination, the physical batch and every selector determined by the frozen records obey the interface of \Cref{thm:restricted-order-repetition}.
\end{lemma}
\begin{proof}
Assign the effective offset $U$ at its pre-$V$ cut, mapping invalid trits to zero for analysis only. Fresh uniform $V$ gives $\rho_{YUQ_0}=\tau_Y\otimes\rho_{UQ_0}$ for $Y=U+V$, where $Q_0$ is the prior laboratory/reference state and excludes Alice's saved $V$. Fresh $X^{\rm in}$ is independent. Dave knows $U$, receives $Y$ and reconstructs $V$; Charlie's physical instrument reads only $X^{\rm in}$. The hidden test and partition are independent trusted coins, revealed at their prescribed times. Every selector freezes before any $C$ is used, and only canonical $(C,X)$ is free late input. Later answers and the test decision are guessed at their first disclosures. Malformed column records are replaced in the game by odd-parity columns; this preserves every originally winning coordinate. Invalid selector branches use a legal fallback and are excluded from the target event. One game contains all $L$ coordinates, so no question-dependent other-block output is initial advice.
\end{proof}

\begin{claim}[The sparse test bounds global losses]
\label{claim:sparse-test-to-random-blocks}
Let $\mathsf{Lose}$ be the positions with $C_a=1$ and a Magic Square loss against Bob's fixed effective questions and columns. For $\kappa_{\rm test}<\eta_{\rm MS}$,
\[
 \Pr[\mathsf{tst}=1,\ |\mathsf{Lose}|>\eta_{\rm MS}L]
 \le \exp[-2(\eta_{\rm MS}-\kappa_{\rm test})^2t_A].
\]
\end{claim}
\begin{proof}
Privately presample unused anchor coins. Because the test receipt reveals no contents, sample $T_A$ after the complete raw records and Bob's $B$ input are fixed. Those records and the coins determine the loss set. Sampling without replacement gives the displayed bound. Invalid openings only reduce acceptance. No per-microblock loss bound is asserted.
\end{proof}

\begin{claim}[One global audit controls every later opening]
\label{claim:blockwise-vanishing-audit}
In ideal CP, or in the equivocal-Bob native experiment with one statistically binding Alice table, let $W$ count malformed, noncanonical or actively inconsistent data positions in that table. For $\eta_A>\kappa_{\rm aud}$,
\begin{equation}
 \Pr[\mathsf{aud}=1,\ W>\eta_A n]
 \le \exp[-2(\eta_A-\kappa_{\rm aud})^2km]+\varepsilon_{\rm bind}.
 \label{eq:global-native-audit}
\end{equation}
The resulting selected-raw-error bound holds simultaneously for every accepted native opening of the same table.
\end{claim}
\begin{proof}
\emph{Independence.} At the raw-output cut the fresh audit masks are independent of the joint state containing frozen $Y,B$ and Alice. Bob's private classical computation of labels, selectors and size flags can be deferred. Before table freezing, ideal receipts do not read their contents; in the native experiment the simulated commitment channels do not read their target values. Bob's receiving algorithm for Alice's table uses fresh coins, not its audit mask. Thus the entire pre-table channel is independent of the mask, jointly with $Y,B$. This is stronger than independence of Alice's marginal alone.

\emph{Sampling.} Fix $Y,B$ and the unique table outside binding failure. The $km$ sampled bad-entry indicators are independent and their mean sum is $W/N_0$. For $W>\eta_A kN_0m$, acceptance requires a downward deviation of $(\eta_A-\kappa_{\rm aud})km$. Hoeffding gives \eqref{eq:global-native-audit}. A sampled malformed opening rejects. The cap on $\mathcal E$ is imposed only as a further restriction; the sampling argument never conditions on $\mathcal E$. Binding gives one possible full table, so the bound covers all openings without a union over witnesses.
\end{proof}

\begin{lemma}[Global transcript reconstruction]
\label{lem:blockwise-transcript-reconstruction}
On the good-audit event and a successful $\Pay$ check, honest Bob's decoded output is the $d$th component of $\Pay$'s unique pair.
\end{lemma}
\begin{proof}
For each selected position, $C_aB_a[X_a]$ equals the committed $A_a[\beta_{i,j}]$ unless that position is counted in $W$. Correct raw inputs give the same pad for every public hash matrix. Consequently at most $W$ nonerased symbol sums are incorrect. $\Pay$ and $\Label$ imply that the correct sums extend to $H=f+R*g\in\CodeA^\sigma$ with secret $s_d$. The parameter bounds give $2W+|\mathcal E|\le 0.06m<0.15m$, within \Cref{lem:ag-decoding}. Decoding returns that unique word.
\end{proof}

\begin{lemma}[The committed payload pair]
\label{lem:committed-payload-pairing}
From Alice's logged table and accepted public payload one can compute in polynomial time the pair $(s_0^*,s_1^*)=(\secmap_\sigma(f),\secmap_\sigma(f)+\secmap(g))$. The computation uses no private receiver record or choice.
\end{lemma}
\begin{proof}
Recompute both pads, obtain $g_i,f_i$ on $I$, and solve their linear code-extension equations. Their uniqueness was established in \Cref{def:payload-consistency-proof}. $\Pay$ checks both recovered secrets.
\end{proof}

\begin{claim}[Honest correctness]
\label{claim:blockwise-leakage-correctness}
Under \eqref{eq:honest-device-fault} with $q_{\rm dev}\le\nu/2$, honest parties output $(\emptystring,S_d)$ except with negligible probability.
\end{claim}
\begin{proof}
A class in a microblock has distribution $\operatorname{Bin}(N_0-1,1/3)$. Its two tails outside $[N_0/8,N_0/2]$ are each bounded by $e^{-N_0/18}$. The probability of a bad symbol is at most $24e^{-N_0/18}<1/400$. Disjoint trusted questions and audit coins make the symbol indicators independent, so more than $m/100$ erasures has probability $e^{-\Omega(m)}$. This erasure estimate uses only the trusted questions and audit coins. Equation~\eqref{eq:honest-noise-promise} bounds the total fault count. Conditional on any permitted fault pattern, hidden-test and global-audit upper-tail bounds give acceptance except with $e^{-\Omega(t_A)}+e^{-\Omega(m)}$. The total fault bound and parameter slack put the selected sums within the global decoding radius. Honest predicates are true; \eqref{eq:blockwise-Bob-output} gives the output.
\end{proof}

\subsection{Choice hiding and the corrupt-sender simulation}
\label{sec:Bob-to-Alice-leakage}

We bound the entropy of Bob's label vector after selector disclosure. Code masking then gives choice hiding, which combines with the audit to prove the corrupt-sender simulation.

\begin{lemma}[Leakage-robust post-selector entropy]
\label{lem:post-J-selector-entropy}
Before the fresh diagonal, the totalized ideal-CP experiment satisfies
\begin{equation}
 \pg(\beta\mid E_A)\le
       2^{\,2\ell_{\sf A}+3t_A+1-km},
 \label{eq:selector-global-entropy}
\end{equation}
where $E_A$ includes Alice's entire state, public transcript, submitted-record log and untouched reference. The same absolute bound holds on the continuing substate. It also holds in the native equivocal-Bob experiment with target-independent record simulations.
\end{lemma}
\begin{proof}
\emph{Totalization and disclosures.} Assign canonical test/partition defaults on invalid commands, sample every label even on aborted or erased branches, and keep private size flags undisclosed. Remove the physical communication by \Cref{lem:interactive-transcript-guessing}. At the first release of Bob's $3t_A$ tested answer bits replace them by an independent guess. Correct guesses reproduce the complete joint continuation, costing $2^{2\ell_{\sf A}+3t_A}$ in probability. Do not condition the reduced strategy on guess correctness.

\emph{Symmetry.} In this reduced experiment all actual Bob outcomes stay private. Whole-laboratory no-signalling makes Alice and her reference independent of the remaining uniform questions. Questions on the already fixed test set may be conditioned on freely. The partition precedes tested-answer disclosure. Within each microblock swap question labels $0,1$ outside its audit and flip the independent $\beta_{i,j}$. The selected set, the two-candidate size condition and the whole erasure mask remain unchanged. All labels are therefore uniform jointly with Alice's view before the one aggregate decision; appending that bit costs a factor of two. Alice's own table, $X,C$ disclosures and submitted-input log are channels of the same view. Undo the guessing factors and restrict the positive state to actual continuation without normalization.
\end{proof}

\begin{corollary}[The code mask hides the choice]
\label{cor:selector-hash-uniform}
Put $b_A=2\ell_{\sf A}+3t_A+1$. For each fixed binary choice $d$, the continuing state after $M=\beta+R_d$ differs from a uniform $M\in\Ffield^m$ independent of $E_A$, with the fresh diagonal retained, by at most
\[
 \varepsilon_m=2^{-(\xi m-b_A-O(1))/2}
\]
in trace norm, plus the negligible bounded-sampling error.
\end{corollary}
\begin{proof}
Apply \Cref{lem:diagonal-mask} to \eqref{eq:selector-global-entropy} and the prescribed-secret fiber $\CodeB^\sigma_d$.
\end{proof}

\begin{theorem}[Simulation against corrupt Alice]
\label{thm:receiver-leakage-simulation}
\Cref{prot:di-ot-sim-a} simulates the complete corrupt-Alice execution, including the honest input and output and the reference, with error at most $\varepsilon_a+2\varepsilon_m+\negl(\lambda)$, where $\varepsilon_a$ is \eqref{eq:global-native-audit}.
\end{theorem}
\begin{proof}
Fix the external honest choice $d_{\rm ext}$, retained only as a comparison label. In the genuine endpoint, replace the physical decoder's output by $s_{d_{\rm ext}}^*$, computed from the final $\Pay$ log. \Cref{lem:blockwise-transcript-reconstruction} makes the outputs equal outside mass $\varepsilon_a$. This replacement does not condition the masking state on the audit good event.

After it, every public receiver action depends on public syntax, the transcript and CP results, not on private $Y,B,\beta$. $\Label$ is true at both genuine choice endpoints, so its success symbol has the same continuation there. Apply \Cref{cor:selector-hash-uniform} to compare the choice-$d_{\rm ext}$ endpoint to a uniform analytical mask and then to choice zero. Use the same witness-free continuation and the same output function $s_{d_{\rm ext}}^*$ in these comparisons. The middle comparison does not request a CP proof of a possibly false $\Label$ statement. The final endpoint has a genuine true $\Label$ and is the simulator's execution. Ideal OT chooses the indicated component of the submitted pair. Retiring the unused private decoder state preserves the complete compared marginal, including the reference and public aborts.
\end{proof}

\subsection{Global Fourier hiding and the corrupt-receiver simulation}

The next lemma also supplies the packed source commitment. For that use the code secret map has $q$ coordinates, and $\Label$ requires $\secmap_\sigma(R)=d\mathbf1^q$. A common classical label specifying candidate inputs is permitted if it belongs to the admissible pre-input state, independent of fresh game questions, or is a function of information already available to Bob or explicitly charged below. An arbitrary label computed from Alice's undisclosed raw source is not granted by this statement.

For the source version, $\Mode$ denotes the consistency proof of \Cref{blk:def:mode}. Its witness opens Bob's committed offsets and labels, matches every selector to its question class, and makes the public mask minus the witnessed label vector a zero-secret codeword.

\begin{table}[!htbp]
\centering
\begin{tabular}{p{0.34\textwidth}p{0.43\textwidth}r}
\hline
\textbf{Game quantity or extra information}&
\textbf{Protocol step}&\textbf{Charge}\\
\hline
Uniform questions and frozen columns&Steps 1--4&$0$\\
Early selected set&Steps 5--6, frozen before Step 7&$0$\\
Late canonical $(C,X)$&Steps 9 and 13; anchors reveal no raw $X$&$0$\\
\hline
Physical-window communication&Steps 3--4&$2\ell_{\sf B}$\\
Alice's audit answers&Step 9, three bits at each of $km$ positions&$3km$\\
Alice's test decision&Step 9&$1$\\
Selected pads, granted for analysis&Step 13&$k^2m$\\
\hline
\end{tabular}
\caption{The whole-batch game correspondence and entropy debits. All source-dependent disclosures are charged at their first release.}
\label{tab:game-correspondence}
\end{table}

\begin{lemma}[Global hiding of the unselected payload]
\label{lem:ordered-block-hidden-pads}
In ideal CP, or in the native equivocal-Alice experiment with binding Bob records and a sound real $\Mode$ proof, compare two ordinary honest payload constructions with the same effective selected input. They may depend on the permitted common label and Bob's classical prefix. Their complete receiver-reference states before $\Pay$ differ in trace distance by at most
\begin{equation}
 \varepsilon_f=\varepsilon_{\rm bad}
   +2^{(3k^2m+1-\gamma L+2\ell_{\sf B}+b_{\rm pub})/2},
 \qquad b_{\rm pub}=3km+O(1).
 \label{eq:leakage-ot-regime}
\end{equation}
Here $\varepsilon_{\rm bad}$ is the single global test, binding and false-$\Mode$ exceptional mass in that statistical experiment. The estimate has no factor depending on the number $q$ of secrets.
\end{lemma}
\begin{proof}
\emph{One good substate.} Delete the common exceptional mass before communication guessing or Fourier expansion. The remaining subnormalized state $\widehat\rho$ retains all immutable classical records and the reference. For a fixed ambient Fourier index $\alpha\in\Ffield^{km}$, retain the early branch on which it is zero outside $I$ and
\[
 w_i=\sum_j e_j\alpha_{i,j}\ (i\in I),\quad
 w_i=0\ (i\notin I),\qquad
 w\in\CodeD_0^\perp\setminus\CodeD^\perp .
\]
Dual distance gives $\wt(w)>0.047m$. The union $K_\alpha$ of selectors with $\alpha_{i,j}\ne0$ contains at least $0.047mN_0/8$ positions. It depends only on the frozen selectors, labels and erasure set and the fixed index $\alpha$, not on the later diagonal, mask or payload. It is a legal early selector. The density interval $[0.047/(8(k+1)),1]$ and \Cref{lem:density-bucketing} give uniform constants for every possible size of this selector. Let $X_\alpha$ concatenate the canonical complementary bits $A_a[1-\beta_{i,j}]$ on these selectors, in the fixed selector order. Let $E_0$ contain Dave's residual state, his frozen records, the permitted common label, the untouched reference, and the late canonical questions, before the charged source disclosures and extractor pads. At a disclosure cut, the communication-elimination strategy uses a guess until the corresponding debit is charged below. Guessing $X_\alpha$ reconstructs Alice's full row at each active winning position, using $B_a[X_a]$ and row parity. The selector theorem and causal communication elimination yield
\[
 \pg(X_\alpha\mid E_0)_{\widehat\rho,\mathrm{sensitive}}
 \le2^{2\ell_{\sf B}-\gamma L}.
\]
This is a family of guessing inequalities on one state, not a union of separately smoothed statements. The game is unconditioned on later proof acceptance; invalid or nonsensitive early branches use legal fallback selectors and are excluded from the target event. The permitted input label is included in the initial state or in already charged receiver information, so the same game bound applies.

\emph{Disclosures and fresh seeds.} Charge the messages in \Cref{tab:game-correspondence}. Adaptive later $\sigma,M$ are postprocessing after the audit charge. Grant the selected seeds and pads $O_{i,j}=K_{i,j}^{\beta_{i,j}}$, costing $k^2m$ bits. The complementary seeds are still independent. Thus
\[
 h=\gamma L-2\ell_{\sf B}-b_{\rm pub}-k^2m
\]
lower-bounds the hidden concatenation's absolute conditional min-entropy. For nonzero $\alpha_{i,j}$, the functional $\operatorname{Tr}_{\Ffield/\Ftwo} (\alpha_{i,j}S_{i,j}^{1-\beta_{i,j}}Z_{i,j}^{1-\beta_{i,j}})$ is a uniform one-bit linear hash of that raw string. Combining its independent seed coordinates over the support gives a uniform one-bit hash of $X_\alpha$. \Cref{lem:linear-extraction} bounds the unnormalized operator Fourier coefficient by $2^{(1-h)/2}$, with all public seeds retained. For variable lengths and classical prefixes apply the hash bound blockwise and Cauchy--Schwarz; the global guessing bound suffices.

\emph{Which coefficients depend on the input.} Given the granted selected pads, the ciphertext pair is equivalently
\[
 U_{i,j}=p_{i,j}^{\beta_{i,j}},\qquad
 V_{i,j}=K_{i,j}^{1-\beta_{i,j}}+e_jg_i.
\]
For fixed legal $R,g$, translation $f\mapsto H=f+R*g$ maps $\CodeA^\sigma_{s_0}$ bijectively to $\CodeA^\sigma_{s_d}$. Consequently $H$, and its uniform selected splits $U$, are independent of $g$ and the raw data conditional on the prefix. Their laws agree for the two candidate pairs. The character of the added $g$ mask is indexed by $w$. Averaging uniform $g\in\CodeD_\Delta$ annihilates it unless $w\in\CodeD_0^\perp$; if $w\in\CodeD^\perp$ its phase is input independent and cancels. Only the sensitive coefficients above remain. Omitting erased ciphertexts is essential: their indices extend by zero directly into the full dual, since $(\CodeD|_I)^\perp=\{w_I:\operatorname{ext0}_I(w_I)\in\CodeD^\perp\}$.

\emph{Fourier inversion.} The inverse factor $|\Ffield|^{-k|I|}$ cancels when summing trace norms over the ciphertext values. Each sensitive difference has norm at most $2\cdot2^{(1-h)/2}$; the trace-distance factor is $1/2$. The single ambient space contains at most $|\Ffield|^{km}=2^{k^2m}$ indices, covering every early erasure pattern without another pattern count. Their sum is the exponential term in \eqref{eq:leakage-ot-regime}. Restore the one deleted bad mass once. The common-label and prefix argument is pointwise in the selected-payload algebra and uses a global entropy bound, so it introduces no $2^q$ or $|\Ffield|^q$ factor.
\end{proof}

\begin{lemma}[The $\Pay$ success continuation]
\label{lem:pay-success-continuation}
If two honest payload endpoints use the same canonical table and ordinary valid input pairs, their pre-$\Pay$ comparison extends through $\Pay$, delivery and the corrupted party's subsequent processing.
\end{lemma}
\begin{proof}
$\Pay$ is true at both endpoints by construction. Its prescribed ideal certificate is therefore the same success symbol on their respective public tuples; scheduling and possible abort are one channel on the corrupt state and transcript. All syntax and projection checks concern the same table. Apply this channel to the compared states and retain the honest output interface and reference.
\end{proof}

\begin{theorem}[Simulation against corrupt Bob]
\label{thm:sender-leakage-simulation}
\Cref{prot:di-ot-sim-b} simulates the complete corrupt-Bob execution with trace-distance error at most $\varepsilon_f+\negl(\lambda)$.
\end{theorem}
\begin{proof}
$\Label$ determines $d$ before fresh branch seeds. Its extraction uses Bob's existing committed records and public mask, so $d$ and the dummy pair satisfy the prefix condition of \Cref{lem:ordered-block-hidden-pads}. That lemma compares the real pair with $(S_d,S_d)$ while retaining Bob's complete device residual, quantum memory, transcript and reference. \Cref{lem:pay-success-continuation} appends the same success continuation and permits Bob's late abort after its ideal release.
\end{proof}

\begin{proof}[Proof of \Cref{thm:unified-di-ot}]
\Cref{claim:blockwise-leakage-correctness} gives correctness. \Cref{thm:receiver-leakage-simulation,thm:sender-leakage-simulation} give the two strict-QPT simulations, using only the declared physical interface and polynomial classical code operations. \Cref{lem:unified-parameters} makes their statistical errors negligible. The comparisons include the entire corrupt continuation, honest input/output and untouched reference. Apply \Cref{thm:sequential-composition} at the prescribed retirement boundaries.
\end{proof}

The statistical proofs above also apply to the native experiments used below. Against a corrupt receiver, equivocate Alice's test-mask and table commitments before deferring the test; leave Bob's binding records and $\Mode$ real. Against a corrupt committer, equivocate Bob's records while receiving Alice's table with the honest binding algorithm. Sequential primitive errors are added to the final comparison and are not multiplied by guessing factors or Fourier counts. This preserves the common bad-event and all-native- opening guarantees needed by the source commitment.

\section{Commit-and-prove from one committed mask string}
\label{sec:bc-cp-compiler}

The OT protocol of \Cref{sec:leakage-ot} uses immutable records and owner-authorized predicate checks. We realize this interface with one preprocessing commitment in each direction. Each subsequent record uses a fresh interval of the committed random string; ordinary post-quantum zero knowledge certifies its predicates. Only the preprocessing commitment uses devices. \Cref{sec:di-bc} constructs it in one window.

\subsection{The string-commitment interface}
\label{sec:bc-contract}

The compiler uses the extractable string commitment of \Cref{def:vector-commitment}. Its extraction guarantee preserves the continuing state, and its native-opening predicate checks a public transcript against a candidate string and a classical witness.

The quantifier over all witnesses in \eqref{eq:vector-all-openings} will make an inconsistent predicate a false NP statement. Agreement only against efficiently produced native openings would not supply this implication. The property is required in the extraction experiment, where the proof will use it.

\subsection{The masked record protocol}
\label{sec:compiler-statement}

For a mask bank, a public capacity $q$ bounds the total length of new records in one owner direction. Records may be created after arbitrary admissible outer computation. Their lengths, roles, and identifiers are public. All accesses to the honest bank's private state use the following interface.

\begin{protocolbox}
\protocolstyle
\noindent Alice owns the records; Bob verifies. Public capacity is $q$. Use the string commitment of \Cref{def:vector-commitment} and the ordinary classical-message ZK proof of \Cref{lem:npbc-pqzk}. A rejected check or malformed command causes terminal session abort.
\begin{protocolsteps}
\item \emph{Alice (preprocessing).} Sample $\mu\leftarrow_U\Ftwo^q$ and complete a string commitment to $\mu$, with Bob as receiver. Save its transcript $\tau$, native witness $w$, and an allocation counter initially zero. Abort if preprocessing fails.
\item \emph{Alice (new record).} On input $M_a\in\Ftwo^{n_a}$, allocate the next unused interval $I_a\subseteq[q]$ of length $n_a$. Abort on insufficient capacity. Send the tagged tuple
\[
 (a,I_a,c_a),\qquad c_a=M_a\oplus\mu_{I_a}.
\]
\item \emph{Bob (receipt).} Check the identifier, role, prescribed length, and next unused interval. Store $(a,I_a,c_a)$, advance the counter, and send the metadata-only receipt.
\item \emph{Alice (predicate).} Publish an authorized request $(\mathbf a,P,y)$ for previously receipted same-owner records $\mathbf a=(a_1,\ldots,a_r)$. Abort unless their handles, roles, lengths, saved intervals, and schedule are valid. If $P(M_{a_1}\|\cdots\|M_{a_r},y)=0$, abort. Otherwise prove
\begin{equation}
 \begin{gathered}
 \exists u,w:\quad \Ver_{\rm str}(\tau,u,w)=1\quad\land\\
 P\bigl((c_{a_1}\oplus u_{I_{a_1}})\|\cdots\|
        (c_{a_r}\oplus u_{I_{a_r}}),y\bigr)=1
 \end{gathered}
 \label{eq:masked-cp-relation}
\end{equation}
using the same bank witness $(\mu,w)$ for every component. A projection is the singleton request with $P(M;S,z)=[\pi_S(M)=z]$.
\item \emph{Bob (certificate).} Check routing and syntax against the saved tuple $(\tau,(a_h,I_{a_h},c_{a_h})_{h=1}^r)$, including each immutable interval and record length, and verify the proof. Caller-supplied replacement intervals or ciphertexts are rejected. Previously certified projections agree on overlapping positions of the same handle. Delivery and withholding follow \Cref{def:cp}.

\end{protocolsteps}
\noindent Steps~2--5 may repeat on the public polynomial schedule. Intervals are disjoint and immutable. The bank is never natively opened; all its private state is discarded at terminal return.
\par\smallskip\noindent\emph{Where each step is used.} Step~1: \Cref{def:vector-commitment}; Steps~2--3: \Cref{def:cp}; Steps~4--5: \eqref{eq:masked-cp-relation} and \Cref{thm:masked-cp}.
\end{protocolbox}
\par\nobreak\noindent
\begin{minipage}{\linewidth}
\captionsetup{hypcap=false}
\captionof{algorithmica}{Masked commit-and-prove with one preprocessing commitment per owner. All record and predicate phases are classical.}
\label{prot:masked-cp}
\end{minipage}

A record opening can reveal its consumed mask positions through $c_a$. The unopened bank witness and unused intervals remain private. The ZK proof's internal commitments are the ordinary classical commitments of \Cref{fact:pqowf-np-anchors}; they do not use the bank. Capacity therefore counts record bits, not proof messages. Reusing a record in joint predicates consumes no further mask positions.

\begin{theorem}[Masked commit-and-prove]
\label{thm:masked-cp}
Assume qOWFs and \Cref{def:vector-commitment}. \Cref{prot:masked-cp} realizes \Cref{def:cp} for adaptive sequential phases with total new-record length at most $q$ in each direction. It has auxiliary-input-uniform QPT simulation for either static corruption, retaining aborts and the complete continuing state. Each direction uses one source commitment. All later honest operations and messages are classical and make no source or device calls.
\end{theorem}

\begin{protocolbox}
\protocolstyle
\noindent Public state consists of $\tau$ and the allocated records $(a,I_a,c_a)$. Each simulator preserves the entire corrupt workspace and the untouched reference. Map real rejection to ideal session abort.
\par\smallskip\noindent\emph{Corrupt owner; SIM plays the verifier.}
\begin{protocolsteps}
\item \emph{SIM (preprocessing).} Run source extraction and retain $u^*$. Put $\bar u^*=u^*$ if $u^*\in\Ftwo^q$ and $\bar u^*=0^q$ otherwise.
\item \emph{SIM (record).} After the honest syntax checks, set $M_a^*=c_a\oplus\bar u^*_{I_a}$ and submit it to ideal CP before releasing the receipt.
\item \emph{SIM (predicate).} Submit the public tuple request to ideal CP before emulating the honest verifier. Delay ideal delivery during the emulation and use the records fixed at their respective receipts. Release success only when the proof accepts and the ideal predicate is true; otherwise reproduce rejection or withholding with the ideal delivery commands.
\end{protocolsteps}
\noindent\emph{Corrupt verifier; SIM plays the owner.}
\begin{protocolsteps}
\item \emph{SIM (preprocessing).} Commit honestly to $0^q$, using only the declared honest-device interface of the source commitment.
\item \emph{SIM (record).} On an ideal receipt of public length $n_a$, allocate the prescribed interval and send independent uniform $c_a\leftarrow_U\Ftwo^{n_a}$.
\item \emph{SIM (predicate).} Relay the ideal public tuple request. For an authorized success, run the fixed ZK-simulation circuit on the public instance \eqref{eq:masked-cp-relation}. For ideal rejection, reproduce the owner's abort. Keep the same bank transcript throughout.
\end{protocolsteps}
\par\smallskip\noindent\emph{Where each step is used.} Corrupt owner, Steps~1--3: \eqref{eq:vector-all-openings} and the soundness step of \Cref{thm:masked-cp}; corrupt verifier, Steps~1--3: its witness-removal, hiding, and unused-mask steps.
\end{protocolbox}
\par\nobreak\noindent
\begin{minipage}{\linewidth}
\captionsetup{hypcap=false}
\captionof{algorithmica}{Persistent simulators for masked CP. Each numbered list follows the phases of one owner direction.}
\label{prot:masked-cp-sim}
\end{minipage}

\subsection{Proof of the masked-CP theorem}
\label{sec:masked-cp-proof}

\begin{proof}[Proof of \Cref{thm:masked-cp}]
Fix the adversary and a public polynomial bound $T$ on predicate phases. The compared state contains the caller's prescribed input--output interface, public transcript, corrupt workspace, and reference. The honest bank internals remain encapsulated as in \Cref{def:cp}.

\emph{Step 1: a corrupt owner's records.} Replace the preprocessing receiver by the source extractor. Item~(i) of \Cref{def:vector-commitment} preserves the joint state. Every later honest verification uses the public transcript and fresh coins, so the comparison extends through the continuation. At each record receipt the simulator can submit
\begin{equation}
 M_a^*=c_a\oplus\bar u^*_{I_a},\qquad
 \bar u^*=\begin{cases}u^*,&u^*\ne\perp,\\0^q,&u^*=\perp.
 \end{cases}
 \label{eq:masked-record-extract}
\end{equation}
This fixes the ideal message when the real receipt is issued.

\emph{Step 2: a corrupt owner's predicates.} Outside the bad event of \eqref{eq:vector-all-openings}, every native witness has $u=u^*\ne\perp$. Thus all components of any requested tuple recover the values fixed at their receipts:
\[
 (c_{a_1}\oplus u_{I_{a_1}})\|\cdots\|
 (c_{a_r}\oplus u_{I_{a_r}})
 =M_{a_1}^*\|\cdots\|M_{a_r}^*.
\]
If the ideal tuple predicate is false, the NP statement \eqref{eq:masked-cp-relation} is false. When $u^*=\perp$, there is no native witness at all. Statistical soundness from \Cref{lem:npbc-pqzk}, uniformly over false instances and arbitrary auxiliary states, gives

\begin{equation}
 \Pr_{\rm sim}[\text{some false predicate on the extracted records}
                 \text{ is certified}]
 \le\varepsilon_{\rm ao}+T\varepsilon_{\rm snd}.
 \label{eq:masked-cp-soundness}
\end{equation}
This argument covers adaptively chosen tuples and pays the single all-opening event once for the bank. The simulator publishes a request through ideal CP before emulating its proof, so rejected projection claims remain public. It delays delivery during the emulation and reproduces refusal, withholding, and terminal abort. Outside \eqref{eq:masked-cp-soundness}, coupling the verifier's coins gives identical transcripts and corrupt states. The existential bad event is bounded inside the extraction experiment, not transferred through a computational indistinguishability claim.

\emph{Step 3: remove honest witnesses before using hiding.} For a corrupt verifier, let $H_0$ be the real execution. In $H_1$, keep the commitment to $\mu$ and replace every honest-owner proof by the ZK simulator of \Cref{lem:npbc-pqzk}. Each replaced statement is true: one saved bank witness serves every component of the owner's true tuple predicate, and item~(iv) of \Cref{def:vector-commitment} supplies its native witness on every continuing branch. Sequential auxiliary-input ZK gives $H_0\approx_c H_1$. After these replacements, the continuation does not read the native witness.

In $H_2$, change the preprocessing commitment from $\mu$ to $0^q$, retaining a separate genuine uniform $\mu$ for the ciphertexts. The same strictly polynomial ZK-simulation circuit runs on each syntactically valid public instance, including false ones. Its simulation guarantee was used only in the preceding true-instance replacement. The complete suffix is therefore an efficient channel that does not inspect the honest commitment witness. Labelled joint hiding, with $\mu$ retained privately by the reduction, gives $H_1\approx_c H_2$. The reduction may retain the actual caller and its inputs, although the final simulator need not know those inputs.

\emph{Step 4: unused masks in the dummy-bank experiment.} Immediately before allocating a fresh interval $I_a$, its mask satisfies
\begin{equation}
 \rho^{\mu_{I_a}M_aER}
 =\tau_{\mu_{I_a}}\otimes\rho^{M_aER}.
 \label{eq:unused-mask-invariant}
\end{equation}
Here $E$ contains the preceding exposed state and transcript and all admissible continuation registers, but not private bank internals. For the first record, the dummy commitment is independent of $\mu$. Inductively, previous ciphertexts use only consumed intervals; requests and predicate results concern existing records; simulated proofs read only public instances and the verifier's state; and the caller has no access to unused masks. These operations preserve \eqref{eq:unused-mask-invariant}, including for an adaptive next input. Writing $\rho^{M_aER}=\sum_x|x\rangle\langle x|\otimes\rho_x^{ER}$,
\begin{align*}
 \rho^{c_aM_aER}
 &=2^{-n_a}\sum_{v,x}|x\oplus v,x\rangle\langle x\oplus v,x|
                      \otimes\rho_x^{ER}\\
 &=\tau_{c_a}\otimes\rho^{M_aER}.
\end{align*}
Consequently, the simulator may sample each $c_a$ uniformly, using only its public length. Later requests and outcomes are supplied by ideal CP. Even a full record opening reveals only consumed mask positions, so the induction continues. This identifies $H_2$ with the corrupt-verifier simulator of \Cref{prot:masked-cp-sim}.

\emph{Step 5: both directions and outer continuations.} Preprocess the two banks sequentially. First replace the honest owner's predicate proofs and switch its bank to the dummy commitment as in Steps~3--4. Then replace the corrupt owner's preprocessing by extraction and use Steps~1--2. At each replacement all other pending honest records and future device shares remain untouched references. The full suffix includes the caller's prescribed adaptive classical operations and any intervening declared device invocation. Such an invocation occurs between completed proof phases; no active simulation segment crosses it. Each reached boundary satisfies the source and device hypotheses. These are joint replacements of the two stateful modules, not independent replacements of their masked records.

At terminal return, erase the honest masks, native witnesses, raw tables, and internal proof coins. Keep the prescribed input--output interface and public transcript. The resource-relative simulator uses only the declared source device interfaces and makes no extra calls. For each induced QPT test the advantage is bounded by
\[
 \nu_{\rm ext}+\nu_{\rm hide}
 +T\nu_{\rm ZK}+\varepsilon_{\rm ao}+T\varepsilon_{\rm snd}
 =\negl(\lambda),
\]
where the bounds are uniform over the reached auxiliary states. The sum of the honest preprocessing and proof-completeness errors bounds honest failure. The public phase and size bounds give one polynomial running-time bound for the simulator.
\end{proof}

\subsection{Substitution in the OT protocol}
\label{app:bc-ot-compiler}

Let $L_0=kN_0m+t_A$ and $n=kN_0m$ be the outer OT sizes. Count only bits in newly created CP records. Alice's test mask and canonical data table require $q_A=L_0+6n$ bits. Bob's offsets, answers, audit masks, selector masks, full labels, and symbol-erasure flags require at most $q_B=5L_0+2n+km+m$ bits. Hence
\begin{equation}
 q_A,q_B\le q_{\rm bank}:=10L_0.
 \label{eq:cp-bank-capacity}
\end{equation}
Repeated projections and joint $\Label$ or $\Pay$ requests create no new records. Both preprocessing capacities include all late metadata.

\begin{corollary}[OT from two preprocessing commitments]
\label{cor:bc-ot-compiler}
Under \Cref{def:vector-commitment} and the hypotheses of \Cref{thm:unified-di-ot}, two sequential preprocessing commitments of capacity \eqref{eq:cp-bank-capacity}, followed by \OTprotocolparts{} with \Cref{prot:masked-cp}, realize one-bit OT with auxiliary-input-uniform QPT simulation and negligible error. The CP phases use no devices after preprocessing.
\end{corollary}
\begin{proof}
Apply \Cref{thm:masked-cp} jointly to both directions and all pending records, retaining the outer state and the prescribed device interfaces. The resulting execution is the ideal-CP protocol of \Cref{thm:unified-di-ot}. Apply that theorem and add the two simulation errors. All entropy arguments take place in this ideal-CP execution. \Cref{thm:sequential-composition} then composes completed OT calls.
\end{proof}

\section{Device-independent string commitment from qOWFs}
\label{sec:di-bc}

The masked-CP construction of \Cref{sec:bc-cp-compiler} needs one extractable commitment to an entire random string. This section constructs that commitment with one device window. Its public native opening relation and agreement with all native openings provide the interface of \Cref{def:vector-commitment}. \Cref{sec:di-bc-coin} uses this source and masked CP to obtain simulation-secure bit commitment and coin tossing. Alice is the source committer and Bob the source receiver.

The field is $\Ffield$, with fixed binary basis $e_1,\ldots,e_k$ and $k=6$. The letter $q$ continues to denote the number of committed bits. Choose an admissible code length $m$ with $q\le\lfloor m/100\rfloor$ and use $q$ secret evaluation positions. The codes $\CodeB,\CodeD,\CodeA$, secret maps, and diagonal twists are those of \Cref{lem:multiplicative-codes}. Put
\[
 n=kN_0m,\qquad L_q=n+t_A,\qquad t_A=\lceil m/200\rceil,
\]
where $N_0$ is the fixed microblock size. Products are coordinatewise over $\Ffield$. Only $\CodeA$ and $\CodeB$ are twisted: $\CodeA^\sigma=\sigma*\CodeA$, $\CodeB^\sigma=\sigma*\CodeB$, and $\secmap_\sigma(v)=\secmap(\sigma^{-1}*v)$. All vector fibers below use these $q$ secret positions. We first establish the classical simulation tools, then give the native records, device protocol, and security proofs.

\subsection{Concrete classical commitments and bounded quantum simulation}
\label{app:concrete-primitives}

This subsection proves the implementation properties used inside the source proof. The generic definition of equivocation does not promise two simultaneous classical witnesses. We use the BCKM compiler \cite[Section~4.1, Protocols~1--2]{OT2} on Naor's classical commitment. All assumptions use the auxiliary-input convention of \Cref{def:comp-aug-sec}: the qOWF and resulting PRG resist the allowed nonuniform quantum advice. The corresponding OWF-to-PRG convention is discussed in \cite[footnote~6]{OT2}.

\begin{lemma}[Post-quantum zero knowledge for NP]
\label{lem:npbc-pqzk}
Under qOWFs, every NP relation has a classical-message sequential proof with perfect completeness, negligible statistical soundness, and quantum-computational zero knowledge.  For every nonuniform QPT verifier there is one QPT simulator, fixed from the verifier independently of its auxiliary state, whose output channel is computationally indistinguishable from the real verifier channel uniformly on arbitrary quantum auxiliary inputs entangled with an inaccessible reference.  Polynomially many sequential invocations preserve this guarantee.
\end{lemma}

This adapts the GMW proof for NP~\cite{GMW91} and Watrous's quantum-verifier analysis~\cite[Sections~3 and~6.2]{Watrous09ZK} to Naor's classical statistically binding commitment. The proof below gives the required simulation guarantees.

The source mode relation and masked-CP predicates both use this ordinary proof system. Its internal commitments are classical and local to the proof; no external CP or device call occurs in an active simulation segment. Fix a simulator circuit on every syntactically valid public input, with a polynomial clock and default rejection on malformed inputs. Its zero-knowledge guarantee is used only for true instances. The same total circuit is an efficient continuation on other instances in the hiding hybrids of \Cref{sec:masked-cp-proof}.

\begin{lemma}[Naor setup and auxiliary-input uniformity]
\label{lem:naor-auxiliary-setup}
Let $G:\mathbb F_2^\lambda\to\mathbb F_2^{3\lambda}$ be the quantum-secure PRG and put
\[
 \mathsf{Com}_a(b;s)=G(s)\oplus ba,
 \qquad\mathsf{Ver}_a(c,b,s)=[c=G(s)\oplus ba].
\]
Hiding holds uniformly over every receiver-chosen $a\in\mathbb F_2^{3\lambda}$ and allowed auxiliary/reference state independent of the fresh seed. For honest uniform $a$, except with probability $2^{-\lambda}$, every possible commitment string has at most one bit admitting a native opening.
\end{lemma}
\begin{proof}
For each bit replace $G(s)$ by independent uniform $U$. The reduction applies the known controlled XOR with $ba$ and runs the proposed quantum distinguisher on its auxiliary state. The uniform experiment is independent of the bit, so two PRG hybrids prove hiding. A nonnegligible supremum over $a$ and auxiliary states would supply a maximizing sequence as nonuniform advice, contradicting the PRG assumption. The reduction need not prepare that advice. If $a$ is produced by a measured quantum first-message procedure, include this procedure in the reduction; no conditioning on a rare setup value is necessary.

Two openings to different bits require $a=G(s_0)\oplus G(s_1)$. There are at most $2^{2\lambda}$ such strings among $2^{3\lambda}$ setups. Outside this setup-wide event, unique-bit binding holds simultaneously for all strings and seeds, even against an unbounded quantum committer. Fresh honest setups give the corresponding polynomial union bound independently of the preceding auxiliary state.
\end{proof}

\begin{lemma}[A fixed clock for almost constant success]
\label{lem:fixed-clock-rewinding}
Let $Q$ be a polynomial-size purified partial-simulation circuit. Suppose its success probability satisfies $|p(\psi)-\theta|\le\epsilon(\lambda)$ for every auxiliary input, where $\epsilon$ is negligible and $\theta,1-\theta$ are known and inverse polynomial. One circuit fixed from $Q$, independently of the input and the unknown function $\epsilon$, produces its conditional successful state with negligible error jointly with a reference. Its clock is polynomial on every input.
\end{lemma}
\begin{proof}
Use the measurement/uncompute/phase-flip/recompute procedure of \cite[Lemmas~8--9 and Figure~3]{Watrous09ZK}. Measure success initially, then allow at most
\[
 T=\left\lceil\frac{\lambda}{\theta(1-\theta)}\right\rceil
\]
corrective rounds, each followed by a success measurement. If none succeeds, return a distinguished additional simulation-failure flag. The reference is never acted on. There are at most $2T+1$ calls to $Q$ or $Q^*$.

The hypothesis on all input states bounds the spectrum of the success POVM within $\epsilon$ of $\theta$. The comparison-unitary construction in the proof of Watrous's Lemma~9 gives a mathematical unitary $U$ with constant success $\theta$ and $\|Q-U\|\le\sqrt{2\epsilon}$; $U$ need not be efficiently computed. Replacing all circuit uses by $U$ changes the channel by $O(T\sqrt\epsilon)$, by telescoping operator norms. Tensoring the identity on an arbitrary reference changes no bound. For $U$, the exact calculation of Lemma~8 gives failure probability
\[
 (1-\theta)(1-2\theta)^{2T}
 \le e^{-4\theta(1-\theta)T}\le e^{-4\lambda}.
\]
The normalized successful $Q$ and $U$ states differ by at most $\operatorname{poly}(1/\theta)O(\sqrt\epsilon)$, using $p\ge\theta/2$ for sufficiently large $\lambda$. These errors are negligible. This uses a fixed polynomial clock with an exponentially small ideal tail, not Markov truncation of an expected-time simulator, and requires no knowledge of $\epsilon$.
\end{proof}

\paragraph{Ordinary aborts and malformed messages.}
Every activation has an explicit polynomial clock and message length bound. Missing, invalid, or abort messages set an ordinary-abort flag; the honest controller sends no further protocol messages and preserves the state at that cut. The remaining circuit is padded with dummy operations. Purify classical measurements for the active rewinding segment, and dephase the public classical output registers afterward.

Let the partial simulator guess a challenge $\widehat J$ from a set of size $h$, with $\theta=1/h$. On a valid challenge $J$, its internal success is $[J=\widehat J]$. On an ordinary-abort branch, set internal success using an independent Bernoulli($\theta$) coin and, on success, output the saved ordinary abort. There is no witness promise on this branch. In the commitment-hiding comparison where the sent commitments are independent of $\widehat J$, the total internal success probability is $\theta$, whether or not the receiver aborts. Computational hiding therefore gives $\theta\pm\negl$ for the original circuit, uniformly over auxiliary inputs. The binary BCKM coin is exact. General uniform indices use rejection sampling capped at $\lambda$ trials and a valid fixed fallback, giving error at most $2^{-\lambda}$.

To check the joint abort state, let $\Omega_{\rm ab}$ be its subnormalized receiver-reference output before conditioning on internal success, and let $\theta_c$ be the actual bounded coin probability. Independence of this coin gives exactly
\[
 \Omega_{\rm ab,success}=\theta_c\Omega_{\rm ab},
 \qquad|\theta_c-\theta|\le2^{-\lambda}.
\]
After conditioning on success of probability $p=\theta\pm\negl$, its weight is multiplied by $\theta_c/p=1\pm\negl$. No normalization on the public abort or nonabort branch is used. On valid challenges the usual hiding hybrids change only commitments whose openings are not sent; on earlier abort branches no challenge-dependent opening is sent. The loss in translating these comparisons to the internal successful state is at most the polynomial factor $1/\theta$. Thus the complete output channel, including ordinary abort, is preserved. Abort occurring after a successfully simulated valid challenge is ordinary subsequent channel processing. Public transcript copies are included in the comparison; they remain internal during the active rewinding segment and are released at its completion.

\begin{lemma}[Offline witnesses for the classical BCKM instantiation]
\label{lem:bckm-offline-witnesses}
Under qOWFs, the BCKM compiler applied to Naor's classical commitment \cite[Theorem~4.1 and Section~4.1]{OT2} gives a statistically binding, quantum-computationally hiding commitment with the following simulation. For every QPT receiver, one strict-QPT commit simulator is independent of the target bit. On each completed nonabort branch it retains private classical native witnesses for either bit of every unopened handle. A later opening uses the requested witness without interaction or further rewinding to compute that witness. The real and simulated commit/open channels are computationally indistinguishable, uniformly over admissible auxiliary inputs and jointly with an untouched reference. Ordinary abort may have arbitrary probability; only additional simulator failure is negligible. Sequential completed phases preserve this guarantee. Statistical binding is asserted for real commitments.
\end{lemma}
\begin{proof}
In each repetition, BCKM's partial simulator guesses which pair the receiver will challenge. It commits that pair to equal bits and the other pair to opposite bits, using four ordinary Naor commitments. It retains all its own seeds. Hiding and the abort convention above give internal success $1/2\pm\negl$, so \Cref{lem:fixed-clock-rewinding} supplies the successful partial state with negligible trace error. Open the checked equal pair and repeat. All final public mask bits $e_i$ are uniform and independent of the eventual target.

Number the members of the unchecked pair by $a\in\{0,1\}$. Their bits are $d_i\oplus a$, with retained seeds $s_{i,a}$. For either target bit $b$, choose
\[
 \alpha_i(b)=b\oplus e_i\oplus d_i,
 \qquad w_b=(\alpha_i(b),s_{i,\alpha_i(b)})_{i=1}^{\lambda}.
\]
The selected base opening has bit $b\oplus e_i$, exactly the native verification equation. Previously opened pairs are unchanged. The two $w_b$ are simultaneous classical NP witnesses. This step uses Naor's one-message opening; a quantum or interactive base opening would not by itself supply such offline witnesses.

After each completed rewind, measure only the simulator-owned seed register $S$ and retain a persistent private classical copy. Exactly,
\[
 \operatorname{Tr}_S\mathsf{Dephase}_S(\rho_{SVR})
       =\operatorname{Tr}_S\rho_{SVR}.
\]
Thus there is no additional disturbance of the joint receiver-reference marginal. Future opening operations only read this saved copy and send classical functions of it; they never overwrite or uncompute it. Expanding in matrix units $|s\rangle\langle s'|$ shows that the off-diagonal terms vanish when this unchanged private copy is finally traced, with or without the earlier measurement. This proves the same continuation channel even when all seeds are measured together. Subsequent rewinding acts on the new repetition and receiver workspace, not on completed-round seed copies; entanglement with those copies is covered by the reference guarantee.

On nonabort branches check all stored seed/transcript equations and the equal/opposite pair conditions. The ideal successful partial state always passes; trace-distance contractivity bounds any additional rejection by the negligible rewinding error. Skip this check on ordinary abort. A polynomial union bound covers all repetitions and handles, without calling the receiver's genuine abort mass negligible. The resulting witnesses are exactly valid on each completed remaining branch. The ordinary BCKM sequential hiding/equivocation hybrids apply to these same classical selected-base openings and retained receiver states. Uniformity in the auxiliary input follows from \Cref{lem:naor-auxiliary-setup,lem:fixed-clock-rewinding}.
\end{proof}

\paragraph{Target deferral and real-table binding.}
First replace a real commitment while keeping its genuine target at the genuine sampling cut. Only in the target-independent simulated experiment may an unused fresh audit target be sampled later. No target-dependent message length, early abort, or visible decision may precede that later cut. This gives statistical deferred sampling inside that experiment, not a statistical claim about real hiding. For the opposite mode, the saved witnesses can change unopened labels and offsets, but cannot change already disclosed values.

Meanwhile, Alice's table commitments received by honest Bob remain real. Outside their honest Naor setup bad events, the four base bits in every BCKM repetition are unique if they exist. Two possible outer openings would require the unchecked pair to be opposite in every repetition, whereas the checked pair must be equal. A fresh fair challenge leaves such an opposite pair unchecked with probability at most $1/2$ per repetition, giving $2^{-\lambda}$ for one handle. This argument applies after arbitrary simulated prefixes because the honest receiver's new randomness is fresh. Save completed transcripts in protected records; later total simulation channels cannot increase the mass of their historical bad-binding event. This proves the all-native-opening binding property in the joint extraction experiment.

\paragraph{Ordinary statistically sound quantum ZK.}
\begin{proof}[Proof of \Cref{lem:npbc-pqzk}]
Reduce the fixed NP instance to a graph with $h\ge2$ edges. Before sequential GMW checks, let the verifier send one Naor setup $a$, used with fresh seeds for each binary-encoded vertex color. An honest verifier chooses uniform $a$. Outside the setup-wide event, every commitment is uniquely binding, so every round on a false graph rejects with probability at least $1/h$. With $\lceil\lambda h\rceil$ fresh challenges, soundness is at most $2^{-\lambda}+e^{-\lambda}$, plus the negligible bounded-sampling error. This is statistical soundness against arbitrary provers. Valid fallback permutations and edges preserve perfect completeness.

For a malicious verifier, run its setup step normally and return ordinary abort on invalid setup; do not rewind across it. For every well-formed $a$, \Cref{lem:naor-auxiliary-setup} gives hiding. Guess an edge, give its endpoints uniformly distinct colors, and commit arbitrary valid colors elsewhere. Apply the abort-balanced partial circuit above with $\theta=1/h$ and then \Cref{lem:fixed-clock-rewinding}. On valid success reveal the two endpoint seeds. The true-instance hybrid changes unopened dummy colors to a randomly permuted proper coloring; the simulator never receives that coloring. Iterate with the entire receiver/reference state preserved. This adapts Watrous's GMW argument, whose literal presentation assumes perfectly binding noninteractive commitments and valid challenge encodings.

The circuit is total and strictly polynomial on false but syntactically valid instances as well, although its ZK guarantee is invoked only on true instances. This is the property needed when a witness is removed before the analytical uniform-mask or dummy-bank hybrid. No external device, another honest session, or irreversible external transcript log responds while one of these rewinding segments is active.
\end{proof}

\subsection{Classical native records and the mode relation}
\label{sec:bc-step-one}

We use the classical commitment of \Cref{lem:bckm-offline-witnesses} and the zero-knowledge proof of \Cref{lem:npbc-pqzk}. The commitment simulator saves both seeds of each unchecked opposite-bit pair. Their measured private copies are retained and only classically read; later rewinding does not uncompute completed copies. Each classical commitment or proof invocation runs to completion, with the other party as receiver or verifier, before a device window or another outer operation begins.

All records use fixed canonical bit encodings. Every trit is represented by two bits with the fourth encoding rejected; row and column records use the parity conventions of the physical experiment. Separately opened fields have separate bit handles. Write $\Ver(c,v,w)$ for the conjunction of the corresponding native-opening checks. Public records include session identifiers, fixed lengths, all earlier openings, and their verification results.

\begin{definition}[Packed source mode relation]
\label{blk:def:mode}
The public instance contains the offset commitments $c^U$, Alice's offsets $V$, the ordered microblocks $I_{i,j}$, the test and audit sets, the opened erasure and selector masks, the unopened full-label commitments $c^\beta$, and $(\sigma,M)$. Let $\mathcal I=[m]\setminus\mathcal E$. Its syntax requires
\[
\begin{gathered}
 \sigma\in(\Ffield^*)^m,\quad M\in\Ffield^m,\quad
 |\mathcal E|\le\lfloor m/100\rfloor,\\
 J^{i,j}=\varnothing\quad(i\in\mathcal E),\qquad
 J^{i,j}\subseteq I_{i,j}\setminus T_B^{i,j},\\
 N_0/8\le |J^{i,j}|\le N_0/2
       \quad(i\in\mathcal I,\ j\in[k]).
\end{gathered}
\]
All public mask openings must already be verified and consistent. The witness consists of $U'\in\mathbb Z_3^{L_q}$, $\beta'\in\Ffield^m$, and their classical native witnesses. The predicate $\Mode$ checks every required native opening, all previously opened offset values, and
\begin{equation}
\begin{gathered}
 J^{i,j}\subseteq\{a:U'_a+V_a=\beta'_{i,j}\},\\
 R':=M+\beta'\in\CodeB^\sigma,
 \qquad\secmap_\sigma(R')=0^q .
\end{gathered}
\label{blk:eq:mode-equations}
\end{equation}
Here $\beta'_{i,j}\in\mathbb F_2$ is the $j$th binary coefficient of $\beta'_i$. This predicate checks inclusion, not equality with an entire question class, and does not certify that $\mathcal E$ is the set produced by the honest size-failure rule.
\end{definition}

\subsection{The one-window string commitment}
\label{sec:bc-step-two}

The protocol uses the same physical input/anchor experiment as the packed OT, but replaces CP by the inner commitment above. It sends payloads only for $i\in\mathcal I$. There are no payloads, hashes, or pad equations for erased symbols. The full label record nevertheless contains all $km$ bits, including the erased coordinates. The late mode seed $\sigma$ is sampled after the physical leakage interface closes and all selectors are fixed.

\begingroup
\tcbset{lines before break=8}
\begin{protocolbox}
\protocolstyle
\emph{Input:} Alice has $u\in\mathbb F_2^q$. Public parameters satisfy \Cref{lem:unified-parameters}, sized also for $q$. Every malformed message, invalid native opening, or failed required public check causes terminal abort. All messages carry the session identifier.
\begin{protocolsteps}
\item \emph{Alice (hidden test).} Sample fresh uniform physical row questions $X^{\rm in}\in\mathbb Z_3^{L_q}$ and a uniform test set $T_A$ of size $t_A$. Commit to its incidence string.
\item \emph{Bob (offsets).} Sample uniform $U\in\mathbb Z_3^{L_q}$ and commit separately to its trits.
\item \emph{Alice and Bob (one physical window).} Alice sends fresh uniform $V\in\mathbb Z_3^{L_q}$. The laboratories input $X^{\rm in}$ and $Y=U+V$ to the prescribed current device batch, collect all outputs, and close the physical leakage interface. Normalize recorded rows and columns to the prescribed parities. Bob commits to every column $B_a$ before the test set is opened.
\item \emph{Alice (partition).} Open $T_A$ and send a fresh uniform ordering of its complement. Bob checks the opening, size, and ordering. Consecutive groups of $N_0$ define $I_{i,j}$ for $(i,j)\in[m]\times[k]$.
\item \emph{Bob (early masks and labels).} Independently select one uniform audit coordinate $T_B^{i,j}\subset I_{i,j}$ and a fair bit $\beta_{i,j}$ in every microblock. Put $\beta_i=\sum_j e_j\beta_{i,j}$ and $J_b^{i,j}=\{a\in I_{i,j}\setminus T_B^{i,j}:Y_a=b\}$. Erase symbol $i$ iff at least one of its two candidate sizes in some microblock lies outside $[N_0/8,N_0/2]$. Set $J^{i,j}=J_{\beta_{i,j}}^{i,j}$ for nonerased $i$ and $J^{i,j}=\varnothing$ otherwise. Commit to the audit masks, the full $km$-bit label record, the $m$-bit erasure record, and all selector masks, on fixed metadata. Defer every size decision.
\item \emph{Alice (canonical table).} Sample independent fair anchors $C_a$ on all $L_q$ coordinates. For $C_a=1$ retain the actual row and question; for $C_a=0$ replace them by $(X_a,A_a)=(0,000)$ and discard the erased raw controller record. Commit separately to the canonical $(X_a,A_a,C_a)$ fields on all data coordinates.
\item \emph{Bob and Alice (test and global audit).} Bob opens the audit masks and his $(U,B)$ records on $T_A$. Alice verifies their syntax and sends the prescribed hidden-test decision, counting anchors as passing. On success she opens her table on all audit positions. Bob verifies every opening and the canonical syntax. He checks that the \emph{total} number of audit losses is at most $\kappa_{\rm aud} km$ and that $|\mathcal E|\le\lfloor m/100\rfloor$. He sends one aggregate decision and aborts on rejection. No per-block audit acceptance is required.
\item \emph{Bob and Alice (selector disclosure).} Bob opens $\mathcal E$ and every $J^{i,j}$ mask. Alice checks their syntax, disjointness, audit avoidance, and active size bounds and sends the selector success decision. The label record remains unopened.
\item \emph{Bob (mode).} Sample fresh independent $\sigma\in(\Ffield^*)^m$ and uniform $R\in\CodeB^\sigma_{0^q}$. Send $\sigma$ and $M=\beta+R$ and prove \Cref{blk:def:mode} using the actual records. Use the bounded nonzero sampler specified in \Cref{lem:unified-parameters}.
\item \emph{Alice (payload).} Verify $\Mode$ and open $(X,C)$ on the union of selectors. Sample independent uniform $f\in\CodeA^\sigma_{0^q}$ and $D\in\CodeD_u$. For each $i\in\mathcal I$, sample uniform shares $Q_{i,1},\ldots,Q_{i,k}\in\Ffield$ subject to $\sum_jQ_{i,j}=f_i+M_iD_i$. Set
\begin{equation}
 p_{i,j}^b=Q_{i,j}+b e_jD_i,
 \quad Z_{i,j}^b=(A_a[b])_{a\in J^{i,j}},
 \quad F_{i,j}^b=p_{i,j}^b+K_{i,j}^b,
\label{blk:eq:masked-payload}
\end{equation}
where $K_{i,j}^b$ is the field element represented by a fresh uniform $k\times|J^{i,j}|$ binary matrix applied to $Z_{i,j}^b$. Send both matrix seeds and ciphertexts for each active microblock.
\item \emph{Bob (receipt).} Verify all required $(X,C)$ openings, anchor syntax, and both branches' matrix and field encodings. Issue the ordinary receipt. Retain only the public transcript $\tau$ and its determined verification records. Discard $U,Y,B,\beta,R$, all mode witnesses, device state, and private scratch. There is no public decoding or private-output-dependent acceptance flag.
\item \emph{Alice (native opening).} On the one native-opening request, send $u$ and open the full committed canonical data table, consistently with every earlier opening. Bob applies the public predicate below, outputs $u$ on acceptance and $\perp$ otherwise, and closes the session. No device operation occurs.
\end{protocolsteps}
\smallskip\noindent\textit{Where each step is used.}
Steps 1--8: \Cref{blk:lem:audits,blk:lem:sparse,blk:lem:label-entropy}. Step 9: \Cref{blk:lem:witness,blk:lem:mode-soundness}. Step 10: \Cref{blk:lem:algebra,blk:lem:pads}. Steps 11--12: \Cref{blk:thm:state,blk:thm:agreement,lem:vector-honest-witness}.
\end{protocolbox}
\par\nobreak\noindent
\begin{minipage}{\linewidth}
\captionsetup{hypcap=false}
\captionof{algorithmica}{One-window packed string commitment and public native opening.}
\label{prot:di-bc}
\end{minipage}
\endgroup

\paragraph{The public native-opening predicate.}
Given $\tau,u,w$, verify $u\in\mathbb F_2^q$, all table openings, canonical syntax, and agreement with earlier fields. Recompute both pads from the table on every active microblock and put $p=F+K$. For each $i\in\mathcal I$, require the values $(p_{i,j}^1+p_{i,j}^0)/e_j$ to agree for all $j$, and denote their common value by $D_i$. Put
\begin{equation}
 f_i=\sum_{j=1}^k p_{i,j}^0+M_iD_i\quad(i\in\mathcal I).
 \label{blk:eq:opening-equations}
\end{equation}
Solve the two linear extension systems into $\CodeD$ and $\CodeA^\sigma$. Reject unless they have solutions with $\secmap(D)=u$ and $\secmap_\sigma(f)=0^q$. The solutions, when they exist, are unique because $|\mathcal E|<d(\CodeD),d(\CodeA)$. Call this predicate $\Ver_{\rm str}(\tau,u,w)$, where $w$ is the classical full-table opening. It reads no private receiver record. In the mask-bank application the bank is never natively opened; this same predicate is evaluated inside the ordinary NP proof.

\subsection{Extraction, agreement, and continuing state}
\label{blk:sec:extraction}

We establish the payload identity and the alternative mode witness, then prove agreement with all native openings and preservation of the continuing state.

\begin{lemma}[Payload identity and native witnesses]
\label{blk:lem:algebra}
For any accepted native opening and any actual $R=M+\beta\in\CodeB^\sigma$, the word obtained by summing the actual selected plaintexts on active symbols is the restriction of
\[
 H=f+R*D\in\CodeA^\sigma,
 \qquad\secmap_\sigma(H)=\secmap_\sigma(R)*u.
\]
In particular, its secret is zero in ordinary mode and $u$ when $\secmap_\sigma(R)=\mathbf1^q$.
\end{lemma}
\begin{proof}
On $i\in\mathcal I$ the native equations give $\sum_jp_{i,j}^{\beta_{i,j}} =f_i+(M_i+\sum_j e_j\beta_{i,j})D_i=f_i+R_iD_i$. Multiplication compatibility gives $R*D\in\CodeA^\sigma$ and the displayed secret identity. Unique extension supplies the same full $f,D$ used on all active symbols. No erased payload is needed.
\end{proof}

\begin{lemma}[Alternative $\Mode$ witness]
\label{blk:lem:witness}
Suppose Bob's records use the commit simulator of \Cref{lem:bckm-offline-witnesses}. Keep the entire physical experiment and all disclosed records fixed. If the actual mask satisfies $R\in\CodeB^\sigma_{\mathbf1^q}$, there is an efficiently computed classical native witness for the zero-mode statement.
\end{lemma}
\begin{proof}
The constant function one belongs to $\CodeB$. Put $v=\sigma*\mathbf1^m$, $\beta'=\beta+v$, and $R'=R+v$. Then $M=\beta'+R'$ and $R'\in\CodeB^\sigma_{0^q}$. Read the saved native witness for each bit of the new full label $\beta'$, including its erased coordinates. For each $a\in J^{i,j}$, set $U'_a=\beta'_{i,j}-V_a$ in $\mathbb Z_3$ and select the saved native witnesses for that trit. Elsewhere keep $U'_a=U_a$. Selected sets are disjoint and avoid every previously opened test or audit field, so these assignments are consistent. Audit, erasure, and selector masks keep their actual opened values. The predicate requires only inclusion of the selected positions, so all its clauses hold. No late device call or equivocal-opening interaction is used to compute this polynomial-length witness.
\end{proof}

\begin{lemma}[Soundness of ordinary $\Mode$]
\label{blk:lem:mode-soundness}
When Bob's records are real, outside their statistical binding error and the ordinary proof's soundness error, an accepted $\Mode$ fixes his actual early labels and selected question classes and gives $M+\beta\in\CodeB^\sigma_{0^q}$.
\end{lemma}
\begin{proof}
Outside the setup-wide and cut-and-choose binding failures of the real inner commitments, every committed field has at most one native value. An accepted false NP statement has negligible probability, uniformly over its adaptively chosen instance. A true statement therefore supplies these unique early values and the mode equations. The argument does not assert that every real handle has an opening: if the required openings do not exist, the statement is false. This event is paid once before entropy or character sums.
\end{proof}

\begin{lemma}[One audit controls every native opening]
\label{blk:lem:audits}
In the extraction experiment, let $W$ count canonical data coordinates with $C_a=1$ and $A_a[Y_a]\ne B_a[X_a]$ in Alice's unique possible table. If no full table admits native openings, the bad event below is defined to be false and the subsequent agreement claim is vacuous. For the fixed constants of \Cref{lem:unified-parameters},
\[
 \Pr[\mathsf{aud}=1\ \land\ W>m/40]
 \le \varepsilon_a+\varepsilon_{\rm bind}
       +\varepsilon_{\rm eq\text{-}fail},\qquad
 \varepsilon_a=e^{-\Omega(m)}.
\]
This controls every later accepted native opening simultaneously.
\end{lemma}
\begin{proof}
Alice's commitments are received by the ordinary honest algorithm, even though Bob's sender-role records are simulated. Fresh honest setup strings and cut-and-choose challenges give statistical binding against Alice with arbitrary preceding auxiliary state. Keep an immutable copy of each completed table transcript; later simulation does not change this historical event.

In the fixed equivocal-Bob experiment, the audit and selector commit circuits do not read their targets. No target-dependent size flag is sent before Alice fixes the table. Thus the audit coins may be sampled after this cut, exactly in this experiment. Condition on a possible table, the physical records, and the partition. If $w_{i,j}$ positions in a microblock are bad, its one audit has loss probability $w_{i,j}/N_0$; these audit coins are independent. The average audit expectation is $W/(kN_0m)$. Hoeffding bounds acceptance when $W>\eta_A kN_0m$ by $\exp[-2(\eta_A-\kappa_{\rm aud})^2km]$. The chosen $\eta_A\le1/(40kN_0)$ proves the claim. Invalid earlier openings only reject. Binding assigns at most one table for analysis and imposes no code-membership assumption on it. Consequently this one event controls all future openings, with no union bound over witnesses and no computational transfer of an existential event.
\end{proof}

\begin{protocolbox}
\protocolstyle
\emph{Setting:} Alice is corrupt; SIM plays the source receiver.
\begin{protocolsteps}
\item \emph{Records and physical execution.} Receive Alice's records honestly. Use target-independent equivocal commitments for all Bob-owned fields. Execute the genuine prescribed physical inputs, answers, and leakage window, and keep all tests, audit masks, labels, erasures, and selectors at their actual values. Open public fields genuinely and reproduce every ordinary abort.
\item \emph{Mode.} Sample the genuine fresh diagonal $\sigma$ and uniform $R\in\CodeB^\sigma_{\mathbf1^q}$. Send $M=\beta+R$ and run the real $\Mode$ prover with the alternative native witness of \Cref{blk:lem:witness}.
\item \emph{Private extraction.} After all public payload checks, use $\widehat Z_{i,j}=(C_aB_a[X_a])_{a\in J^{i,j}}$ to estimate the actual $\beta_{i,j}$-branch pads. Sum the resulting selected plaintexts to obtain a word $\widehat H$ on $\mathcal I$. Run the total decoder of \Cref{lem:ag-decoding}, with the public error bound $T=\lfloor m/40\rfloor$ and erasures $\mathcal E$. If it returns a codeword whose secret lies in $\mathbb F_2^q$, retain that secret as $u^*$; otherwise retain $\perp$.
\item \emph{Receipt and retirement.} Issue the same public receipt and discard the same honest receiver systems as the real protocol. The private $u^*$ is omitted from the state comparison. Any later native opening is verified from $\tau$ by the ordinary predicate.
\end{protocolsteps}
\smallskip\noindent\textit{Where each step is used.}
Steps 1--2: \Cref{lem:bckm-offline-witnesses,blk:lem:witness,blk:lem:label-entropy}. Steps 3--4: \Cref{blk:thm:agreement,blk:thm:state}.
\end{protocolbox}
\par\nobreak\noindent
\begin{minipage}{\linewidth}
\captionsetup{hypcap=false}
\captionof{algorithmica}{A fixed source extractor using a real proof of $\Mode$.}
\label{prot:di-bc-sim}
\end{minipage}

\begin{theorem}[Agreement with every native opening]
\label{blk:thm:agreement}
For the fixed simulator of \Cref{prot:di-bc-sim},
\[
 \Pr_{\rm sim}[\exists u,w:\Ver_{\rm str}(\tau,u,w)=1,
                                      \ u\ne u^*]
 \le \varepsilon_a+\varepsilon_{\rm bind}
                  +\varepsilon_{\rm eq\text{-}fail}.
\]
\end{theorem}
\begin{proof}
Work outside the single exceptional event of \Cref{blk:lem:audits}. For any accepted native opening, \Cref{blk:lem:algebra} supplies a full $H\in\CodeA^\sigma$ with secret $u$. A correct raw selected coordinate gives the same bit to Alice's pad and the extractor's pad; anchors give zero to both. A wrong active symbol therefore requires at least one of the $W$ bad selected coordinates. The word $\widehat H$ has at most $W\le m/40$ errors relative to this $H$ and at most $m/100$ erasures. The decoding margin is
\[
 2W+|\mathcal E|\le 0.06m
       <m-2\deg G-g,
\]
so the checked total decoder returns $H$. Its secret is binary because the native predicate accepted $u$. Hence $u^*=u$. The argument begins with an arbitrary existing native witness; it does not require the extractor to know beforehand that the received word admits one. If the decoder returns $\perp$, no accepted native opening exists on this event. The same reasoning applies to every witness for the one fixed table and transcript.
\end{proof}

\begin{lemma}[Label entropy in the committer-side comparison]
\label{blk:lem:label-entropy}
In the totalized equivocal-Bob experiment, immediately before the fresh diagonal is sampled, let $E_A$ contain Alice's full continuing state, the public transcript, and the untouched reference. Then
\begin{equation}
 H_{\min}(\beta\mid E_A)
       \ge km-2\ell-3t_A-1.
 \label{blk:eq:label-guess}
\end{equation}
The bound also holds on the actual continuing substate, without normalizing that substate.
\end{lemma}
\begin{proof}
Apply the joint communication-elimination and selector-symmetry argument of \Cref{lem:post-J-selector-entropy} in this fixed native experiment. It is legitimate here because Bob's commit simulator is target independent and his public native openings can be generated from the guessed value and saved seed witnesses. Their polynomial encoding lengths are not additional raw-data disclosures. Guess the physical-window communication with debit $2\ell$ and the $3t_A$ test-answer bits at their first release. All remaining Bob outcomes stay private. Whole-laboratory no-signalling gives the uniform unused questions needed for the label exchange in each microblock. The two candidate size tests and the honest symbol-erasure rule are symmetric in labels zero and one. Therefore the disclosed selectors and erasures do not debit label entropy. The one aggregate audit/size decision costs one bit. Restore the guessing factors and restrict to actual continuation without a rare-event normalization.
\end{proof}

\begin{theorem}[State-preserving source extraction]
\label{blk:thm:state}
Omitting $u^*$, \Cref{prot:di-bc-sim} and the real commitment phase have computationally indistinguishable joint corrupt states, public transcripts, honest retained verification records, and references. The simulator is strict QPT and is fixed independently of its unknown auxiliary state.
\end{theorem}
\begin{proof}
First replace Bob's real inner commitments sequentially by their equivocal versions, opening disclosed fields genuinely. Alice's commitments remain real. Replace the real zero-mode proof by the fixed-clock ZK simulator while the statement is true. The receipt retains only $\tau$, so the remaining comparison uses one total witness-free continuation on every syntactically valid $M$.

By \Cref{blk:lem:label-entropy,lem:diagonal-mask}, the joint message $(\sigma,\beta+R_0,E_A)$ is within $\varepsilon_m$ of $(\sigma,U,E_A)$, and the latter is within $\varepsilon_m$ of $(\sigma,\beta+R_1,E_A)$, where $R_0,R_1$ are uniform in the zero and all-one secret fibers. In the middle experiment no honest prover is asked to prove a false statement; the same total proof-simulation circuit is merely an efficient continuation. In the final endpoint, \Cref{blk:lem:witness} gives a true zero-$\Mode$ statement with alternative native witnesses, so restore the real prover. Bob's commitments stay equivocal in that endpoint.

Finally compute $u^*$ locally and retire the receiver scratch. This adds no message and, with $u^*$ omitted, cannot change the retained marginal. Including abort blocks, the comparison error is
\[
 \nu_{\rm eq,B}+2\nu_{\rm ZK,\Mode}+2\varepsilon_m
                 +\varepsilon_{\rm samp}.
\]
Every channel preserves the inaccessible reference. The bounded primitive clocks, finite-field operations, and checked decoder are polynomial on all inputs. Subsequent native verification depends only on $\tau$, so any permitted efficient continuation preserves the comparison. This is the same extractor for which agreement was proved; agreement was not transported through these computational hybrids.
\end{proof}

\subsection{Joint hiding and honest native witnesses}
\label{blk:sec:hiding}

We use one common good substate for all payload comparisons. This gives joint hiding; we then show how the honest committer obtains a native opening witness.

\begin{lemma}[A common good substate in the native receiver experiment]
\label{blk:lem:sparse}
Replace Alice's hidden-test and table commitments by target-independent equivocal commitments, opening genuine fields, and keep Bob's records and $\Mode$ real. In this one fixed experiment there is a common good substate, independent of the Fourier character used later, whose complement has mass $\varepsilon_{\rm bad}+\varepsilon_{\rm bind} +\varepsilon_{\rm snd}=\negl(\lambda)$. On it the hidden test bounds the total active physical losses and accepted $\Mode$ fixes the early labels, selector inclusion, and $R=M+\beta\in\CodeB^\sigma_{0^q}$.
\end{lemma}
\begin{proof}
The first replacement is an efficient sequential computational hybrid. Only inside that experiment is the private hidden test statistically deferred past Bob's frozen records. Apply the sampling proof of \Cref{claim:sparse-test-to-random-blocks} to the entire population, and remove the one event of test acceptance with excessive total loss. No condition on all constant-size microblocks is required. Remove the statistical binding event for Bob's real records and false $\Mode$ acceptance, using \Cref{blk:lem:mode-soundness}. These are common transcript events; they are removed before any character sum. The repetition theorem is applied to the original question law and its restricted-order early selector, not to a postselected input law.
\end{proof}

\begin{lemma}[Vector payload hiding in the native source]
\label{blk:lem:pads}
On the native experiment of \Cref{blk:lem:sparse}, payloads for any prescribed $u_0,u_1\in\mathbb F_2^q$ have receiver-reference distance at most
\[
 \varepsilon_f=
 2\varepsilon_{\rm bad}
 +2^{(3k^2m+1-\gamma L_q+2\ell+b_{\rm pub})/2},
 \qquad b_{\rm pub}\le3km+O(1),
\]
in addition to the single binding and $\Mode$-soundness errors. The comparison may retain a common classical label specifying both inputs, including a fresh bank label $\mu$ used for $u_0=\mu,u_1=0^q$. There is no factor exponential in $q$.
\end{lemma}
\begin{proof}
For each fixed valid prefix, $R$ has zero secret. A fresh uniform $f\in\CodeA^\sigma_0$ makes $H=f+R*D$ uniform in that fiber, independently of $u$. Selected-branch shares conditioned on $H$ have the same distribution for both inputs.

Apply \Cref{lem:ordered-block-hidden-pads} to the remaining affine Fourier terms. Its code of differences is the untwisted $\CodeD$. For completeness, on $\mathcal I$ the difference variables are $V_{i,j}=e_jD_i+K_{i,j}^{1-\beta_{i,j}}$. A character with coefficients $a_{i,j}$ aggregates to $w_i=\sum_j e_ja_{i,j}$. Its dependence on the prescribed secret vanishes unless the zero extension of $w$ lies in $\CodeD_0^\perp\setminus\{0\}$. Indeed $(\CodeD_0|_{\mathcal I})^\perp$ is exactly the set of restrictions whose zero extensions belong to $\CodeD_0^\perp$. Consequently every input-sensitive term has at least $0.047m$ active symbol positions; there is no subtraction of $|\mathcal E|$ and no erased payload information to condition on.

The resulting missing-string selector is early and has the fixed density required by the physical guessing theorem. Communication elimination is performed before its no-communication bound, paying $2\ell$. Charge each public source-dependent audit value at its first disclosure, totaling at most $3km+O(1)$, and grant selected pads only afterward. The matrix-hash Fourier estimate then gives the displayed exponent. In the native experiment, openings of Alice's fields are generated from guessed values and retained source-independent equivocation data, so they implement the same causal disclosure argument. The common removed substate is paid once outside all characters. The message label is fixed before the fresh physical questions (or belongs to already granted side information); bank masks satisfy this condition. It is not an uncharged function of Alice's later private raw data.
\end{proof}

\begin{theorem}[Labelled joint hiding]
\label{blk:thm:hiding}
The packed source satisfies item~(iii) of \Cref{def:vector-commitment}, uniformly over string inputs and admissible auxiliary/reference states.
\end{theorem}
\begin{proof}
Use the same equivocal-Alice experiment at both input endpoints, opening genuine values and leaving Bob's records real. Remove the common bad substate of \Cref{blk:lem:sparse}, apply \Cref{blk:lem:pads}, and restore Alice's real commitments at the other endpoint. The two computational hybrid errors are added at the efficient endpoints; neither they nor the common bad mass are multiplied by the number of characters. In particular,
\begin{equation}
 (\mu,\mathsf{View}_{\rm rec}(\mu))
 \approx_c (\mu,\mathsf{View}_{\rm rec}(0^q)),
 \label{eq:labelled-vector-hiding}
\end{equation}
with the same private classical label on both sides. The native witness stays private. Every subsequent allowed efficient operation that does not access that hidden witness is a common continuation.
\end{proof}

\begin{lemma}[Honest native witnesses]
\label{lem:vector-honest-witness}
On every continuing honest-committer branch, including against a malicious receiver, Alice computes a valid native witness from her saved classical table and decommitments, without another device call.
\end{lemma}
\begin{proof}
The matrix pads recomputed from that table recover Alice's own $p_{i,j}^b$. Their differences and sums recover $D_i,f_i$ on $\mathcal I$, and her chosen full $D,f$ solve the extension systems with secrets $u,0^q$. They are unique by the public erasure cap and code distances. All earlier table openings used the same records. This proof requires no receiver decoding success or honest choice of $\sigma,M$ beyond the checked public syntax.
\end{proof}

\subsection{The source theorem and its scope}
\label{blk:sec:scope}

\begin{theorem}[One-window extractable DI string commitment]
\label{thm:di-vector-commitment}
Assume qOWFs and the physical, invocation, and retirement conditions of \Cref{def:admissible-extension,def:compiled-device-calls}. For every public $1\le q\le\poly(\lambda)$ and physical-window budget $\ell\le\poly(\lambda)$, size the packed schedule of \Cref{lem:unified-parameters} also for $m\ge100q$. Then \Cref{prot:di-bc} satisfies \Cref{def:vector-commitment} with strict-QPT algorithms and auxiliary-input-uniform simulators. It uses one device window and
\[
 L_q=\Theta(q+\ell+\log^2\lambda)
\]
coordinates, up to the fixed code-length padding convention. There are fixed positive local noise and leakage fractions under \eqref{eq:honest-device-fault} with $q_{\rm dev}\le\nu/2$ and the stated physical-window leakage model. Neither constant depends on $q,m,\lambda$.
\end{theorem}
\begin{proof}
Items (i)--(iv) of the string-commitment interface are respectively \Cref{blk:thm:state,blk:thm:agreement,blk:thm:hiding,lem:vector-honest-witness}. All comparisons include ordinary rejection and withholding. In the honest execution the native witness is valid whenever the commitment phase continues. The symmetric size-failure probability per symbol is at most $24e^{-N_0/18}<1/400$; independent trusted questions give $\Pr[|\mathcal E|>m/100]=e^{-\Omega(m)}$. The derived bound \eqref{eq:honest-noise-promise}, hidden-test sampling, and one global audit give negligible remaining failure. Individual microblocks may contain errors; the global decoder corrects the resulting symbol errors.

The label-masking, Fourier, and decoder margins follow from \Cref{lem:unified-parameters}. Their statistical errors are $2^{-\Omega(m)}$ or $e^{-\Omega(m)}$, and the polynomial sums of inner primitive errors are negligible. Taking the first admissible length at least a fixed multiple of $q+\ell+\log^2\lambda$ gives the resource bound.

The source never invokes CP. Its completed primitive phases and its one physical batch obey the same phase barriers. At receipt the receiver retains only the public verification record; other pending honest states remain untouched. The committer's private witness is encapsulated until an authorized native opening or a masked-CP predicate and is erased at terminal return. Thus the direct stateful mask-bank proof applies without a bootstrap cycle. Completed outer calls subsequently use \Cref{thm:sequential-composition}; no device reuse or concurrent invocation is inferred.
\end{proof}

\section{Device-independent bit commitment and coin tossing}
\label{sec:di-bc-coin}

The string commitment of \Cref{thm:di-vector-commitment} provides extraction and joint hiding through commitment completion. Here masked CP adds efficient simulation of a bit-commitment opening against a corrupt receiver, and commit-and-reveal gives coin tossing with abort. Under qOWFs and the model of \Cref{sec:model-assumptions}, both protocols use one device window and $O(\ell+\log^2\lambda)$ coordinates, with polynomial classical resources and the source's constant honest-fault tolerance and physical leakage rate.

\subsection{Bit commitment with simulation for both parties}
\label{sec:bc-both-sides}

A full certified projection of a one-bit CP record is a bit-commitment opening. The following protocol realizes the interface of \Cref{def:bc-two-sided}, including simulation of the opening against a corrupt receiver.

\begin{protocolbox}
\protocolstyle
\noindent Alice commits to $b\in\{0,1\}$; Bob receives. Use \Cref{prot:masked-cp} in Alice's direction only, with capacity $q=1$ and source \Cref{prot:di-bc}. Rejection is terminal abort.
\begin{protocolsteps}
\item \emph{Alice (preprocessing).} Complete the mask-bank preprocessing with Bob. Its mask is $\mu\in\{0,1\}$ and its transcript is $\tau$.
\item \emph{Alice (commitment).} Create the one-bit record $b$ by sending $c=b\oplus\mu$ with its prescribed identifier.
\item \emph{Bob (receipt).} Verify the masked-record syntax, save $c$, and send the receipt.
\item \emph{Alice (opening).} Publish the full projection value $b$ and prove $\exists u,w:\Ver_{\rm str}(\tau,u,w)=1\land c\oplus u=b$ by the predicate phase of masked CP.
\item \emph{Bob (verification).} Verify that projection and send accept or reject. On acceptance output $b$. Close the session and discard all call-local private state; Alice does the same on the result.
\end{protocolsteps}
\noindent\emph{Where each step is used.} Step~1: \Cref{thm:di-vector-commitment}; Steps~2--5: \Cref{thm:masked-cp,thm:di-bc-two-sided}.
\end{protocolbox}
\par\nobreak\noindent
\begin{minipage}{\linewidth}
\captionsetup{hypcap=false}
\captionof{algorithmica}{One-window DI bit commitment with a certified opening. Only preprocessing uses devices.}
\label{prot:di-bc-certified}
\end{minipage}

\begin{theorem}[One-window simulation-secure DI bit commitment]
\label{thm:di-bc-two-sided}
Under the assumptions of \Cref{thm:di-vector-commitment}, \Cref{prot:di-bc-certified} satisfies \Cref{def:bc-two-sided} with auxiliary-input-uniform QPT simulators and negligible error. It uses one device window and $O(\ell+\log^2\lambda)$ coordinates, with constant honest-fault tolerance and constant physical leakage rate, as for the source commitment at $q=1$.
\end{theorem}

\begin{protocolbox}
\protocolstyle
\noindent Use the single bank transcript $\tau$ throughout.
\par\smallskip\noindent\emph{Corrupt committer; SIM plays Bob.}
\begin{protocolsteps}
\item \emph{SIM (preprocessing).} Run \Cref{prot:di-bc-sim} and retain $u^*$, using $0$ if extraction returns $\perp$.
\item \emph{SIM (commitment).} On a valid $c$, extract $b^*=c\oplus u^*$ and send the ordinary receipt.
\item \emph{SIM (opening).} Run honest proof verification. Accept the ideal opening only when the proof accepts and its public bit is $b^*$; reproduce every other abort or withholding decision.
\end{protocolsteps}
\noindent\emph{Corrupt receiver; SIM plays Alice.}
\begin{protocolsteps}
\item \emph{SIM (preprocessing).} Commit honestly to $0$ with the single prescribed source-device batch.
\item \emph{SIM (commitment).} Send an independent uniform $c$ and complete the receipt phase, without receiving $b$.
\item \emph{SIM (opening).} On receiving $b$, publish it and run the fixed ZK simulator for the public projection statement. Reproduce the receiver's continuation, verification result, and abort decision.
\end{protocolsteps}
\noindent\emph{Where each step is used.} Each list specializes the corresponding three phases of \Cref{prot:masked-cp-sim}; their joint state guarantees are \Cref{thm:masked-cp}.
\end{protocolbox}
\par\nobreak\noindent
\begin{minipage}{\linewidth}
\captionsetup{hypcap=false}
\captionof{algorithmica}{The two bit-commitment simulators. The receiver simulation obtains the bit only at opening.}
\label{prot:di-bc-certified-sim}
\end{minipage}

\begin{proof}[Proof of \Cref{thm:di-bc-two-sided}]
Restrict \Cref{thm:masked-cp} to one owner, one record bit, and one full projection. Its corrupt-owner simulator extracts $c\oplus u^*$ at the receipt, and its predicate-soundness bound controls every accepted opening. Its corrupt-verifier simulator needs only the receipt length before opening and the public bit at opening. These are the simulators in \Cref{prot:di-bc-certified-sim}. The same theorem preserves the joint continuing state and all abort branches. The single source uses one window; all remaining phases are classical.
\end{proof}

\subsection{Coin tossing with abort}
\label{sec:coin-toss}

The ideal functionality $\ideal(\mathsf{CT})$ samples a uniform bit $c$. With two honest parties it outputs $c$ to both. With one corrupt party, it first releases $c$ to the simulator, which chooses delivery or abort; the honest party receives $c$ or $\perp$. The output need not remain uniform after conditioning on delivery.

\begin{protocolbox}
\protocolstyle
\noindent Alice and Bob use \Cref{prot:di-bc-certified}; rejection or missing messages cause terminal abort.
\begin{protocolsteps}
\item \emph{Alice.} Sample uniform $a\in\{0,1\}$ and complete its commitment phase with Bob, including the one-window preprocessing.
\item \emph{Bob.} Sample uniform $b'\in\{0,1\}$ and send it.
\item \emph{Alice.} Open $a$ through the certified opening phase.
\item \emph{Bob.} Verify the opening and send its result. On acceptance output $a\oplus b'$; otherwise output $\perp$.
\item \emph{Alice.} On acceptance output $a\oplus b'$; otherwise output $\perp$.
\end{protocolsteps}
\noindent\emph{Where each step is used.} Steps~1 and~3--4: \Cref{thm:di-bc-two-sided}; Steps~2--5: \Cref{cor:di-coin}.
\end{protocolbox}
\par\nobreak\noindent
\begin{minipage}{\linewidth}
\captionsetup{hypcap=false}
\captionof{algorithmica}{DI coin tossing: Blum's commit-and-reveal protocol~\cite{Blum83} with a one-window certified commitment.}
\label{prot:di-coin}
\end{minipage}

\begin{corollary}[One-window DI coin tossing with abort]
\label{cor:di-coin}
Under the assumptions of \Cref{thm:di-bc-two-sided}, \Cref{prot:di-coin} realizes $\ideal(\mathsf{CT})$ in the sense of \Cref{def:comp-aug-sec}, with auxiliary-input-uniform QPT simulators for either corruption and negligible error. It uses one device window and $O(\ell+\log^2\lambda)$ coordinates, with polynomial classical computation and communication, and inherits the source's constant honest-fault tolerance and constant physical leakage rate at $q=1$.
\end{corollary}
\begin{proof}
Correctness follows from commitment correctness and the uniform XOR of the honest coins.

\emph{Corrupt Alice.} Given the ideal coin $c$, the simulator runs the committer simulation of \Cref{prot:di-bc-certified-sim} and obtains $a^*$. It sends $b'=c\oplus a^*$, verifies the opening, and permits ideal delivery exactly on acceptance. The fresh $c$ makes $b'$ uniform independently of the preceding state. Except on the commitment agreement error, acceptance fixes $a=a^*$ and the honest output is $c$; the source simulator preserves the continuing state.

\emph{Corrupt Bob.} Given $c$, the simulator runs the input-independent commitment simulation. After receiving $b'$, it opens to $c\oplus b'$ and reproduces the delivery or abort decision. To compare with the real execution, first replace the honest commitment to its uniform $a$ by that simulation followed by simulated opening to $a$. \Cref{def:bc-two-sided}(ii) preserves the joint state, including any choice of $b'$ based on the commitment transcript. In this hybrid, $a$ is first used at opening. Thus $c:=a\oplus b'$ is uniform and independent of the earlier state, and $a=c\oplus b'$, exactly as in the simulator. The final acceptance message only chooses delivery or abort. All later operations are classical.
\end{proof}

\section{The three-window construction and secure computation}
\label{sec:complete-construction}

The source string commitment uses one device window, and masked CP needs one such commitment in each direction. Both preprocessings finish before the outer OT begins. This section assembles the schedule, accounts for the shared physical budget, and derives secure computation by sequential composition.

The schedule satisfies \Cref{def:compiled-device-calls}. Sizing each batch for $\ell$ suffices for a global budget $\ell$; it does not renew that budget. A device window counts one parallel input--output interval, not its circuit depth or adaptive leakage messages. Classical rounds between windows are counted separately.

With \Cref{sec:unified-parameters}, set
\begin{equation}
 \begin{gathered}
 m_0=\mathsf{Next}(\max\{100,100\ell,s_\lambda\}),\qquad
 L_0=kN_0m_0+\lceil m_0/200\rceil,\\
 q_{\rm bank}=10L_0,\qquad
 m_{\rm pre}=\mathsf{Next}(100q_{\rm bank}),\qquad
 L_{\rm pre}=kN_0m_{\rm pre}+\lceil m_{\rm pre}/200\rceil.
 \end{gathered}
 \label{eq:three-window-sizes}
\end{equation}
The capacity bound \eqref{eq:cp-bank-capacity} includes all late labels, symbol-erasure flags, and answer records. Each record is masked once, including when it participates in several joint predicates.

\begin{protocolbox}
\protocolstyle
\noindent Alice holds $(S_0,S_1)\in\{0,1\}^2$ and Bob holds $D\in\{0,1\}$. Public parameters are \eqref{eq:three-window-sizes}. Use three distinct batches satisfying \Cref{def:compiled-device-calls}; every rejection causes global abort.
\begin{protocolsteps}
\item \emph{Alice (first preprocessing).} Sample a uniform length-$q_{\rm bank}$ mask string and complete Step~1 of \Cref{prot:masked-cp} with Bob, using \Cref{prot:di-bc} on the first batch. Save the private bank state and the shared public transcript.
\item \emph{Bob (second preprocessing).} Complete the same operation with the roles reversed and fresh uniform masks, using the second batch. Save the bank for Bob-owned records.
\item \emph{Alice and Bob (transfer).} Run \Cref{prot:di-ot} on the third batch. Implement every CP record, projection, and joint predicate by Steps~2--5 of \Cref{prot:masked-cp}, using the corresponding owner's bank. On delivery or abort, discard all private call-local state as in \Cref{def:compiled-device-calls}.
\end{protocolsteps}
\noindent\emph{Where each step is used.} Steps~1--2: \Cref{thm:di-vector-commitment,thm:masked-cp}; Step~3: \Cref{cor:bc-ot-compiler,thm:unified-di-ot,thm:qowf-di-ot}.
\end{protocolbox}
\par\nobreak\noindent
\begin{minipage}{\linewidth}
\captionsetup{hypcap=false}
\captionof{algorithmica}{Complete DI OT in three device windows, including preprocessing.}
\label{prot:three-window-ot}
\end{minipage}

\begin{theorem}[Robust DI OT in three device windows]
\label{thm:qowf-di-ot}
Assume qOWFs and \Cref{def:compiled-device-calls}, with \eqref{eq:three-window-sizes} for a polynomially bounded $\ell$. Then \Cref{prot:three-window-ot} realizes one-bit OT with
\[
 \begin{gathered}
 \forall\text{ static }\mathsf P^*\in\QPT\quad
 \exists\Sim\in\QPT\quad\forall\eta^{\Rref}_{\rm adm}:\\
 \mathsf{Real}_{\mathcal P,\mathsf P^*}(\eta^{\Rref}_{\rm adm})
 \approx_c
 \mathsf{Ideal}_{\mathsf{OT},\Sim}(\eta^{\Rref}_{\rm adm}).
 \end{gathered}
\]
Simulation is auxiliary-input-uniform for either corruption and admits efficient sequential composition. Arbitrary joint device measurements are allowed. The device bounds are
\begin{equation}
 R_{\rm DI}\le3,\qquad
 L_{\rm total}=2L_{\rm pre}+L_0
 =\Theta(\ell+\log^2\lambda).
 \label{eq:three-window-resources}
\end{equation}
For the fixed constants in \eqref{eq:constant-global-rates}, a total physical-window leakage budget at most $c_{\rm global}L_{\rm total}$ is admissible regardless of its distribution among the three windows. Honest correctness has negligible error under \eqref{eq:honest-device-fault} with $q_{\rm dev}\le\nu_0:=\nu_{\rm global}/2>0$. All prescribed interparty messages are classical, and total classical computation and communication are polynomial in $\lambda$.
\end{theorem}

The constants are fixed before $\lambda$ or $\ell$ is chosen. For zero leakage the same protocol has constant noise tolerance and $\Theta(\log^2\lambda)$ device use, with no coordinate-locality restriction. The three-window bound concerns one complete OT; $K$ sequential calls use at most $3K$ windows. A global leakage rate across all $K$ calls does not alone imply the separate budget conditions at each reached call.

\subsection{Proof and resource bounds}
\label{sec:three-window-proof}

\begin{proof}[Proof of \Cref{thm:qowf-di-ot}]
\emph{Step 1: substitute both banks jointly.} \Cref{thm:di-vector-commitment} supplies state-preserving extraction, agreement with all native openings, labelled joint hiding, and honest native witnesses. Each source invocation satisfies those hypotheses with the other bank and future outer batch retained as protected state. Apply \Cref{thm:masked-cp} to both owner directions and all pending records together. Its simulation includes the intervening outer device invocation between completed CP phases. It introduces no extra device operation and preserves public requests on rejection. The resulting execution has ideal CP records and the prescribed outer device interface. The joint $\Label$ check uses the immutable offset and selector records of Bob's one bank.

\emph{Step 2: apply the outer OT theorem.} Apply \Cref{thm:unified-di-ot} in that ideal-CP execution. Its causal and entropy proofs concern the same frozen records, questions, anchors, and leakage. For any fixed final QPT test, the error is at most
\begin{equation}
 \nu_{\rm outer}+\nu_{\rm ext}+\nu_{\rm hide}
 +T_{\rm CP}\nu_{\rm ZK}+\varepsilon_{\rm ao}
 +T_{\rm CP}\varepsilon_{\rm snd}=\negl(\lambda),
 \label{eq:three-window-error}
\end{equation}
where $T_{\rm CP}\le\poly(\lambda)$ bounds the predicate phases. The source bounds include their internal simulation and entropy errors and are uniform over reached auxiliary states and induced tests. Computational errors are added at the efficient endpoints. At terminal return the honest banks' private state is discarded, leaving exactly the interface of \Cref{def:comp-aug-sec}. \Cref{thm:sequential-composition} then applies at admissible boundaries.

\emph{Step 3: correctness and tolerance.} Honest failure is bounded by the sum of the two source failures, the outer failure, and proof-completeness errors. \Cref{lem:unified-parameters} makes this sum negligible. A total of at most $\nu_{\rm global}L_{\rm total}$ physical faults gives at most $\nu L_i$ faults in each batch $i$, by its size comparison. Independent faults of probability at most $\nu_{\rm global}/2$ exceed this total only with probability $\exp[-\Omega(L_{\rm total})]$. The global leakage calculation in the same lemma charges the sum of the three window costs once, allowing arbitrary concentration in any batch.

\emph{Step 4: windows and coordinates.} The only source calls are the two preprocessings. Later CP operations use masks and classical-message ZK. Each source and the outer OT use one window. Since $N_0$ is fixed and admissible code lengths increase by factors below nine,
\[
 m_0,L_0=\Theta(\ell+s_\lambda),\qquad
 q_{\rm bank},m_{\rm pre},L_{\rm pre}=\Theta(L_0).
\]
This proves \eqref{eq:three-window-resources}. The field-code algorithms, matrix seeds, and witnesses have polynomial size. All classical primitives retain security parameter $\lambda$, and their fixed polynomial clocks bound the total resources.
\end{proof}

\begin{corollary}[A concrete polylogarithmic regime]
\label{cor:concrete-polylog-leakage}
For fixed $c\ge1$ and total leakage budget at most $\log^c\lambda$, put $a=\max\{c,2\}$. Then \Cref{prot:three-window-ot} has
\[
 R_{\rm DI}\le3,\qquad L_{\rm total}=\Theta(\log^a\lambda),
 \qquad q_{\rm dev}\le\nu_0,
\]
with negligible security and correctness errors, where $\nu_0>0$ is fixed. Classical time and communication are polynomial in $\lambda$.
\end{corollary}
\begin{proof}
Substitute $\ell=\log^c\lambda$ into \Cref{thm:qowf-di-ot}.
\end{proof}

\subsection{Sequential composition of completed calls}
\label{app:sequential-composition}

We compose completed calls by sequential replacement with efficient simulation; see also Unruh~\cite{Unruh_2}. The theorem retains the physical conditions of \Cref{def:admissible-extension,def:ot-call-boundary}. Pending records within one call use the stateful guarantee of \Cref{thm:masked-cp}.

\begin{theorem}[Efficient sequential composition]
\label{thm:sequential-composition}
Let QPT protocols $\mathcal P_f$ and $\mathcal P_g^f$ realize $f$ and $g$ as in \Cref{def:comp-aug-sec}. Under at most one static corruption, replace $q(\lambda)\leq\poly(\lambda)$ completed ideal-$f$ calls by $\mathcal P_f$. If every reached boundary in every replacement hybrid satisfies \Cref{def:ot-call-boundary}, then
\[
 \mathcal P_g^{\mathcal P_f}\ \text{realizes }g,\qquad
 \begin{aligned}
 \adv&\leq\nu_g+q\nu_f,\\
 \varepsilon^{\rm cor}&\leq
 \varepsilon_g^{\rm cor}+q\varepsilon_f^{\rm cor}.
 \end{aligned}
\]
Simulation is auxiliary-input-uniform QPT. Here $\nu_g,\nu_f$ are uniform negligible bounds for the induced QPT tests, and the correctness bound assumes joint output--reference trace-distance errors at most $\varepsilon_f^{\rm cor}$ at each reached honest call and $\varepsilon_g^{\rm cor}$ for the outer protocol.
\end{theorem}

\begin{proof}[Proof of \Cref{thm:sequential-composition}]
\emph{Step 1: setting.} Fix a corrupted party, a QPT outer adversarial machine $\mathcal A$, and a polynomial upper bound on the number of calls. Calls are completed and non-overlapping; all messages, including the subroutine's proof and opening messages, are processed before return. A prescribed ideal resource used within a call is retained with its specified interface. Distinct invocation labels prevent messages from being interpreted in another call. The corruption is static, although the subroutine's sender/receiver roles may be selected by the outer protocol whenever both component guarantees cover the corresponding role assignment.

\emph{Step 2: call slots and their simulators.} Pad to $q$ potential call slots, with an immutable classical reach flag for each slot. An unreached slot is the identity. The number of reached calls, their inputs, and the decision to continue may depend on preceding outputs and transcripts, subject to boundary admissibility. The corrupted call machine is the fixed code running $\mathcal A$ during the subroutine; the slot index, role, and all adaptive history are classical or quantum auxiliary input. The single-call quantifier order supplies a simulator for this fixed machine that works uniformly for those inputs. The original adversarial code and the call simulator have polynomial bounds independent of the slot. Simulators run successively, return the prescribed corrupt continuation state, and are not recursively substituted into their own code. Consequently, polynomially many such calls still give a QPT machine.

\emph{Step 3: the hybrids.} Let $H_j$, for $0\leq j\leq q$, replace the first $j$ potential calls by the ideal-$f$ interface and these call simulators, leaving later calls real. Fix $j$. The prefixes of $H_{j-1}$ and $H_j$ agree through the boundary of slot $j$. Let $\mathsf C_j$ be the complete corrupted continuation state, $\mathsf H_j$ the honest subroutine input--output interface, and $U_j$ all remaining outer registers untouched during the call, including scheduling and control state. Apply the single-call guarantee with reference $U_j\Rref$. Call-local honest scratch and device systems may be traced out only if they are irrevocably unavailable to the continuation. In particular, a discarded device cannot subsequently be reused as an unmodeled communication or memory channel.

\emph{Step 4: one replacement.} Fix a final QPT distinguisher $\mathcal D$. The common suffix of the two hybrids, followed by $\mathcal D$, is a QPT test $\mathcal D_j$ of the complete post-call state. It includes all later adaptive choices and any abort continuation. The suffix may act on $U_j$ after return; the external $\Rref$ stays untouched until the final distinguishing test. The theorem's hypothesis makes the normalized reached boundary state admissible even when earlier slots have been simulated. The call simulator does not depend on that state, its unknown description, or the reference. Hence single-call security bounds this test's advantage by a negligible $\nu_j(\lambda)$, uniformly over the reached boundary state.

\emph{Step 5: weighting by the reach probability.} Write $p_j=\Pr[\mathsf{Reach}_j=1]$, the same in both prefixes. On the unreached substate, the hybrids agree. If $p_j>0$, normalize that branch only to invoke the uniform single-call guarantee and multiply back by $p_j$. Thus
\begin{equation}
 \left|\Pr[\mathcal D(H_{j-1})=1]
       -\Pr[\mathcal D(H_j)=1]\right|
 \leq p_j\nu_j(\lambda)\leq\nu_j(\lambda).
 \label{eq:composition-call-advantage}
\end{equation}

Aborts inside a reached call are ordinary output blocks in this same comparison, with the full corrupt state and reference retained. The honest abort flag is absorbing, so no later honest call is made on that branch; a corrupt party's allowed continuing computation remains part of the common suffix.

\emph{Step 6: summation.} The call tests form a polynomially bounded family, with the index and adaptive control carried as auxiliary data for a fixed call machine. Nonuniform QPT security, uniform over auxiliary states, also covers a sequence of indices maximizing the adjacent advantage. Therefore their maximum $\nu_f(\lambda)$ is negligible for the fixed outer adversary and distinguisher. \Cref{lem:efficient-composition-hybrids} yields
\[
 \left|\Pr[\mathcal D(H_0)=1]-\Pr[\mathcal D(H_q)=1]\right|
 \leq\sum_{j=1}^q\nu_j(\lambda)
 \leq q(\lambda)\nu_f(\lambda).
\]

\emph{Step 7: the outer simulator.} The final hybrid $H_q$ is $\mathcal P_g^f$ with one stateful QPT adversary comprising $\mathcal A$ and the call simulators. The outer security theorem gives one QPT simulator $\Sim_g$, using only $\ideal(g)$ and the resources explicitly allowed in that outer theorem. In particular, it has no undeclared ideal-$f$ oracle. The declared device interfaces of \Cref{def:comp-aug-sec} remain explicit resources of the outer ideal experiment. This simulator is chosen independently of the initial extension and $\Rref$. If its induced distinguishing error is $\nu_g(\lambda)$, the final advantage is at most
\[
                   \nu_g(\lambda)+q(\lambda)\nu_f(\lambda).
\]
The simulations are chosen before the final distinguisher; only the negligible advantage bounds may depend on it. This proves the required auxiliary-input-uniform quantifier order and worst-case polynomial running time. Repeat the argument for the other corruption case.

\emph{Step 8: correctness.} For correctness, use the same replacements with both parties honest and the assumed joint input--output/reference correctness guarantees. At each reached slot, these compare the real call with its ideal output while retaining $U_j\Rref$. Contractivity, weighting by reach probability, and telescoping give $\varepsilon_g^{\rm cor}+q\varepsilon_f^{\rm cor}$ in the same distance convention as the component bounds. In the OT instantiations here, the honest outputs are deterministic functions of the classical inputs. The untouched input--reference marginal is unchanged by a complete honest call, and the incorrect-output/abort substate has mass at most the stated failure probability. Consequently its trace distance $\tfrac12\|\cdot\|_1$ from the ideal joint output is at most that probability. Thus the existing honest-error bounds supply the reference-preserving correctness needed here, including honest aborts.
\end{proof}

\paragraph{Application to the OT protocol.}
\Cref{thm:unified-di-ot} provides the single-call simulation and correctness guarantees. At every invocation, retain all corrupt device residuals and place untouched honest outer state in $\Rref$; only discarded honest call-local systems are traced out. Each call has fresh honest coins, distinct record handles, and its own prescribed device window. Every call satisfies its own prescribed global leakage budget and, for honest correctness, its noise promise. The next call's state must again be admissible in every hybrid. No arbitrary device reuse or interleaving is inferred from the single-call result.

For the compiled OT, first apply \Cref{thm:masked-cp,cor:bc-ot-compiler} within each OT call. Their direct simulation retains pending record states across sequential commit phases and delayed openings. This is the stateful masked-CP guarantee, rather than an application of the completed-call theorem to overlapping record lifetimes. It establishes the required efficient OT call simulator; the theorem then composes completed OT calls.

\subsection{The MPC corollary}
\label{app:mpc-proof}
The classical OT-to-MPC reduction uses polynomially many completed pairwise OT calls. Its statistical hybrid and quantum lifting give the outer simulator; the three-window OT theorem implements each pairwise call under the stated device conditions.

\begin{corollary}[Device-independent multiparty computation]
\label{cor:di-mpc}
Assume qOWFs and private authenticated classical channels.
Suppose that every reached OT call in every replacement hybrid satisfies the conditions in \Cref{def:ot-call-boundary,def:compiled-device-calls}, and that honest--honest calls also satisfy the corresponding correctness conditions.
For every fixed $p\geq2$ and classical PPT functionality $\mathcal F$ on $p$ parties, with public input/output lengths and corrupt-output-first, adversarially suppressible honest delivery, there is a DI protocol $\Pi_{\mathcal F}$ with security with abort against static corruption of any proper subset $C\subsetneq[p]$:
\[
 \begin{gathered}
 \forall\mathcal A\in\QPT\quad\exists\Sim\in\QPT
 \quad\forall\eta^{\Rref}_{\rm adm}:\\
 \mathsf{Real}_{\Pi_{\mathcal F},\mathcal A}(\eta^{\Rref}_{\rm adm})
 \approx_c \mathsf{Ideal}_{\mathcal F,\Sim}(\eta^{\Rref}_{\rm adm}).
 \end{gathered}
\]
$\Pi_{\mathcal F}$ uses $\poly(\lambda)$ sequential pairwise DI OT calls and $\poly(\lambda)$ total resources. Each admissible OT call has the guarantees of \Cref{thm:qowf-di-ot}, with $\varepsilon_{\rm corr}=\negl(\lambda)$ under the stated honest-device assumptions.
\end{corollary}

\noindent The ideal experiment retains the declared device resources of \Cref{def:comp-aug-sec}, and an absorbing abort suppresses future outputs, not outputs already delivered. The three-window bound and constant leakage rate apply per OT call. A single global fractional leakage budget for an arbitrary number of calls requires an additional allocation condition.

\begin{proof}[Proof of \Cref{cor:di-mpc}]
\emph{Step 1: the outer protocol.} Use the classical statistical OT-hybrid MPC protocol of Ishai, Prabhakaran, and Sahai~\cite{IPS08}, including their statistically UC-secure string-OT to bit-OT reduction~\cite[Section~2]{IPS08}. Use its public polynomial clock and the stated length/delivery convention. An abort before OT input submission withholds the corrupt input; an abort after corrupt receiver output is a stop after that standard OT output. A missing completion signal triggers the prescribed absorbing outer abort. Honest final delivery follows the same adversarially suppressible interface. For reactive outputs this rule applies only to future delivery. Thus the staged OT interface in \Cref{def:idealOT} adds no delivery guarantee to the target functionality. Unruh's quantum lifting theorem makes this classical hybrid quantum secure \cite[Theorem~15 and Section~6]{Unruh_2}. Schedule its pairwise calls sequentially and replace each by \Cref{thm:qowf-di-ot}. All other honest outer registers remain untouched during a call. Use private authenticated channels, pad each call's traffic to its public polynomial bound, and release only its prescribed input--output interface to the caller. Internal transcripts and honest call-local device systems are retired. The same current device share is supplied through its declared interface in each replacement; future shares remain protected. The final simulator retains these interfaces, so the conclusion is resource-relative.

\emph{Step 2: calls with one corrupted endpoint.} Let $C\subsetneq[p]$ be the corrupted set. In a call with exactly one endpoint in $C$, the coalition plays that endpoint: its joint workspace, including the workspaces and device shares of the other corrupted parties, is the corrupt workspace of \Cref{def:admissible-extension}. Retain the honest parties' states and the pending outer registers in the reference and apply \Cref{thm:sequential-composition}. Classical call inputs may be correlated with this reference. Conditional on each input value $x$ and reached boundary, the stated invocation hypothesis supplies an admissible extension. The same simulator and the same negligible bound $\nu_f$ work for all $x$. For any induced QPT test $\mathcal D$, put
\[
 d_x:=\Pr[\mathcal D(\mathsf{Real}_x)=1]
       -\Pr[\mathcal D(\mathsf{Ideal}_x)=1].
\]
Linearity and the triangle inequality give
\begin{align}
 \bigl|\Pr[\mathcal D(\mathsf{Real})=1]
       -\Pr[\mathcal D(\mathsf{Ideal})=1]\bigr|
 &=\Bigl|\sum_xp_xd_x\Bigr|\\
 &\leq\sum_xp_x|d_x|
 \leq\sum_xp_x\nu_f=\nu_f.
 \label{eq:mpc-input-averaging}
\end{align}
Thus conditioning for this calculation imposes no independence on the outer classical data and requires no input-dependent simulator.

\emph{Step 3: calls between corrupted parties.} In a call with both endpoints in $C$, both roles and both device shares belong to the adversary, so the call is part of its own computation. In the OT-hybrid execution, the adversary supplies both inputs to such an ideal call and learns nothing that it does not already hold. These calls are therefore not replaced and contribute no error term.

\emph{Step 4: calls between honest parties.} A call with two honest endpoints needs a separate argument. Its complete DI implementation, including source commitment calls, satisfies the stated honest-device correctness assumptions. Hence its probability of abort or incorrect delivery is at most a uniform negligible $\varepsilon_f$. Its private, isolated execution acts only on the two honest laboratories and call-local registers. After those registers are discarded, no-signalling leaves the unconditioned external reference unchanged. For fixed classical inputs $x$, let $o_x$ denote the correct ideal output and write the real output--reference state as $\sigma^x_{O\Rref}=\sum_o|o\rangle\!\langle o|\otimes\sigma_{\Rref}^{x,o}$. Then $\sum_o\sigma_{\Rref}^{x,o}=\rho_{\Rref}^x$ and $\sum_{o\ne o_x}\operatorname{Tr}\sigma_{\Rref}^{x,o}\leq\varepsilon_f$. Consequently,
\begin{align}
 \tfrac12\bigl\|\sigma^x_{O\Rref}
       -|o_x\rangle\!\langle o_x|\otimes\rho_{\Rref}^x\bigr\|_1
 &=\tfrac12\bigl\|\sigma_{\Rref}^{x,o_x}-\rho_{\Rref}^x\bigr\|_1
       +\tfrac12\sum_{o\ne o_x}\operatorname{Tr}\sigma_{\Rref}^{x,o}\\
 &=\sum_{o\ne o_x}\operatorname{Tr}\sigma_{\Rref}^{x,o}
 \leq\varepsilon_f.
 \label{eq:mpc-honest-pair}
\end{align}
Averaging over $x$ preserves this bound. In particular, we do not infer independence of an abort flag from no-signalling: abort branches are included in the negligible error above. This is why the honest-device hypothesis on honest--honest calls is explicit in the corollary.

\emph{Step 5: summation.} For at most $q=\poly(\lambda)$ reached and replaced calls, the hybrid inequality is
\[
 \operatorname{Adv}_{\rm MPC}
 \leq\nu_{\rm outer}+q(\nu_f+\varepsilon_f)=\negl(\lambda).
\]
Uniform bounds apply at every replacement boundary and retain the probability weights of reached branches. The composed simulator is QPT. The outer protocol uses only trusted classical operations; all quantum device use occurs in the substituted calls. A callwise union bound on honest failures preserves the stated robustness under the per-call hypotheses. A single fractional budget over polynomially many calls requires a separate bound on its concentration in each call. The channel hypothesis is the secure-channel model used in~\cite[Section~2]{Unruh_2}.
\end{proof}

\section{Selector-aware repetition and hidden-string entropy}
\label{app:selector-repetition}

This section supplies the entropy input to the OT and commitment proofs. We define the game and prove selector-aware threshold repetition, then derive the hidden-string bound. The repetition proof uses the product-anchored threshold theorem of Kundu and Tan (\Cref{fact:kundu-tan}).

\subsection{Selector-aware repetition with late erasures}
\label{sec:selector-repetition-input}

In $G_1^{\rm sel}$ the verifier sends independent uniform trits $\Xraw,Y$ to Charlie and Dave. Charlie returns an even row $\Araw$; Dave's $Y$-controlled instrument freezes $F=(B,J)$, where $B$ is an odd column and $J\in\{0,1\}$, and retains arbitrary quantum memory. Only then does the verifier sample $C\gets\mathrm{Bern}(1-\delta_*)$ and send Dave $\widehat X=\Xraw$ if $C=1$, and $\widehat X=*$ otherwise. Dave returns $\widehat A[0,1]$; late maps preserve $F$ and its inaccessible classical copy. Put
\[
\begin{gathered}
 W=\begin{cases}
 1,&C=0,\\
 \Id\{\Araw[0,1]=\widehat A[0,1]\}
       \MS(\Xraw,Y,\Araw,B),&C=1,
 \end{cases}
\\ Q=JW,\qquad
 L=C(1-\MS(\Xraw,Y,\Araw,B)).
\end{gathered}
\]
The parallel game uses product questions and fresh independent $C_i$, with arbitrary entanglement and joint operations across coordinates.

\begin{theorem}[Selector-aware constant-threshold repetition]
\label{thm:restricted-order-repetition}
For every $\delta_*\in(0,1)$, $\rho\in(0,1]$, and $\gamma_{\rm sel}\in(((1+\delta_*)/2)\rho,\rho)$, there exist $\eta_{\rm sel},c_{\rm sel}>0$ and $n_0\in\mathbb N$ such that $\forall\mathcal S$ allowed and $\forall n\geq n_0$,
\begin{equation}
 \Pr_{\mathcal S}\!\left[
 \sum_{i=1}^nJ_i\leq\rho n,\quad
 \sum_{i=1}^nL_i\leq\eta_{\rm sel}n,\quad
 \sum_{i=1}^nJ_iW_i\geq\gamma_{\rm sel}n
 \right]\leq2^{-c_{\rm sel}n}.
 \label{eq:selector-repetition}
\end{equation}
The allowed strategies are those defined above; no lower bound on $\sum_iJ_i$ is assumed.
\end{theorem}
The \hyperref[par:intro-core]{proof overview} gives the main idea.

Honest Alice always queries $\Xraw_j$ and records $\Araw_j$ before using $C_j$. Her committed table has $(X_j,A_j)=(\Xraw_j,\Araw_j)$ on active coordinates and $(X_j,A_j)=(0,000)$ on anchors. Raw anchor questions are absent from subsequent controller messages; $C$ does not reach either device. Thus the protocol supplies precisely this late-erasure interface. Private presampling of unused $C$ is equivalent for the sparse-test argument.

\paragraph{Why the Magic Square game and not CHSH.}
CHSH is far easier to implement, and it certifies entropy in the usual DI protocols, but it does not survive the two-round structure above. CHSH is a projection game: once Dave holds his answer $b$ and learns both questions, Charlie's winning answer is determined, $a=b\oplus xy$. A strategy that wins therefore leaves no randomness in Charlie's answer given Dave's, and the hidden string, which is what a guessing Dave must fail to produce after learning $\Xraw$, would have no entropy at all. Magic Square is not a projection game: Charlie's row has two free bits, Dave's column fixes only the entry at the intersection, and Fu and Miller~\cite{FM-BRE} show that the other entry stays close to uniform in every strategy that wins with probability close to one. That second bit is the hidden string. The proof uses the game through this one-copy bound (\Cref{lem:conservative-residual-guessing}) and, in the hidden-string lemma, through the row parity that reconstructs the third entry. A replacement game therefore needs both a comparable residual-guessing bound and a reconstruction rule with the same role in the hidden-string argument.

\paragraph{Notation for the proof.}
The game is that of \Cref{sec:selector-repetition-input}. We write $Z^{\rm late}$ for Dave's second-round question, that is, $Z^{\rm late}=\bot$ if $C=0$ and $Z^{\rm late}=\Xraw$ if $C=1$, so that $\bot$ is the erasure symbol $*$ of $\widehat X$; and we use the scores
\begin{align}
 W&=(1-C)+C\,\mathbf1\{\widehat A[0,1]=\Araw[0,1]\}
       \MS(\Xraw,Y,\Araw,B),\\
 Q&=JW,\qquad
 L=C\bigl(1-\MS(\Xraw,Y,\Araw,B)\bigr).
 \label{eq:late-erasure-scores}
\end{align}
In the $n$-copy game, the first-round questions are sampled productwise, both parties may use arbitrary entangled states and joint first-round operations, Dave may use one joint second-round operation on all the revealed questions and his residual state, and $C$ is sampled after all first-round outputs, including $J$, are fixed.

\begin{lemma}[Conservative residual guessing bound]
\label{lem:conservative-residual-guessing}
Suppose a Magic Square strategy with even-row and odd-column outputs and independent uniform questions has average loss $\varepsilon$. For fixed questions $x,y$ and $k\ne y$, after producing its column, Dave's probability of guessing $\Araw[k]$ from arbitrary retained quantum memory and the revealed $x$ is at most $1/2+15\sqrt\varepsilon$.
\end{lemma}
\begin{proof}
We make the normalization in the reflection argument of Fu--Miller~\cite[Section~3 and Appendix~B]{FM-BRE} explicit and use a conservative constant. Purify the state to $|\psi\rangle$ and dilate the local instruments to projective measurements, retaining all local ancillas. Write $F_{ij}$ and $G_{ij}$ for Alice's and Dave's bit reflections, respectively, and let $\varepsilon_{ij}$ be the loss on questions $(i,j)$, so $\sum_{ij}\varepsilon_{ij}=9\varepsilon$. The parity and consistency relations give, for $r\ne x$ and $k\ne y$,
\[
\begin{aligned}
 \|(F_{ij}-G_{ij})|\psi\rangle\|&=2\sqrt{\varepsilon_{ij}},\\
 \|\{F_{xk},F_{ry}\}|\psi\rangle\|
 &\le2\sum_{ij}\sqrt{\varepsilon_{ij}}
 \le18\sqrt\varepsilon.
\end{aligned}
\]
For clarity, the second estimate follows from the following nine-transfer chain on $|\psi\rangle$; an arrow labelled $ij$ has norm error at most $2\sqrt{\varepsilon_{ij}}$, and unlabelled equalities are exact:
\[
\begin{aligned}
 F_{00}F_{11}
 &\xrightarrow{11}F_{00}G_{11}
   =-F_{02}F_{01}G_{21}G_{01}
   \xrightarrow{01}-F_{02}G_{21}\\
 &\xrightarrow{02}-G_{21}G_{02}
   =G_{21}G_{22}G_{12}
   \xrightarrow{12}F_{12}G_{21}G_{22}\\
 &\xrightarrow{22}F_{12}F_{22}G_{21}
   \xrightarrow{21}F_{12}F_{22}F_{21}
   =F_{12}F_{20}\\
 &\xrightarrow{20}F_{12}G_{20}
   =-F_{11}F_{10}G_{00}G_{10}
   \xrightarrow{10}-F_{11}G_{00}
   \xrightarrow{00}-F_{11}F_{00}.
\end{aligned}
\]
Row and column permutations give the general indices above.

Fix $r\ne x$, set $F=F_{xk}$, $H=F_{ry}$, $G=G_{ry}$, and $P_s=(I+(-1)^sF)/2$. Let $\Pi_b$ be Dave's full-column projectors for input $y$, so $G\Pi_b=(-1)^{b[r]}\Pi_b$. For $u_b=HP_0\Pi_b|\psi\rangle$ and $v_b=(-1)^{b[r]}P_1\Pi_b|\psi\rangle$, orthogonality gives
\begin{align*}
 \left(\sum_b\|u_b-v_b\|^2\right)^{1/2}
 &=\left\|\left(\tfrac12(HF+FH)+P_1(H-G)\right)|\psi\rangle\right\|
 \le15\sqrt\varepsilon,\\
 \sum_b(\|u_b\|+\|v_b\|)^2
 &\le2\sum_b(\|P_0\Pi_b|\psi\rangle\|^2
             +\|P_1\Pi_b|\psi\rangle\|^2)=2.
\end{align*}
Let $\sigma_s$ be Dave's subnormalized state when the target bit is $s$, retaining his classical column outcome. Partial trace over Alice removes the conjugation by $H$, so its $b$-blocks are the partial traces of $u_bu_b^\dagger$ and $v_bv_b^\dagger$. Since $\|uu^\dagger-vv^\dagger\|_1\le(\|u\|+\|v\|)\|u-v\|$, contractivity and Cauchy--Schwarz imply $\|\sigma_0-\sigma_1\|_1\le15\sqrt{2\varepsilon}$. \Cref{fact:helstrom} therefore gives $\pg\le1/2+(15/\sqrt2)\sqrt\varepsilon \le1/2+15\sqrt\varepsilon$. Any original instrument's residual state is obtained from the retained dilation by a local channel, so the same bound applies to it. An input-independent reference may be included in Dave's initial share; his copied early operation can act as the identity on that reference.
\end{proof}

\begin{lemma}[Selected one-copy guessing]
\label{lem:late-erasure-selected}
For a Magic-Square strategy with independent $\Xraw,Y\gets_U\{0,1,2\}$, let $\varepsilon=\Pr[\MS(\Xraw,Y,\Araw,B)=0]$. For every $Y$-controlled first instrument with output $(B,J)$ and arbitrary retained quantum memory, every subsequent row-guessing measurement given $\Xraw$ and that memory satisfies
\begin{equation}
 \Pr[J=1,\ \widehat A[0,1]=\Araw[0,1]]
 \leq \tfrac12\Pr[J=1]+15\sqrt{\varepsilon}.
 \label{eq:late-erasure-selected-bound}
\end{equation}
For the scores in~\eqref{eq:late-erasure-scores},
\begin{equation}
 \mathbb E Q\leq s_*\mathbb E J+15\sqrt{\mathbb E L},
 \qquad s_*:=\frac{1+\delta_*}{2}.
 \label{eq:late-erasure-one-copy-score}
\end{equation}
\end{lemma}
\begin{proof}
We use \Cref{lem:conservative-residual-guessing}, a conservative form of the residual-guessing argument of Fu--Miller~\cite{FM-BRE}. Coarsen the first instrument's public output from $(B,J)$ to $B$ and retain $J$ as part of Dave's memory. The lemma includes arbitrary instruments and all retained ancillas, so its pointwise estimate applies to this residual state.

For each $y\in\{0,1,2\}$, fix $k(y)=(y+1)\bmod 3$.  Complete the guessed two-bit row to even parity, writing $\widehat A[2]=\widehat A[0]\oplus\widehat A[1]$. Consider the following allowed bit guess: if $J=1$, output $\widehat A[k(y)]$; if $J=0$, output an independent fair bit. Write $p=\Pr[J=1]$.  Averaging the cited pointwise inequality over the uniform first questions gives
\begin{align}
 &\Pr[J=1,\widehat A[0,1]=\Araw[0,1]]+\tfrac12(1-p)
 \nonumber\\
 &\quad\leq
 \Pr[J=1,\widehat A[k(Y)]=\Araw[k(Y)]]+\tfrac12(1-p)
 \nonumber\\
 &\quad=\Pr[\text{the completed bit guess is correct}]
 \leq\tfrac12+15\sqrt{\varepsilon}.
 \label{eq:late-erasure-fair-completion}
\end{align}
Subtracting $(1-p)/2$ proves~\eqref{eq:late-erasure-selected-bound}.

Imposing the additional Magic Square check can only reduce its left side. Since $C$ is independent of the complete first-round experiment, the same $p$ occurs on its two branches and $\mathbb E L=(1-\delta_*)\varepsilon$. Thus
\begin{align}
 \mathbb E Q
 &\leq\delta_*p+(1-\delta_*)\bigl(p/2+15\sqrt{\varepsilon}\bigr)
 \nonumber\\
 &=s_*p+15\sqrt{1-\delta_*}\sqrt{\mathbb E L}
 \leq s_*p+15\sqrt{\mathbb E L}.
\end{align}
\end{proof}

\begin{fact}[Two-round product-anchored threshold theorem {\cite[Definition~3 and Theorem~29]{kundu_tan}}]
\label{fact:kundu-tan}
Let $\mathcal G$ be a two-round game with finite alphabets in which the first-round questions to Charlie and Dave are independent, both first-round answers are fixed before the second round, and only Dave receives a second-round question. That question is $\bot$ with a fixed probability $\delta_*>0$, independently of the first-round questions, and otherwise consists of Charlie's first-round question together with an additional finite-valued question. Let $V_i$ indicate that copy $i$ is won. For every $v\geq\omega^*(\mathcal G)$ and every fixed $\xi>0$ with $v+\xi<1$, there are constants $c_{\rm KT}>0$ and $n_{\rm KT}$, depending only on $\delta_*$, $\xi$, and the alphabet sizes, such that every strategy for $n$ parallel copies, with arbitrary entanglement and arbitrary quantum memory retained between the two rounds, satisfies
\begin{equation}
 \Pr\left[\sum_{i=1}^n V_i\geq(v+\xi)n\right]
 \leq 2^{-c_{\rm KT}n}\qquad(n\geq n_{\rm KT}).
 \label{eq:late-erasure-kt}
\end{equation}
\end{fact}

Dave's first-round answer alphabet below is $(B,J)$, which is finite. We take the first-round predicate identically one and place all checks in the second-round predicate.

\paragraph{Proof of the threshold bound.}
Recall $s_*=(1+\delta_*)/2$, so the hypothesis of \Cref{thm:restricted-order-repetition} reads $s_*\rho<\gamma_{\rm sel}<\rho$. All constants chosen below depend only on $\delta_*,\rho,\gamma_{\rm sel}$.

\begin{proof}[Proof of \Cref{thm:restricted-order-repetition}]
\emph{Goal and plan.} We must show that a strategy which keeps its Magic Square losses below $\eta_{\rm sel}n$, selects at most $\rho n$ coordinates, and is correct on at least $\gamma_{\rm sel}n$ of them, succeeds with probability $2^{-\Omega(n)}$. The tool is \Cref{fact:kundu-tan}, which bounds the probability of winning many copies of a product-anchored game. Two things stand between the selector game and that class. First, the selector is Dave's own choice, so "correct on the selected coordinates" is not a per-copy predicate; Step~1 repairs this by scoring an unselected coordinate with a fair coin, which makes the selection cost Dave exactly what a coin would. Second, the threshold theorem bounds wins, not losses; Step~1 also mixes the Magic Square test with the guessing test at a ratio $M:1$, so that a loss costs $M$ points and a selected success earns one. The value of the lifted game is then below one, by the one-copy bound of \Cref{lem:conservative-residual-guessing}. Steps~2 and~3 show that the selector event forces a large score in the lifted game, and Step~4 applies the threshold theorem. The constants are chosen so that the $15\sqrt\varepsilon$ guessing advantage minus the $M\varepsilon$ loss penalty is at most $\zeta$ for every $\varepsilon$, as shown by the completed square below.

Put $g=\gamma_{\rm sel}-s_*\rho>0$, and choose
\begin{equation}
 M=\left\lceil\frac{450}{g}\right\rceil+1,\quad \bar M=M+1,\quad
 \zeta=\frac{225}{4M}<\frac g8,\quad
 u=\frac g8,\quad \eta_{\rm sel}=\frac{g}{8M}.
 \label{eq:late-erasure-constants}
\end{equation}
Completing the square gives, for every $\varepsilon\geq0$,
\begin{equation}
 15\sqrt\varepsilon-M\varepsilon
 =\frac{225}{4M}-M\left(\sqrt\varepsilon-\frac{15}{2M}\right)^2
 \leq\zeta.
 \label{eq:late-erasure-square}
\end{equation}

\emph{Step 1: a binary product-anchored game.} We define the lifted game and bound its value. Keep the same first questions and first answers.  On $C=0$, accept. On $C=1$, the verifier samples an independent mode $T$ with $\Pr[T=\mathsf{test}]=M/\bar M$ and $\Pr[T=\mathsf{guess}]=1/\bar M$, and an independent fair bit $Q$. It gives Dave the second question $(\Xraw,T,Q)$, and checks
\begin{equation}
 V=\begin{cases}
 1,&C=0,\\
 \MS(\Xraw,Y,\Araw,B),
          &C=1,\ T=\mathsf{test},\\
 \mathbf1\{Q=0\},&C=1,\ T=\mathsf{guess},\ J=0,\\
 W,&C=1,\ T=\mathsf{guess},\ J=1.
 \end{cases}
 \label{eq:late-erasure-binary-lift}
\end{equation}
This is a game of the class in \Cref{fact:kundu-tan}, and the timing is what puts it there: the anchor $C$ is sampled after $(B,J)$ are frozen, so it is independent of both first questions and of the selector; all output alphabets are finite; and neither $C$ nor $T,Q$ is available before $(B,J)$ is fixed. Had the erasures been revealed to Dave before the selector, Dave could have selected only unerased coordinates and the anchor would no longer be independent. The coins $T,Q$ are part of the second-round question, so the predicate is a deterministic function of questions and answers.

For any one-copy strategy let $p=\Pr[J=1]$, and let $\varepsilon$ be its first-round Magic Square loss.  The first-round experiment does not depend on $T,Q,C$.  On the guess branch, the final measurement may depend on $Q$; Lemma~\ref{lem:late-erasure-selected} holds for each of these measurements.  Therefore
\begin{align}
 \Pr[V=1\mid C=1]
 &\leq\frac{M(1-\varepsilon)+(1-p)/2+p/2+15\sqrt\varepsilon}{\bar M}
 \nonumber\\
 &=\frac{M+1/2+15\sqrt\varepsilon-M\varepsilon}{\bar M}
 \leq\frac{M+1/2+\zeta}{\bar M}.
 \label{eq:late-erasure-lift-value}
\end{align}
Consequently
\begin{equation}
 \omega^*(\mathcal G)\leq
 v:=\delta_*+(1-\delta_*)\frac{M+1/2+\zeta}{\bar M}<1.
 \label{eq:late-erasure-v}
\end{equation}
Here $\zeta<g/8<1/16$, since $g<(1-s_*)\rho\leq(1-\delta_*)/2$.

\emph{Step 2: the two anchor counts.} The selector event speaks of all coordinates, whereas the lifted score only counts unerased ones, so we need the erasures to be typical both overall and inside the selector. For the given parallel strategy, define
\begin{equation}
 N_{\rm anc}=\sum_i(1-C_i),\qquad
 N_{J,0}=\sum_iJ_i(1-C_i).
\end{equation}
Conditional on the entire frozen $J$, the bits $1-C_i$ remain independent Bernoulli variables of mean $\delta_*$. Hence the additive form of \Cref{fact:Chernoff} gives
\begin{align}
 \Pr[N_{\rm anc}<\delta_*n-un]&\leq e^{-2u^2n},\\
 \Pr[N_{J,0}>\delta_*\textstyle\sum_iJ_i+un]&\leq e^{-2u^2n}.
 \label{eq:late-erasure-counts}
\end{align}
For the second inequality, a fixed selector of size $k>0$ gives the bound
\[
 e^{-2u^2n^2/k}\leq e^{-2u^2n}.
\]
At $k=0$, the bad event is empty.  Write $\mathcal T$ for the two typical events.  Thus $\Pr[\mathcal T^c]\leq2e^{-2u^2n}$.

\emph{Step 3: selected success implies a large binary score.} We show that on the selector event, with typical erasures, the lifted score exceeds the threshold $(v+\xi)n$ by a constant fraction of $n$. Extend the given parallel strategy to the binary game by ignoring all the added $T,Q$ coins in its final guessing map.  It receives the same erased question vector, so this is a legal strategy with the same distribution of $(\Xraw,Y,\Araw,B,J,C,\widehat A)$. Conditional on that full classical transcript, the binary scores $V_i$ are independent.  Let $\mu_V$ be their conditional mean sum.  Set $N_J=\sum_iJ_i$, $N_Q=\sum_iJ_iW_i$, and $N_L=\sum_iL_i$. Expanding the four cases in~\eqref{eq:late-erasure-binary-lift} yields
\begin{align}
 \bar M\mu_V
 &=\bar M N_{\rm anc}+M(n-N_{\rm anc}-N_L)
     +\tfrac12(n-N_{\rm anc}-N_J+N_{J,0})+(N_Q-N_{J,0})
 \nonumber\\
 &=(M+\tfrac12)n+\tfrac12N_{\rm anc}-\tfrac12N_J
       -\tfrac12N_{J,0}+N_Q-MN_L,\\
 \bar M vn&=(M+\tfrac12)n+\tfrac12\delta_*n
            +(1-\delta_*)\zeta n.
 \label{eq:late-erasure-mean-expansion}
\end{align}
Let $\mathcal E$ be the event in~\eqref{eq:selector-repetition}. On $\mathcal E\cap\mathcal T$, the difference satisfies, line by line,
\begin{align}
 \bar M(\mu_V-vn)
 &=\tfrac12(N_{\rm anc}-\delta_*n)+N_Q-\tfrac12N_J
          -\tfrac12N_{J,0}-MN_L-(1-\delta_*)\zeta n
 \nonumber\\
 &\geq N_Q-s_*N_J-MN_L-un-(1-\delta_*)\zeta n
 \nonumber\\
 &\geq\bigl(g-M\eta_{\rm sel}-u-(1-\delta_*)\zeta\bigr)n
 \geq\frac{5g}{8}n.
 \label{eq:late-erasure-mean-gap}
\end{align}
Take $\xi=g/(4\bar M)$.  It is positive and satisfies $v+\xi<1$ because
\begin{equation}
 1-v=\frac{(1-\delta_*)(1/2-\zeta)}{\bar M}
 >\frac{7(1-\delta_*)}{16\bar M}
 >\frac{g}{4\bar M}.
\end{equation}
The mean in~\eqref{eq:late-erasure-mean-gap} exceeds $(v+\xi)n$ by at least $3gn/(8\bar M)$. The additive form of \Cref{fact:Chernoff}, applied conditionally on the transcript, therefore gives
\begin{equation}
 \Pr\left[\mathcal E\cap\mathcal T\cap
      \left\{\sum_iV_i<(v+\xi)n\right\}\right]
 \leq\exp\left(-2\left(\frac{3g}{8\bar M}\right)^2n\right).
 \label{eq:late-erasure-coin-tail}
\end{equation}

\emph{Step 4: apply \Cref{fact:kundu-tan}.} A score above $(v+\xi)n$ is exactly what the threshold theorem forbids, so the selector event is contained in the union of the three rare events. Combining~\eqref{eq:late-erasure-kt}, \eqref{eq:late-erasure-counts}, and~\eqref{eq:late-erasure-coin-tail}, we obtain
\begin{equation}
 \Pr[\mathcal E]
 \leq 2e^{-2u^2n}+2^{-c_{\rm KT}n}
       +\exp\left(-2\left(\frac{3g}{8\bar M}\right)^2n\right).
 \label{eq:late-erasure-final-tail}
\end{equation}
Let
\begin{equation}
 c_0=\min\left\{\frac{2u^2}{\ln2},\ c_{\rm KT},\
                 \frac{2(3g/(8\bar M))^2}{\ln2}\right\}>0,
 \qquad c_{\rm sel}=c_0/2.
\end{equation}
The right side of~\eqref{eq:late-erasure-final-tail} is at most $4\cdot2^{-c_0n}\leq2^{-c_{\rm sel}n}$ whenever $n\geq4/c_0$. Increase $n_0$ to cover this condition and the finite-size requirement of~\eqref{eq:late-erasure-kt}.  Every choice depends only on the fixed game and thresholds, not on the strategy or its dimension.
\end{proof}

\noindent The theorem derives a receiver-chosen-selector bound from the threshold theorem of Kundu and Tan~\cite{kundu_tan}, using the one-copy bound of Fu and Miller~\cite{FM-BRE} as its game-specific input. Fair completion of unselected guesses and the weighted score connect the selector event to the imported product-anchored game.

\subsection{Min-entropy of the hidden raw string in one block}
\label{app:source-analysis}

We now derive entropy for one hidden column string from the game. Tests, audit disclosures, and leakage are charged in the protocol proofs (\Cref{tab:game-correspondence}). The global payload proof uses the same row-reconstruction argument and density bound.

A useful special case concerns one hidden column string. For a strategy of the game with a nonempty frozen selector $J$, define
\begin{equation}
 N_b=|\{j\in J:Y_j=b\}|,\qquad
 H=\arg\min_{b\in\{0,1\}}N_b,
 \label{eq:transcript-hidden-index}
\end{equation}
breaking ties in favor of zero, and put
\begin{equation}
 K=\{j\in J:Y_j\ne H\},\qquad
 |K|\geq|J|/2.
 \label{eq:transcript-selector-K}
\end{equation}
Both $H$ and $K$ are functions of Dave's early record $(Y,B,J)$. For a selector containing only question class $\beta$, the other branch is $H=1-\beta$. The global payload proof in \Cref{lem:ordered-block-hidden-pads} uses the same row-reconstruction rule with a separate complementary branch in each microblock.

\begin{lemma}[Hidden raw string from the frozen record]
\label{lem:raw-entropy-J}
Fix $\delta_*\in(0,1)$ and constants $0<\rho_{\min}\leq\rho_{\max}\leq1$. Let $\mathcal S$ be a strategy of \Cref{thm:restricted-order-repetition} with frozen $J$ satisfying $\rho_{\min}n\leq|J|\leq\rho_{\max}n$. Put
\[
 Z_H^{\rm raw}=(\Araw_j[H])_{j\in J},\quad
 Z_H^{\rm can}=(C_j\Araw_j[H])_{j\in J},\quad
 \mathsf{Good}=\Bigl\{\sum_{j=1}^nL_j\leq\eta n\Bigr\}.
\]
There exist $\eta,c>0$ and $n_0$, depending only on the fixed parameters, such that for either $Z_H\in\{Z_H^{\rm raw},Z_H^{\rm can}\}$, every late map $\mathcal G$, which receives Dave's entire residual state and the late questions, and every $n\geq n_0$,
\begin{equation}
 \Pr_{\mathcal S,\mathcal G}
 [\mathsf{Good}\land\mathcal G\text{ outputs }Z_H]\leq2^{-cn}.
 \label{eq:hidden-guess-J}
\end{equation}
\end{lemma}

For the game's odd-column alphabet, replace an analysis record $B_j^*$ by $\overline B_j=(B_j^*[0],B_j^*[1],1\oplus B_j^*[0]\oplus B_j^*[1])$, using $001$ for malformed records. Every original Magic-Square-winning coordinate already has odd column parity, so $\overline B_j=B_j^*$ there. The replacement cannot increase loss and preserves every successful reconstruction coordinate. Actual committed records and opening checks are unchanged.

\begin{lemma}[Constant-density bucketing]
\label{lem:density-bucketing}
Fix $0<r_{\min}\leq r_{\max}\leq1$ and $s_*=(1+\delta_*)/2$. For sufficiently small constants $h,\eta>0$, cover $[r_{\min}n,r_{\max}n]$ by finitely many intervals with lower endpoints $r_jn$ and upper endpoints $\rho_j n$, including the right endpoint of the last interval, and choose
\[
 s_*\rho_j<g_j<r_j-\eta,\qquad
 \rho_j=\min\{r_j+h,r_{\max}\}.
\]
If $\eta<\min_j\eta_{\rm sel}(\rho_j,g_j)$, then $\exists c>0,n_0$ such that every admissible strategy with a pre-late-question selector $K$ satisfies, for $n\geq n_0$,
\[
 \Pr\!\left[r_{\min}n\leq|K|\leq r_{\max}n,\quad
             \sum_{j=1}^nL_j\leq\eta n,\quad
             \sum_{j\in K}W_j\geq|K|-\eta n\right]
       \leq2^{-cn}.
\]
Here $L_j,W_j$ and $\eta_{\rm sel}$ are as in \Cref{thm:restricted-order-repetition}.
\end{lemma}
\begin{proof}
First choose $h>0$ with $(1-s_*)r_{\min}>s_*h$, and fix $g_j$ strictly between $s_*\rho_j$ and $r_j$. Apply \Cref{thm:restricted-order-repetition} to these finitely many fixed pairs $(\rho_j,g_j)$. Then choose $\eta>0$ below both $\min_j(r_j-g_j)$ and all resulting loss constants. Taking the minimum exponent and a finite union bound proves the claim. All choices precede $n$ and the adversarial state.
\end{proof}

In the protocols, Dave's question string is the effective question string, defined as follows. Let $\mathsf{trit}(00)=0$, $\mathsf{trit}(01)=1$, $\mathsf{trit}(10)=2$, and $\mathsf{trit}(11)=0$ for analysis. Decode the original committed two-bit offset at position $j$ as $U_j^*$ using this map, and put $Y_j^*=V_j+U_j^*$. The original two-bit record is retained, and actual projection and syntax checks still reject the encoding $11$ wherever it is opened. With $Y=Y^*$ and $B=\overline B$, \eqref{eq:transcript-hidden-index} and \eqref{eq:transcript-selector-K} define $H$ and $K$ from the immutable early record.

\begin{proof}[Proof of \Cref{lem:raw-entropy-J}]
For either target, a correct guess supplies the same raw answer bit on every active coordinate. On an active non-losing coordinate $j\in K$, Dave's fixed answer supplies $A_j[Y_j]=B_j[X_j]$. A correct guess for $Z_H$ supplies the different entry $A_j[H]$. If $Y_j=1-H$, both independent row bits are known; if $Y_j=2$, the prescribed row parity supplies the remaining bit. Thus a successful guess yields a legal late answer for $A_j[0,1]$ on every such coordinate. Anchors score automatically. On $\mathsf{Good}$ the induced selected-win count is at least $|K|-\eta n$.

Because $K$ is determined by the frozen $(Y,B,J)$ before the late questions, it is an admissible selector. Its density is between $\rho_{\min}/2$ and $\rho_{\max}$. Apply \Cref{lem:density-bucketing} and shrink $c>0$ to absorb the fixed number of buckets. This gives \eqref{eq:hidden-guess-J} for every late map, including a map acting on arbitrary permitted entangled side information.
\end{proof}

\section*{Acknowledgements}
\addcontentsline{toc}{section}{Acknowledgements}
This project is supported by the National Research Foundation, Singapore through the National Quantum Office, hosted in A*STAR, under its Centre for Quantum Technologies Funding Initiative (S24Q2D0009) and its Advanced Quantum Algorithms and Solutions Funding Initiative (S25Q9DA001 and S25Q9DA002).

\section*{AI use disclosure}
\addcontentsline{toc}{section}{AI use disclosure}
We used OpenAI's ChatGPT interactively to assist with literature searches, exploring proof strategies, examining commitment assumptions and security simulations, identifying potential gaps and inconsistencies, drafting and revising the manuscript, and preparing the LaTeX source.
Its substantive technical assistance covered the string-commitment construction, the masked commit-and-prove protocol and its security simulations, the code-masking and global Fourier arguments.
We also used Anthropic's Claude interactively for manuscript restructuring, editing, and drafting.
The authors retain full responsibility for all mathematical claims, proofs, references, and the final manuscript.

\phantomsection
\addcontentsline{toc}{section}{References}
\begingroup
\Urlmuskip=0mu plus 2mu\relax
\bibliographystyle{alphaurl}
\bibliography{ref}
\endgroup
\end{document}